\documentclass{article}

\usepackage{multirow, makecell}
\usepackage{natbib}
\usepackage{amsmath, amsfonts, bm}
\usepackage{amsthm}
\usepackage{booktabs}

\newtheorem{proposition}{Proposition}
\newtheorem{assump}{Assumption}

\newtheorem{thm}{Theorem}

\newtheorem{lemma}{Lemma}

\usepackage{setspace}
\usepackage{comment}
\usepackage{graphicx}
\usepackage[colorlinks=true, allcolors=blue]{hyperref}

\usepackage{xcolor}
\definecolor{changed}{rgb}{0,0,0}
\definecolor{changed2}{rgb}{0,0,0}

\usepackage{tikz}
\usepackage{tikz-qtree}
\usetikzlibrary{trees}
\usetikzlibrary{automata,positioning}
\usetikzlibrary{shapes,decorations,arrows,calc,arrows.meta,fit}

\newcommand{\E}{\mathbb{E}}

\allowdisplaybreaks

\usepackage[margin=1in]{geometry}

\title{Causal Vaccine Effects on Post-infection Outcomes \\ in the Naturally Infected}

\author{Allison Codi$^{1}$, Elizabeth Rogawski McQuade$^{2}$, Razieh Nabi$^{1}$,\\ 
Mats Stensrud$^{3}$, Kaeum Choi$^1$, David Benkeser$^{1}$\\
$^{1}$Department of Biostatistics and Bioinformatics, Emory University, \\ Atlanta, GA, USA \\
$^{2}$Department of Epidemiology, Emory University, Atlanta, GA, USA \\
$^{3}$Institute of Mathematics, Ecole Polytechnique Fédérale de Lausanne \\
Laussane, Switzerland
}

\begin{document}

\onehalfspacing

\maketitle

\begin{abstract}
Understanding vaccine effects on post-infection outcomes is critical for evaluating the full value proposition of a vaccine. However, defining appropriate causal effects on such outcomes is challenging because infection is affected by vaccination. Existing approaches focus on the \emph{Doomed} stratum, individuals who would be infected regardless of vaccine receipt. For many relevant outcomes, however, this estimand understates vaccine benefit by excluding individuals whose adverse post-infection outcomes improve because vaccination prevented infection. We therefore propose causal estimands for post-infection outcomes in the \emph{Naturally Infected}, individuals who would be infected in absence of vaccine. We derive bounds under minimal assumptions and give point identification results under an exclusion restriction and/or a partial principal ignorability assumption. For point-identified settings, we develop efficient one-step estimators with robustness properties under inconsistent nuisance parameter estimation. We further show under what conditions the same identification functional can be interpreted as targeting an effect among individuals exposed to a sufficiently infectious dose of the pathogen, thereby avoiding reliance on cross-world parameters and fundamentally untestable causal assumptions. Simulations show that the bounds are valid but often wide, and that the point estimators perform well when their identifying assumptions hold. In a reanalysis of a rotavirus vaccine trial, marginal and Doomed-stratum analyses showed little evidence of an effect on antibiotic use, whereas analyses targeting the Naturally Infected suggested a protective effect under principal ignorability-based assumptions.
\end{abstract}

\maketitle

\section{Introduction}

Vaccines primarily reduce the burden of infectious diseases by preventing infections entirely and/or lessening the severity of disease following an infection. Vaccines can further prevent or improve sequelae of infections and reduce the likelihood of onward transmission of a pathogen in infected individuals. To help establish the full value proposition of a vaccine, it is important to appropriately quantify vaccine effects on each of these outcomes.

In this work, we refer to the primary endpoint of interest as an \emph{infection}; our results also apply in settings where the primary endpoint is clinical disease caused by infection. While effects on infection are straightforward to characterize, for endpoints that occur \emph{after} infection, quantifying vaccine effects can be more challenging. A comparison of post-infection outcomes between infected vaccinated and infected unvaccinated individuals will not generally represent an appropriate causal effect of vaccines due to selection bias: individuals infected after receiving a vaccine may differ from those infected in absence of the vaccine. 

A common solution has been to consider vaccine effects in the principal stratum of individuals who would be infected irrespective of whether they received vaccine, referred to as the \emph{Doomed}  stratum \citep{hudgens2006causal}. Such analyses often report estimates of the bounds on vaccine effects in this stratum \citep{hudgens2006causal, zhou2016bayesian} and have been employed  in primary and secondary evaluation of vaccines. \citet{mehrotra2006comparison} compared various statistical tests for primary endpoints of clinical trials that combine infection and post-infection endpoints, where the post-infection endpoints are evaluated in the Doomed principal stratum. Similar methods were used to evaluate varicella vaccines to prevent and reduce the severity of herpes zoster infections \citep{oxman2005vaccine}.

Here, we argue that for many relevant post-infection endpoints, a comparison of outcomes in the Doomed stratum likely understates the vaccine's benefit. By construction, this comparison excludes individuals whose post-infection outcomes were improved because the vaccine prevented infection. These individuals are precisely those who we expect to benefit most with respect to post-infection outcomes. As a result, the estimand omits an important component of the vaccine's effect, leading to an underestimation of its desirable impact. 

Alternatively, a comparison of the average post-infection outcome between all vaccinated and unvaccinated individuals is also unlikely to well characterize vaccine benefit. The population-level difference in outcomes is often small, as only a limited proportion of participants are likely to experience an infection during the study period even without vaccine. Assuming only participants who would have been infected without vaccine are likely to benefit from a vaccine effect on their post-infection outcomes, the vast majority of the trial participants who remain uninfected consequently show little or no difference in their post-infection outcome. This dilution effect, where the benefit among the relatively few is overshadowed by the lack of significant impact in the majority uninfected population, reduces the statistical power to detect a meaningful improvement in post-infection outcomes attributable to the vaccine.

We propose instead a different estimand for quantifying vaccine effects on post-infection endpoints: the effect of vaccine in the principal strata who would be infected in the absence of vaccine, a group of individuals that we term the \emph{Naturally Infected}. This group includes the Doomed stratum in addition to the \textit{Protected} stratum, individuals for whom the vaccine prevents infection. We discuss various assumptions that can be used to identify these effects and sensitivity analyses to these assumptions. Finally, we discuss a different estimand in the context of infectious diseases,  that overcomes concerns that have been raised about the unobservable nature of such estimands in the literature \citep{stensrud2023conditional}. This allows us to overcome the reliance of principal stratum estimands on fundamentally untestable assumptions via the introduction of a \emph{necessary and sufficient exposure} to a pathogen and show that principal strata effects often align with effects in the subgroup of individuals who naturally encounter such exposure events. 

\section{Background}\label{sec:background}


We consider the setting of a randomized controlled vaccine trial where the observed data consist of $X$, a vector of information measured on trial participants at the time of enrollment, $Z$, a binary indicator of vaccine assignment (=1 if vaccine; =0 if placebo or control vaccine), $S$, a binary indicator of experiencing an infection or clinical disease caused by the pathogen of interest during the trial, and $Y$, a post-infection outcome measured at a fixed time following infection or at the end of the study follow-up period. 

We also consider counterfactual outcomes for the $n$ independent trial participants. Let $\bm{Z} = (Z_1, \dots, Z_n)$ and $\bm{S} = (S_1, \dots, S_n)$ denote the vaccine assignment vector and the infection status vector for all individuals in the trial, respectively. For each individual $i$, we let $S_i(\bm{z})$ denote the infection outcome that would be observed under an intervention that sets $\bm{Z} = \bm{z}$. 

To simplify our exposition, we assume no interference and causal consistency (see \textbf{Assumptions 1-2} in Supplement A), which allows us to write potential outcomes {\color{changed} $S_i(z)$ and $Y_i(z)$ that only depend on the vaccine status of individual $i$.} These assumptions are likely reasonable in late-phase studies, where participants represent a small fraction of the at-risk population and the vaccine studied in the trial is not available  outside of the study. 
\addtocounter{assump}{2}

We make the assumption that the vaccine exhibits either a null effect or biological benefit with respect to infection in all individuals. \begin{assump}
    \emph{Monotonicity}. For all individuals $S_i(1) \le S_i(0)$.
    \label{assumption:monotonicity}
\end{assump}
Monotonicity is reasonable for vaccines that have advanced beyond pre-clinical evaluation and into large-scale trials. However, this assumption may be violated if risk behavior differs between vaccinated and unvaccinated individuals. These concerns can be mitigated through blinding  \citep{stensrud2024distinguishing}. {\color{changed}Supplement L describes methods for assessing robustness to the monotonicity assumption.}

Under Assumptions 1-3, all individuals can be categorized according to {\color{changed} the three basic principal strata shown in Table \ref{tab:basic_ps}, which includes the aforementioned Doomed and Protected strata, in addition to the Immune, those who, whether by virtue of pre-existing immunity or lack of exposure to a pathogen, are not infected irrespective of vaccination status.} 

\begin{table}
    \centering
    \caption{Basic principal stratification under no interference and monotonicity. The \emph{Naturally Infected}$^*$ are individuals who would be infected under placebo/control vaccine, $S(0) = 1$.}
    \begin{tabular}{c|c}
      Basic principal stratum   &  Potential infection outcome  \\
                                &   $(S(0), S(1))$              \\
                                \hline
      Immune                    &   (0, 0) \\
      Protected$^*$             &   (1, 0) \\
      Doomed$^*$                &   (1, 1) 
    \end{tabular}
    \label{tab:basic_ps}
\end{table}

\subsection{Evaluation of post-infection endpoints}

Many previous analyses of post-infection outcomes have tended to focus on clinical settings in which the outcome is only well-defined among individuals who become infected \citep{hudgens2006causal, mehrotra2006comparison, halloran2012causal}. In this work, we refer to such endpoints  as \emph{infection-defined}. These analyses have focused on effects in the Doomed principal stratum such as \begin{align}
\E\{ Y(1) - Y(0) \mid S(0) = 1, S(1) = 1 \} \ \mbox{or} \ \frac{\E\{ Y(1) \mid S(0) = 1, S(1) = 1 \}}{\E\{ Y(0) \mid S(0) = 1, S(1) = 1 \}} \ . \label{eq:doomed_only_estimands}
\end{align}
so that potential outcomes are well defined.

Other analyses have focused on post-infection endpoints that are only non-zero among individuals who become infected, which we refer to as \emph{infection-necessary} endpoints \citep{follmann2009chop}. 
Less attention has been given to post-infection endpoints that are well-defined and non-zero even for uninfected individuals. However, such endpoints are common in practice. We refer to them as \emph{infection-unnecessary}. For example, pediatric vaccines aimed at reducing viral diarrhea have the important secondary benefit of reducing the prescribing of antibiotics, an important benefit in controlling the emergence of antimicrobial resistance \citep{hall2022association}. Thus, suppose we are interested in quantifying the vaccine's effect on the probability a child is prescribed antibiotics over some time period. Children can be prescribed antibiotics for any number of indications that may or may not be related to the pathogen targeted by the vaccine making this an infection-unnecessary outcome. {\color{changed} Other examples of infection-unnecessary endpoints that could be affected by vaccines include all-cause endpoints such as mortality or hospitalization, health care costs, and quality of life measures.}



For infection-unnecessary endpoints, estimands such as (\ref{eq:doomed_only_estimands}) do not capture the full effect of the vaccine on the post-infection endpoint, as they omit the benefit afforded to those who had their primary infection avoided by the vaccine.


On the other hand, a marginal effect on post-infection endpoints such as $\E\{Y(1) - Y(0)\}$ may not be sensitive for detecting vaccine effects on post-infection outcomes. Marginal effects combine average post-infection outcomes in each of the three principal stratum, with each stratum weighted by its size in the population of interest, e.g., $\E\{Y(1) - Y(0) \} = \sum_{(s_0, s_1) \in \mathcal{S}} \E\{ Y(1) - Y(0) \mid S(0) = s_0, S(1) = s_1 \} P\{S(0) = s_0, S(1) = s_1\}$. 
These marginal effects may be small since (i) there may be little or no effect of the vaccine in the Immune stratum and (ii) this stratum may constitute a large fraction of the population. {\color{changed} While trials often target high-risk populations for recruitment in an effort to reduce the proportion of Immune individuals in the study population, it is generally not possible to identify and exclude all such individuals prior to enrollment.}
This implies that marginal effects estimands on post-infection endpoints will often be unduly influenced by the Immune stratum and thus assume values close to null, leading to a lack of power to detect effects on post-infection outcomes, even for highly biologically effective vaccines.

A solution for directly quantifying a vaccine's effect on a post-infection endpoint is to consider an estimand that excludes the Immune stratum and instead focus on the other two principal strata. We term the union of the Protected and Doomed principal strata the \emph{Naturally Infected} since individuals in this group would be infected in absence of the vaccine. We propose to study estimands in the Naturally Infected principal stratum, such as 
\begin{align}
    \E\{Y(1) - Y(0) \mid S(0) = 1\} \ \mbox{or} \ \E\{Y(1) \mid S(0) = 1\}/\E\{Y(0) \mid S(0) = 1\} \ . 
    \label{eq:naturally_infected_estimands}
\end{align}%
We expect such effects will be more sensitive for detecting effects on post-infection outcomes when compared to population-level vaccine effects. Moreover, we anticipate they will also be more sensitive than vaccine effects only in the Doomed stratum (\ref{eq:doomed_only_estimands}), as they incorporate the (potentially large) positive benefit of vaccination whereby infection is avoided entirely. 

In Supplement B, we compare our work to related works on principal strata effects.

\section{Identification of effects in the Naturally Infected}

We assume access to $n$ independent copies the random variable $O = (X, Z, S, Y) \sim P$. {\color{changed} We focus on controlled trials with (possibly stratified) randomization, where the following assumptions are generally satisfied by design.
\begin{assump}
    \emph{Vaccine randomization.} We assume that $\{ Y(z), S(z) \} \perp Z \mid X$.
\end{assump}
}

\begin{assump}
    \emph{Positivity.} For $z = 0, 1$, $0 < P(Z = z \mid X = x)< 1$ and $0 < P(S = 1 \mid Z = z, X = x) <  1$ almost surely.
\end{assump}

{\color{changed} We assume a model $\mathcal{M}$ for $P$ consisting of distributions satisfying positivity and monotonicity (i.e., $P(S = 1 \mid Z = 1, X = x) \le P(S = 1 \mid Z = 0, X = x)$ almost surely).} To describe our identification results, it is helpful to introduce notation for key \emph{nuisance parameters}. We define $\mu_{zs}(x) = \E( Y \mid Z = z, S = s, X = x)$, $\rho_z(x) = P(S = 1 \mid Z = z, X = x)$, $\bar{\rho}_z = \E\{ \rho_z(X) \}$, and {\color{changed}$\mu_{z\cdot}(x) = \E\{ Y \mid Z = z, X = x \}$ and $\bar{\mu}_{z\cdot} = \E\{ \mu_{z\cdot}(X) \}$.}

\subsection{Partial identification of effects}

{\color{changed} Some components of Naturally Infected effects \eqref{eq:naturally_infected_estimands} are identified under Assumptions 1-5}.

\begin{thm}
    Under \textcolor{changed}{Assumptions 1, 2, and 4-5}, $\E\{ Y(0) \mid S(0) = 1\}$ is identified by $ \psi_0$ where $\psi_0 = \E \left[ \{ \rho_0(X) / \bar{\rho}_0 \} \mu_{01}(X) \right]$. 
    \label{thm:id_EY0}
\end{thm}
An expression of $\psi_0$ using inverse probability weighting is in \textcolor{changed}{Supplement C} and proof of the theorem in Supplement K.1. We note that the second formulation of the identifying parameter in Theorem \ref{thm:id_EY0} incorporates covariates $X$, which is useful for describing covariate-adjusted estimators later.

Identification of $\E\{ Y(1) \mid S(0) = 1\}$ is more challenging owing to the cross-world nature of the parameter. Without further assumptions, it is not possible to point identify this quantity. The challenge in identification is clarified by noting that \begin{equation}
\begin{aligned}
\E\{ Y(1) \mid S(0) = 1\} &= \E\{ Y(1) \mid S(0) = 1, S(1) = 1\} P\{S(1) = 1 \mid S(0) = 1\} \ + \\
&\hspace{2em}\E\{ Y(1) \mid S(0) = 1, S(1) = 0\} [ 1 - P\{S(1) = 1 \mid S(0) = 1\} ] \ ,
\end{aligned}
\label{eq:nat_inf_mean_under_vax} 
\end{equation}
indicating that to identify $\E\{ Y(1) \mid S(0) = 1\}$ requires identifying (i) the average post-infection outcome under vaccine in the Doomed stratum; (ii) the relative fraction of the Naturally Infected that are Doomed versus Protected; and (iii) the average post-infection outcome under vaccine in the Protected stratum. Components (i) and (ii) are identifiable without further assumption.
\begin{thm}
Under Assumptions 1-5, $\E\{ Y(1) \mid S(0) = 1, S(1) = 1\} = {\color{changed} \E\{ \mu_{11}(X) \rho_1(X) \} / \bar{\rho}_1}$. We also have that $P\{ S(0) = 1, S(1) = 1 \} = \bar{\rho}_1$, $P\{ S(0) = 0, S(1) = 0 \} = 1 - \bar{\rho}_0$, and $P\{ S(0) = 1, S(1) = 0 \} = \bar{\rho}_0 - \bar{\rho_1}$ and thus, $P\{ S(1) = 1 \mid S(0) = 1 \} = \bar{\rho}_1 / \bar{\rho}_0$.
\label{thm:partial_id}
\end{thm}
See Supplement K.2 for proof and intuition. Thus, under Assumptions 1-5, \begin{equation}
\E\{Y(1) \mid S(0) = 1\} = {\color{changed} \frac{1}{\bar{\rho}_0} \left[ \E\{ \mu_{11}(X) \rho_1(X) \}  + \E\{ Y(1) \mid S(0) = 1, S(1) = 0 \} ( \bar{\rho}_0 - \bar{\rho}_1 ) \right] } \ , \label{eqn:partial_id_result}
\end{equation}
indicating that the only component remaining to identify is the average post-infection outcome under vaccine in the Protected stratum. This quantity cannot be identified without further assumptions, though it can be bounded.

\subsection{Identification of bounds}

We consider bounds under simple randomization ($Z \perp X$) for a continuous-valued post-infection outcome, such that ties are not possible. The extension allowing ties is included in Supplement D.1. To derive bounds on the average of $Y(1)$ in the Protected, it is helpful to consider the observed uninfected vaccinated participants are a mixture of Immune and Protected individuals {\color{changed}(Supplement K.2, Figure S2)}. Without further assumptions we do not have any additional knowledge as to which of the vaccinated uninfected individuals are Protected versus Immune. Nevertheless, we can bound the average post-infection outcome under vaccine in the Protected by considering truncated means at the extremes of the distribution of $Y$ in the vaccine uninfected individuals. Specifically, {\color{changed} under monotonicity} the size of the Protected stratum is identified by $\bar{\rho}_0 - \bar{\rho}_1$ and thus the size of the Protected stratum as a fraction of the vaccinated uninfected participants is given by $q = (\bar{\rho}_0 - \bar{\rho}_1) / (1 - \bar{\rho}_1)$. Define $Y_{\ell}$ and $Y_u$ as respectively, the $q$-th and $(1-q)$-th quantiles of the distribution of $Y \mid Z = 1, S = 0$. Let $\bar{\mu}_{10, \ell} = \E(Y \mid Z = 1, S = 0, Y < Y_{\ell})$, $\bar{\mu}_{10, u} = \E(Y \mid Z = 1, S = 0, Y > Y_{u})$, and $\bar{\mu}_{11} = \E\{ Y \mid Z = 1, S = 1 \}$.

\begin{thm}
Under Assumptions 1-5, \[
\E(Y \mid Z = 1, S = 0, Y < Y_{\ell}) \le \E\{ Y(1) \mid S(0) = 1, S(1) = 0\} \le \E(Y \mid Z = 1, S = 0, Y > Y_{u}) \ ,
\]
and thus, $\psi_{1,\ell} \le \E\{ Y(1) \mid S(0) = 1 \} \le \psi_{1,u}$, where \begin{align*}
    \psi_{1,\ell} = \bar{\mu}_{11} \frac{\bar{\rho}_1}{\bar{\rho}_0}
    + \bar{\mu}_{10, \ell} \left( 1 - \frac{\bar{\rho}_1}{\bar{\rho}_0} \right) \ \ , \ \mbox{and} \ \ 
    \psi_{1,u} = \bar{\mu}_{11} \frac{\bar{\rho}_1}{\bar{\rho}_0}
    + \bar{\mu}_{10, u} \left( 1 - \frac{\bar{\rho}_1}{\bar{\rho}_0} \right) \ . 
\end{align*}
\label{thm:bounds}
\end{thm}

For a proof, see Supplement K.3. A relevant special case of Theorem \ref{thm:bounds} is when the post-infection outcome is infection-necessary, as in this case $\bar{\mu}_{10} = 0$ and thus the average post-infection outcome under vaccine in the Naturally Infected is point identified, $\E\{ Y(1) \mid S(0) = 1 \} = \bar{\mu}_{11} \bar{\rho}_1/\bar{\rho}_0$. This result establishes a link between Naturally Infected effects and the \emph{chop lump test} proposed for studying the effect of vaccines on post-infection outcomes \citep{follmann2009chop}. It is straightforward to show that under monotonicity, {\color{changed} the test statistic used in that approach reduces exactly to a plug-in estimate of the additive Naturally Infected effect.}
Thus, our proposal for bounds not only gives a new causal interpretation to this existing test, but also appropriately generalizes the procedure to post-infection outcomes that can be non-zero in uninfected individuals.

The inequality in Theorem \ref{thm:bounds} also holds in any particular covariate stratum, which motivates covariate-adjusted bounds. Such bounds can sometimes be sharper than unadjusted bounds \citep{long2013sharpening}. In Supplement D.2, we propose adjusted bounds and explore the conditions under which bounds are sharpened using covariates.

\subsection{Identification using an exclusion restriction}

We have shown that point identification of Naturally Infected effects is possible in the special case of an infection-necessary endpoint. Such endpoints are specific examples of a broader class of endpoints that satisfy an \emph{exclusion restriction} \citep{angrist1996identification}. We can also make an exclusion restriction assumption for infection-unnecessary endpoints that allows point identification in those settings.
\begin{assump}
{\color{changed} $Y$ satisfies a stochastic exclusion restriction \citep{nordland2024estimation} with respect to $S$ such that $\E\{Y(1) \mid S(0) = 0\} = \E\{ Y(0) \mid S(0) = 0\}$.}
\label{assmp:exclusion_restriction}
\end{assump}

{\color{changed}The exclusion restriction assumption stipulates that the vaccine cannot have an average causal effect on post-infection outcomes in the Immune stratum.} This is likely to be reasonable in settings where the only mechanism by which vaccine affects the post-infection outcome is through preventing infection or reducing its severity.

\begin{thm}
Under Assumptions 1-6,
\begin{align*}
    \E\{ Y(1) \mid S(0) = 1, S(1) = 0 \} = \frac{
    {\color{changed}\E\{ \mu_{10}(X) (1 - \rho_1(X) \}  - \E[ \mu_{00}(X) \{ 1 - \rho_0(X) \} ]} }{\bar{\rho}_0 - \bar{\rho}_1 } \ ,
\end{align*}
and thus we have that $\E\{Y(1) \mid S(0) = 1 \}$ is identified by $\psi_{1,\text{ER}}$, where \begin{align*}
 \psi_{1,\text{ER}} &= 
 \frac{\bar{\mu}_{1\cdot} - {\color{changed} \E[ \mu_{00}(X) \{ 1 - \rho_0(X)\}]}}{\bar{\rho}_0} \ . 
\end{align*}
\label{thm:point_ID}
\end{thm}
Proof of and intuition for this result are included in Supplement K.4. 

\subsection{Identification using partial principal ignorability}

An alternative approach to identification is to use a form of partial principal ignorability. 

\begin{assump}
\emph{Partial principal ignorability}: {\color{changed} $\E\{ Y(1) \mid S(0) = 0, S(1) = 0, X \} = \E\{ Y(1) \mid S(0) = 1, S(1) = 0, X \} = \E\{ Y(1) \mid S(1) = 0, X \}$.}
\label{assmp:weak_principal_ignorability}
\end{assump}

Assumption 7 stipulates that after adjustment for $X$, there is no difference between individuals in the Immune versus Protected principal strata in terms of their post-infection outcomes. This cross-world assumption would be satisfied if $X$ included all common causes of the infection endpoint under placebo and the post-infection outcome under vaccine.

\begin{thm}
    Under Assumptions 1-5 and 7-8, the $X$-conditional mean in the Protected principal stratum is identified as $\E\{ Y(1) \mid S(0) = 1, S(1) = 0, X \} =  \mu_{10}(X)$, and thus $\E\{Y(1) \mid S(0) = 1 \}$ is identified by $\psi_{1,\text{PI}}$ where
\begin{align}
 \psi_{1,\text{PI}} &= \E\bigg( \frac{\rho_0(X)}{\bar{\rho}_0} \left[ \mu_{11}(X) \frac{\rho_1(X)}{\rho_0(X)} +  \mu_{10}(X) \left\{ 1 - \frac{\rho_1(X)}{\rho_0(X)} \right\} \right] \bigg) \ . \label{eqn:id_psi1}
\end{align}
\label{thm:id_EY1}
\end{thm}
Supplement K.5 includes a proof and comparison to other forms of principal ignorability. Supplement E describes sensitivity analysis for assessing robustness to violations of partial principal ignorability.Supplement F describes specific design considerations that impact the plausibility of exclusion restrictions and partial principal ignorability. 

\section{Estimation}

\subsection{Bounds}

To estimate the bounds, we compute estimates $\bar{\rho}_{z,n} = \sum_{i=1}^n S_i I(Z_i = z) / \sum_{i=1}^n I(Z_i = z)$ of $\bar{\rho}_z$, for $z = 0, 1$ and estimate $\bar{\mu}_{11,n} = \sum_{i=1}^n Y_i S_i Z_i / \sum_{i=1}^n S_i Z_i$. We then compute $q_n = (\bar{\rho}_{0,n} - \bar{\rho}_{1,n}) / (1 - \bar{\rho}_{1,n})$, which is used to calculate $Y_{\ell,n}$ and $Y_{u,n}$, the empirical $q_n$-th and $1-q_n$-th quantiles of the distribution of $Y$ given $Z = 1, S = 0$. An estimate of $\bar{\mu}_{10, \ell}$ can then be computed as $\bar{\mu}_{10, \ell, n} = \sum_{i=1}^n Y_i Z_i (1 - S_i) I(Y_i < Y_{\ell, n}) / \sum_{i=1}^n Z_i (1 - S_i) I(Y_i < Y_{\ell, n})$. The estimate $\bar{\mu}_{10, u, n} = \sum_{i=1}^n Y_i Z_i (1 - S_i) I(Y_i > Y_{u, n}) / \sum_{i=1}^n Z_i (1 - S_i) I(Y_i > Y_{u, n})$ can be similarly calculated. The final estimates of the bounds are thus $\ell_n = \bar{\mu}_{11,n} \bar{\rho}_{1,n}/\bar{\rho}_{0,n} + \bar{\mu}_{10, \ell, n} \left( 1 -\bar{\rho}_{1,n}/\bar{\rho}_{0,n} \right)$ and 
$u_n = \bar{\mu}_{11,n} \bar{\rho}_{1,n}/\bar{\rho}_{0,n} + \bar{\mu}_{10, u, n} \left( 1 - \bar{\rho}_{1,n}/\bar{\rho}_{0,n} \right).$
Confidence intervals for $\ell_n$ can be derived using the nonparametric bootstrap.

\subsection{Efficiency theory for point identification results}

We focus on efficient estimation of the identifying functionals $\psi_0$, $\psi_{1,\text{ER}}$, and $\psi_{1,\text{PI}}$. Although some parameters (e.g., $\psi_0$) can be identified without adjustment for covariates $X$, estimators that ignore covariates are generally inefficient \citep{benkeser2021improving}. We focus on one-step estimators that leverage these covariates to achieve efficiency. Not only are these estimators efficient, they are also doubly/multiply robust, as we highlight in our results below. Singly robust estimators are described in the Supplement C.

A key step in deriving one-step estimators is to derive {\color{changed} the efficient influence function} (EIF) of the identifying functional. The EIF can be used to debias plug-in estimators, thereby achieve asymptotic efficiency and often robustness \citep{pfanzagl1982contributions}. 
Thus, for each point identification result, we describe how to construct a plug-in estimator and the form of the EIF for the parameter of interest, thereby enabling construction of an efficient one-step estimator. Results pertaining to our estimators' large sample behavior are stated here with details on regularity conditions included throughout Supplement K. Briefly, our estimators require estimators of certain \emph{nuisance parameters}. {\color{changed} Several nuisance parameters in our problem are conditional means and conditional probabilities. Our theory supports estimation of these quantities using any suitable regression technique. For example, estimates could be based on working parametric models, semiparametric methods, or machine learning approaches. Our asymptotic results require that these regression estimates satisfy certain boundedness and convergence conditions (e.g., Supplement K.7.2)}. 

\subsection{Estimation of Naturally Infected outcomes under placebo}

A plug-in estimator of $\psi_0$ can be generated by first estimating $\mu_{01}$, e.g., by fitting a regression of $Y$ on $X$ in the subset of data with $Z = 0, S = 1$, resulting in estimate {\color{changed}$\mu_{01,n}(x)$ for each $x$}. Next, a regression of $S$ onto $X$ is fit in the subset of data with $Z = 0$ to generate an estimate {\color{changed} $\rho_{0,n}(x)$ of $\rho_0(x)$ for each $x$, which is used to estimate $\bar{\rho}_0$, by defining $\bar{\rho}_{0,n} = n^{-1} \sum_{i=1}^n \rho_{0,n}(X_i)$}. The empirical distribution of $X$ is used to estimate the distribution of $X$, resulting in plug-in estimator $\psi_{0,n} = n^{-1} \sum_{i=1}^n \{ \rho_{0,n}(X_i) / \bar{\rho}_{0,n} \} \mu_{01,n}(X_i).$

Next, we introduce the EIF used to define the one-step estimator. Define $\pi_z(x) = P(Z = z \mid X = x)$ and $\tilde{\psi}_0(x) = \{ \rho_0(x) / \bar{\rho}_0 \} \mu_{01}(x)$. 
\begin{thm}
The EIF of $\psi_0$ in the model $\mathcal{M}$ is $\Phi_0$, where \begin{align*}
    \Phi_{0}(O_i) &= \frac{(1 - Z_i)}{\pi_0(X_i)} \frac{S_i}{\bar{\rho}_0} \{ Y_i - \mu_{01}(X_i) \} + \frac{\{\mu_{01}(X_i) - \psi_0\}}{\bar{\rho}_0}  \frac{(1 - Z_i)}{\pi_0(X_i)} \{ S_i - \rho_0(X_i) \} \\
& \hspace{2em} - \frac{\psi_0}{\bar{\rho}_0}\{ \rho_0(X_i) - \bar{\rho}_0 \} + \tilde{\psi}_0(X_i) - \psi_0 \ .
\end{align*}
\end{thm}

An estimate of the EIF can be computed by plugging in unknown quantities. For this, in addition to estimates $\rho_{0,n}$ and $\mu_{01,n}$ described above, we require an estimate $\pi_{0,n}$ of $\pi_0$. This could be the known randomization probabilities or based on a simple regression fit of $Z$ on $X$. An estimate of the EIF evaluated on a given observation $O_i$ is
\begin{equation}
\begin{aligned}
\Phi_{0,n}(O_i) &= \frac{(1 - Z_i)}{\pi_{0,n}(X_i)} \frac{S_i}{\bar{\rho}_{0,n}} \{ Y_i - \mu_{01,n}(X_i) \} + \frac{\{\mu_{01,n}(X_i) - \psi_{0,n}\}}{\bar{\rho}_{0,n}}  \frac{(1 - Z_i)}{\pi_{0,n}(X_i)} \{ S_i - \rho_{0,n}(X_i) \} \\
& \hspace{2em} - \frac{\psi_{0,n}}{\bar{\rho}_{0,n}}\{ \rho_{0,n}(X_i) - \bar{\rho}_{0,n} \} + \tilde{\psi}_{0,n}(X_i) - \psi_{0,n} \ ,
\end{aligned} \label{eqn:example_estimated_gradient}
\end{equation}
where $\tilde{\psi}_{0,n}(X_i) = \rho_{0,n}(X_i) / \bar{\rho}_{0,n} \mu_{01,n}(X_i)$. The one-step estimator is defined as $\psi_{0,n}^{+} = \psi_{0,n} + n^{-1} \sum_{i=1}^n \Phi_{0,n}(O_i)$. Under regularity conditions (Supplement {\color{changed} K.7.2}), $n^{1/2} (\psi_{0,n}^+ - \psi_0)$ converges in distribution to a mean-zero Gaussian random variable with variance $\E\{ \Phi_0(O)^2 \}$. $\psi_{0,n}^+$ is doubly robust in that it is consistent if either $\pi_{0,n}$ is consistent for $\pi_0$ or if both $\rho_{0,n}$ and $\mu_{01,n}$ are consistent for their respective targets. 

\subsection{Estimation of Naturally Infected outcomes under vaccine}

\subsubsection{Estimation under exclusion restriction}

A plug-in estimate of $\psi_{1, \text{ER}}$ can be computed based on estimates of $\bar{\mu}_{1\cdot}$, $\bar{\mu}_{00}$, and $\rho_0$. As with estimation of $\psi_0$, it is efficient to include covariates in the estimation, e.g., estimating $\mu_{1\cdot}(X) = \E(Y \mid Z = 1, X)$ by regressing $Y$ on $X$ in the subset of data with $Z = 1$. An estimate of $\bar{\mu}_{1\cdot}$ is $\bar{\mu}_{1\cdot,n} = n^{-1} \sum_{i=1}^n \mu_{1\cdot,n}(X_i)$. Similarly, a regression of $Y$ on $X$ in the subset of data with $Z = 0$ and $S = 0$ yields an estimate $\mu_{00,n}(x)$ of $\mu_{00}(x)$ for each $x$. An estimate of $\bar{\rho}_0$, as above, can be based on  marginalizing $\psi_{0,n}$ an estimated regression of $S$ on $X$ in the subset of data with $Z = 0$. The plug-in estimate is {\color{changed} $\psi_{1,\text{ER},n} = [\bar{\mu}_{1\cdot,n} - \sum_{i=1}^n [\mu_{00,n}(X_i) \{ 1 - \rho_{0,n}(X_i) \}] /  \bar{\rho}_{0,n}$}. Let $\tilde{\psi}_{1,\text{ER}}(x) = [ \mu_{1\cdot}(x) - \mu_{00}(x)\{ 1 - \rho_0(x) \}] / \bar{\rho}_0$ \ .

\begin{thm}
    The EIF for $\psi_{1,\text{ER}}$ in the model $\mathcal{M}$ is \begin{align*}
        \Phi_{1,\text{ER}}(O_i) &= 
          \frac{1}{\bar{\rho}_0} \frac{Z_i}{\pi_1(X_i)} \{ Y_i - \mu_{1\cdot}(X)\} - \frac{(1 - Z_i) (1 - S_i)}{\pi_0(X_i) \bar{\rho}_0} \{ Y_i - \mu_{00}(X_i) \} \\
          &\hspace{-3em} + \frac{\{\mu_{00}(X_i) - \psi_{1,\text{ER}}\}}{\bar{\rho}_0} \frac{(1 - Z_i)}{\pi_0(X_i)} \{ S_i - \rho_0(X_i) \} - \frac{\psi_{1,\text{ER}}}{\bar{\rho}_0} \{ \rho_0(X_i) - \bar{\rho}_0 \} + \tilde{\psi}_{1,\text{ER}}(X_i) - \psi_{1,\text{ER}}  \ . 
    \end{align*}
\end{thm}

An estimate $\Phi_{1, \text{ER}, n}$ of the EIF could be computed by plugging in estimates of nuisance parameters as in equation (\ref{eqn:example_estimated_gradient}) and a one-step estimator similarly defined as $\psi_{1,\text{ER},n}^+ = \psi_{1,\text{ER},n} + n^{-1} \sum_{i=1}^n \Phi_{1, \text{ER}, n}(O_i)$. Under regularity conditions (redSupplement K.8.2), $n^{1/2} (\psi_{1,\text{ER}, n}^+ - \psi_{1,\text{ER}})$ converges in distribution to a mean-zero Gaussian random variable with variance $\E\{ \Phi_{1,\text{ER}}(O)^2 \}$. $\psi_{1,\text{ER}, n}^+$ is multiply robust, with four minimal combinations of consistent nuisance parameter estimates that yield a consistent one-step estimate (Table S10, Supplement K.8.2). Notably, in a randomized trial where $\pi_0$ and $\pi_1$ are known, one-step estimators are guaranteed to be consistent irrespective of inconsistent estimation of $\rho_0$ and $\mu_{1\cdot}$.

\subsubsection{Estimation under partial principal ignorability}

To estimate $\psi_{1,\text{PI}}$, we may generate a plug-in estimator by estimating $\rho_0$ and $\bar{\rho}_0$ as described above. Estimates $\mu_{1s,n}$ of $\mu_{1s}$ for $s = 0, 1$, may be obtained by regressing $Y$ onto $X$ in the subset of data with $Z = 1$ and $S = s$. A plug-in estimator is $
\psi_{1,\text{PI},n} = n^{-1} \sum_{i=1}^n \left[ \mu_{11,n}(X_i) \rho_{1,n}(X_i)\right.$ $\left.+ \mu_{10,n}(X_i) \left\{ \rho_{0,n}(X_i) - \rho_{1,n}(X_i) \right\} \right] /\bar{\rho}_{0,n}$. 
We define $\tilde{\psi}_{1, \text{PI}}(x) = \rho_0(x) / \bar{\rho}_0 [ \mu_{11}(x) \rho_1(x) / \rho_0(x) + \mu_{10}(x) \{ 1 - \rho_1(x)/\rho_0(x)\}  ]$.

\begin{thm}
The EIF of $\psi_{1,\text{PI}}$ in the model $\mathcal{M}$ is
\begin{equation*}
\begin{aligned}
    \Phi_{1,\text{PI}}(O_i) &= \frac{Z_i}{\pi_1(X_i)} \frac{S_i}{\bar{\rho}_0} \{Y_i - \mu_{11}(X_i) \} + \frac{Z_i}{\pi_1(X_i)}\frac{(1 - S_i)}{\bar{\rho}_0} \frac{\{\rho_0(X_i) - \rho_1(X_i)\}}{\{1 - \rho_1(X_i)\} } \{Y_i - \mu_{10}(X_i) \} 
\\
&\hspace{2em} + \frac{Z_i}{\pi_1(X_i)} \frac{\{\mu_{11}(X_i) - \mu_{10}(X_i)\}}{\bar{\rho}_0}  \{ S_i - \rho_1(X_i) \} 
\\
&\hspace{3em} + \frac{(1 - Z_i)}{\pi_0(X_i)} \frac{\{\mu_{10}(X_i) - \psi_{1,\text{PI}}\}}{\bar{\rho}_0} \{S_i - \rho_0(X_i) \} - \frac{\psi_{1,\text{PI}}}{\bar{\rho}_0} \{\rho_0(X_i)  - \bar{\rho}_0 \}  \\
&\hspace{4em} + \tilde{\psi}_{1, \text{PI}}(X_i) - \psi_{1,\text{PI}} \ . 
\end{aligned}
\end{equation*}
\label{thm:eif_psi1}
\end{thm}
As above, an estimate $\Phi_{1,\text{PI},n}$ of this gradient can be computed by plugging in estimates of nuisance parameters, and a one-step estimator constructed as $\psi_{1,\text{PI},n}^+ = \psi_{1,\text{PI},n} + n^{-1} \sum_{i=1}^n \Phi_{1,\text{PI},n}(O_i)$. Under regularity conditions (Supplement K.9.2), $n^{1/2} (\psi_{1,\text{PI}, n}^+ - \psi_{1,\text{PI}})$ converges in distribution to a mean-zero Gaussian random variable with variance $\E\{ \Phi_{1,\text{PI}}(O)^2 \}$. $\psi_{1,\text{PI}, n}^+$ is multiply robust with six minimal combinations of consistent nuisance parameter estimates yielding consistent one-step estimates (Table S11, Supplement K.9.2). In contrast to $\psi_{1,\text{ER}}$, consistent estimation of $\pi_0$ and $\pi_1$ is not sufficient and thus in randomized trials (Table S12, Supplement K.9.2), $\psi_{1,\text{PI}}$ requires stronger conditions for consistent estimation than does $\psi_{1,\text{ER}}$.

We include details for performing a sensitivity analysis to the partial principal ignorability assumption in Supplement E. In Supplement H, we provide details on point identification and efficient estimation of the Doomed effects under a principal ignorability assumption.

\subsubsection{Estimation when both assumptions hold}

In the situation where \emph{both} exclusion restriction \emph{and} partial principal ignorability hold, then it is possible to more efficiently estimate Naturally Infected effects. {\color{changed}The key insight in this case is that the conditional mean of $Y(1)$ in the Protected stratum is identified by $\E\{ Y \mid S = 0, X \}$. Intuitively, under principal ignorability, we estimated this mean using only the vaccinated uninfected observations; thus, it is not surprising that there is an efficiency benefit to estimating this mean using data from \emph{all} uninfected observations.} We provide details for this estimator in Supplement G.

\section{{\color{changed2} A single-world exposure representation of Naturally Infected effects}}

Principal strata estimands such as $\E\{ Y(1) \mid S(0) = 1 \}$ involve counterfactuals defined when $Z = 1$ and when $Z = 0$. Thus, except in the special case of infection-necessary outcomes, {\color{changed2} identification requires assumptions such as the exclusion restriction and/or partial principal ignorability that are \textit{fundamentally untestable}. Such cross-world parameters have been criticized \citep{robins2010alternative,richardson2013single} motivating an }alternative parameter that does not involve cross-world quantities in its definition, nor cross-world assumptions in its identification. See Supplement K.10 for proofs for this section.

Suppose that there is a vector-valued, unmeasured variable $\tilde{S} \in \tilde{\mathcal{S}}$ denoting some unmeasured amount of exposure to the pathogen causing the infection outcome $S$ \citep{stensrud2023identification,janvin2025quantification,perenyi2025variant}. For example, $\tilde{S}$ could represent a vector of information about the dose, total duration, and/or route of exposure to a pathogen that an individual was exposed to. We assume that there is a binary coarsening $e: \tilde{\mathcal{S}} \times \mathcal{X} \rightarrow \{0,1 \}$ and define the random variable $E = e(\tilde{S}, X)$. We make the following assumptions about this exposure variable.

\begin{assump}
    \emph{Exposure is necessary and sufficient for infection in absence of vaccine}: $P(S = 1 \mid E = 0, Z = 0) = 0$ and $P(S = 1 \mid E = 1, Z = 0) = 1$.
    \label{assumption:nec_suff_exp}
\end{assump}
Thus, $E = 1$ represents the occurrence of an exposure to the pathogen such that in absence of the vaccine an individual would have $S = 1$ with probability 1, while no one with $E = 0$ would have $S = 1$ \citep{janvin2025quantification}. We allow for this infectious dose to vary by individual characteristics. For example, individuals with previous exposure to a pathogen may require a higher infectious dose than those who are na{\"i}ve to the pathogen.

\begin{assump}
    {\color{changed} \emph{Vaccine unrelated to exposure}. Vaccination does not cause exposure $E(z) = E$ for $z = 0,1$ and vaccine status is conditionally independent of exposure, $Z \perp E \mid X$.}
    \label{assumption:blinding}
\end{assump}
Assumption \ref{assumption:blinding} is plausible in an appropriately blinded randomized trial, where participants are unaware of their vaccine assignment. This may be difficult to justify in an open-label trial and in placebo-controlled trials when their are known side effects of vaccination, as individuals may adjust their risk behavior in response to knowledge of their assigned arm. However, vaccine trials are often designed with active comparator vaccines rather than a placebo vaccine (e.g., a rabies vaccine as a control for a malaria vaccine \citep{bejon2008efficacy}), such that for many trials this assumption is likely plausible.

{\color{changed2} In practice, exposure variables are often not measured; however, realistic experimental designs could collect this information using exposure monitors, mobile phones, or other means \citep{zhang2022monitoring}. The possibility of measuring exposure, at least in principle, motivates an alternative single-world formulation of the Naturally Infected effect and its identification. Under Assumptions \ref{assumption:nec_suff_exp} and \ref{assumption:blinding}, $E=1$ corresponds to $S(0)=1$, so exposure-conditional effects such as $\E\{Y(1) \mid E = 1\} - \E\{Y(0) \mid E = 1\}$ target the same population as the Naturally Infected estimand. Rather than relying on cross-world assumptions for identification, we show that the same identifying functionals can be justified by exposure-conditional versions of the exclusion restriction or partial ignorability involving jointly observable variables that could, in principle, be empirically evaluated if exposure were measured.}

We have the following identification for $\E\{Y(0) \mid E = 1\}$.

\begin{thm}
    Under Assumptions 1-5 and {\color{changed} 8-9}, $\E\{ Y(0) \mid E = 1 \} = \psi_0$.
\end{thm}

\subsection{Identification under exposure-conditional exclusion restriction {\color{changed2}or} exposure ignorability}

Suppose that instead of the typical exclusion restriction (Assumption \ref{assmp:exclusion_restriction}), which is cross-world in nature, we instead assume an exposure-conditional exclusion restriction. 
{\color{changed}
\begin{assump}
    \emph{Exposure-conditional stochastic exclusion restriction}. $\E\{ \E( Y \mid Z = 1, E = 0, X ) \mid E = 0 \} = \E\{ \E( Y \mid Z = 0, E = 0, X ) \mid E = 0 \}$
    \label{assmp:exp_cond_er}
\end{assump}
}

\begin{thm}
    Under Assumptions 1-5 and {\color{changed} 8-10}, $\E\{ Y(1) \mid E = 1 \} = \psi_{1,\text{ER}}$.
    \label{thm:exp_cond_er_id}
\end{thm}

Theorem \ref{thm:exp_cond_er_id} establishes that analyses targeting effects based on $\psi_{1,\text{ER}}$ can equivalently be interpreted as effects in the subpopulation who would naturally be exposed to an infectious dose of the pathogen of interest \citep{stensrud2023identification}. While Assumption \ref{assmp:exp_cond_er} is not testable in settings where $E$ is unmeasured, were $E$ to be measured, any test of mean independence would suffice to evaluate this assumption. 

We can also provide identification under the following single-world assumptions, which replace the cross-world assumption of partial principal ignorability. 

\begin{assump}
    \emph{Exposure is necessary for infection in presence of vaccine}: $P(S = 1 \mid E = 0, Z = 1) = 0$.
    \label{assmp:exp_necess_under_vax}
\end{assump}

\begin{assump}
    \emph{Conditional mean ignorability of exposure}: $\E( Y \mid Z = 1, E = 1, S, X ) = \E( Y \mid Z = 1, E = 0, S, X )$ = $\E( Y \mid Z = 1, S, X )$.
    \label{assmp:cond_ignor_of_exposure}
\end{assump}

Assumption \ref{assmp:exp_necess_under_vax} is likely to be satisfied in settings where Assumption 8 holds. Nevertheless, we separately state this assumption as it is not needed to prove identification under exposure-conditional exclusion restriction. Assumption \ref{assmp:cond_ignor_of_exposure} does not involve potential outcomes, and could be empirically assessed if exposure information were collected. The assumption would be violated, for example, in the presence of an unmeasured common cause of $S$ and $Y$, or if $E$ has a directed path to $Y$ that is not through $S$.

\begin{thm}
    Under Assumptions 1-5, 8-9, 11-12, $\E\{ Y(1) \mid E = 1 \} = \psi_{1, \text{PI}}$.
    \label{thm:exposure_identification}
\end{thm}


In Supplement H.3, we provide an exposure-conditional interpretation for effects in the Doomed stratum. In that case, we define an exposure $E^*$ that is sufficient for infection \emph{irrespective} of vaccine status, such that all individuals would be infected if exposed to $E^*$ (as opposed to $E$ wherein some vaccinated individuals are protected following exposure). 

\section{Simulations}

\subsection{Asymptotic properties of estimators}

We conducted a simulation study to evaluate finite-sample performance of point estimators and estimators of bounds under a range of data-generating mechanisms when causal assumptions required by each method were and were not met (Supplement I.1). We found that estimated bounds appropriately covered the true effect, but were wide. We found that all point estimators performed well when their assumptions were met and poorly when their assumptions were not. When both exclusion restriction and partial principal ignorability held, we found that the semiparametric estimator had the smallest variance, followed by the estimator of $\psi_{1,\text{PI}}$. The estimator of $\psi_{1,\text{ER}}$ tended to have higher variance.

\subsection{Comparing power of estimands in realistic setting}

We additionally studied power of hypothesis tests based on different causal estimands for detecting protective vaccine effects on a post-infection outcome. The goal of this simulation was to explore the potential benefits for using Naturally Infected effects to infer a causal effect of $Z$ on $Y$, compared to using either a marginal effect or an effect in the Doomed principal strata. The data-generating process was calibrated to resemble key features of the PROVIDE study (NCT01375647), a randomized placebo-controlled trial of an oral rotavirus vaccine conducted in Dhaka, Bangladesh from 2011–2014 \citep{colgateDelayedDosingOral2016a}. We generated simulated datasets of size $n=700$. The infection variable of interest $S$ was rotavirus infection, and the post-infection outcome of interest $Y$ was receipt of any antibiotics by week 52. Three baseline covariates $X=(X_1,X_2,X_3)$ denoting respectively gender, height-for-age Z-score, and number of household members were generated to reflect observed distributions in the PROVIDE data. Vaccine assignment $Z$ was generated independently of potential outcomes according to a Bernoulli$(0.5)$ distribution. 

Conditional on $X$, principal stratum membership and potential post-infection outcomes were simulated in such a way that allowed us to (i) satisfy monotonicity, exclusion restriction, and partial principal ignorability; (ii) mimic the distribution of rotavirus infection and antibiotic use observed in the observed data to the extent possible; and (iii) control the level of vaccine efficacy against infection and the size of vaccine effects in principal strata on post-infection outcomes. See Supplement I.2 for details. This approach allowed us to vary the extent to which the effect of $Z$ on $Y$ was driven by the composition of principal strata in the population and the size of the effect in the Doomed vs. Protected principal strata. 

We considered four different compositions of principal strata that can be defined based on vaccine efficacy (i.e., the relative amount of Protected vs. Doomed individuals) and the proportion Immune. First, we simulated a setting with modest vaccine efficacy (66\%) to prevent infection and a relatively low proportion of Immune (40\%). We then held vaccine efficacy fixed (66\%) while increasing the proportion of Immune individuals (60\%) to explore the extent to which increasing baseline immunity dilutes population-level effects. We then held Immune fixed at (40\%) while decreasing (to 50\%) and increasing (to 85\%) vaccine efficacy in order to explore effects in settings where the primary mechanism of vaccines impact is through the prevention of infection versus through improving the post-infection outcome among infected individuals.

For each of these four principal strata compositions, we varied the effect size on post-infection outcomes in the Doomed and Protected principal strata across a two-dimensional grid. For each setting, we simulated and analyzed 1000 datasets. Tests of vaccine effects were carried out using one-step estimators with relevant nuisance parameters estimated using Super Learner incorporating logistic regression, multivariate adaptive regression splines, generalized additive models, and forward stepwise regression. Power was defined as the proportion of simulated datasets wherein the null hypothesis of no effect was rejected using a two-sided level 0.05 Wald test using estimated influence-function-based standard errors. Results were summarized using contour plots highlighting regions with at least 80\% power. 

In a setting with modest vaccine efficacy to prevent infection (top row, Figure \ref{fig:placeholder}), we found that all estimators had at least some power to reject the null hypothesis of no effect of $Z$ on $Y$. However, as the size of the Immune grew with vaccine efficacy held constant (second row), we found that as expected the power to detect effects using a population-level effect disappeared. Tests based on the exclusion restriction-based Naturally Infected effects estimator were also not powered. However, in this setting the principal ignorability-based estimators maintained power to detect effects. The power to detect effects using population effect estimates and exclusion restriction-based estimates was also diminished when vaccine efficacy was reduced holding the proportion of Immune fixed (third row vs. first row), while in this setting principal-ignorability-based estimators maintained power. On the other hand, when vaccine efficacy was increased (fourth row), power improved for the population- and exclusion restriction-based tests, but was still less than the principal ignorability-based ones. 

Across all settings, power was essentially identical between the test based on the exclusion restriction-based estimator and the population estimator, despite the magnitude of the effect being larger for the Naturally Infected effect. Both had inferior power to the principal ignorability-based tests. The semiparametric estimator that assumed both the exclusion restriction and principal ignorability had improved power relative to tests based on the estimator that assumed principal ignorability alone, though the difference was modest.

\begin{figure}
    \centering
    \includegraphics[width=\linewidth]{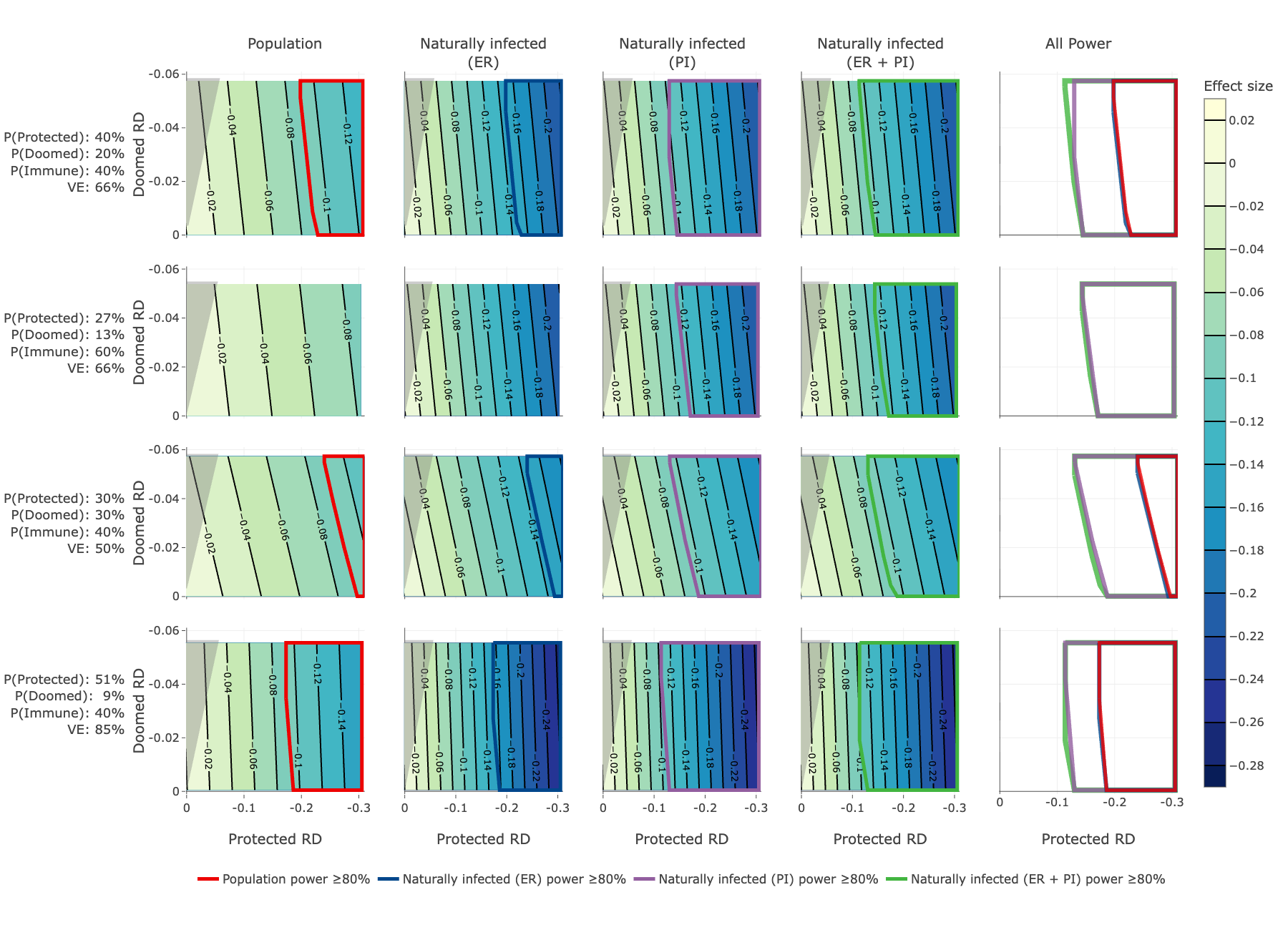}
    \caption{Power of a hypothesis test to reject the null hypothesis of no effect of $Z$ on $Y$ under different principal strata mixtures (rows) based on various effect estimators (columns) and under different principal stratum-specific effect sizes (axes of each figure). The horizontal axis is the risk difference (RD) in the Protected strata; the vertical axis is the RD in the Doomed strata. Grayed areas indicate regions where the effect in the Doomed exceeds the effect in the Protected stratum, which are unlikely in vaccine contexts. Contours indicate the size of each effect and outlined regions indicate where tests had at least 80\% power to detect the difference. The final column shows these regions for each estimator.}
    \label{fig:placeholder}
\end{figure}

\section{Data analysis}\label{sec:data}

The PROVIDE study (NCT01375647) was a randomized placebo-controlled trial of an  oral rotavirus vaccine \citep{colgateDelayedDosingOral2016a}. Seven hundred infants were randomized 1:1 to receive two doses of vaccine or placebo. Rotavirus diarrhea ($S$) was identified via twice-weekly surveillance for diarrhea using a stool rotavirus antigen enzyme immunoassay.  Our analysis considers any episodes of rotavirus diarrhea from birth to one year of age in the per protocol population. Any antibiotic use for all-cause diarrhea ($Y$) was reported by a caregiver at the time of each diarrhea episode. From the available data we included the following adjustment variables: baseline height-for-age Z-score, gender, and number of household members. 

We estimated bounds for the effect of $Z$ on $Y$ in the Naturally Infected, as well as point estimates using one-step estimators based on the exclusion restriction, partial principal ignorability, and both. We compared these estimates to one-step estimates of the marginal effect of $Z$ on $Y$, as well as the effect in the Doomed stratum. All estimates used super learning for nuisance parameter estimation, with the candidate regression library consisting of generalized linear models, generalized additive models, multivariate adaptive regression splines, and stepwise generalized linear models. 

The estimated bounds on the additive effect indicated that the effect of vaccine led from anywhere between a 42.5\% (95\% CI: -56.2\%, -26.0\%) decrease in antibiotic use for diarrhea to an 8.6\% increase (2.4\%, 15.9\%), providing no evidence of vaccine harm or benefit with respect to antibiotic use for diarrhea in the Naturally Infected. Covariate adjustment did not meaningfully impact bound width (Supplement J.1, Table S9). Similarly, there was no evidence of a vaccine effect on antibiotic use when based on the estimated marginal effect, the effect in the Doomed, nor the Naturally Infected when assuming only the exclusion restriction (Table \ref{tab:real_data_results}). On the other hand, the Naturally Infected estimators that assumed partial principal ignorability demonstrated some evidence that the vaccine had a positive effect in reducing antibiotic use for diarrhea, with an estimated 8.5\% lower absolute probability of antibiotic use among the Naturally Infected (95\% CI: 18.3\% lower to 1.4\% higher; p-value = 0.091). The estimate additionally assuming the exclusion restriction had a nearly identical point estimate and a narrower confidence interval (95\% CI: 18.1\% lower to 0.1\% higher; p-value = 0.072). {\color{changed} Sensitivity analyses for partial principal ignorability and monotonicity assumptions are reported in Supplements J.2 and L.3, respectively}.


\begin{table}[!h]
\centering
\caption{Estimate of additive and multiplicative effects of rotavirus vaccine on antibiotic use for diarrhea within the first year of life in marginal, Doomed, and Naturally Infected strata using AIPW estimators}
\begin{tabular}[t]{llcccc}
\toprule
\multicolumn{2}{c}{} & \multicolumn{2}{c}{Additive} & \multicolumn{2}{c}{Multiplicative} \\
\cmidrule(l{3pt}r{3pt}){3-4}
\cmidrule(l{3pt}r{3pt}){5-6}
Estimand & Estimator & Estimate (95\% CI) & p-value & Estimate (95\% CI) & p-value \\
\midrule
Marginal & & -0.009 (-0.084, 0.066) & 0.811 & 0.988 (0.892, 1.094) & 0.811 \\
\cmidrule{1-6}
Doomed & & 0.053 (-0.035, 0.141) & 0.239 & 1.058 (0.963, 1.164) & 0.241 \\
\cmidrule{1-6}
\multirowcell{3}{Naturally\\ infected}
    & ER      & -0.026 (-0.238, 0.186) & 0.810 & 0.971 (0.762, 1.238) & 0.813 \\
    & PI      & -0.085 (-0.183, 0.014) & 0.091 & 0.905 (0.807, 1.015) & 0.090 \\
    & PI + ER & -0.087 (-0.181, 0.008) & 0.072 & 0.903 (0.809, 1.008) & 0.069 \\
\bottomrule
\end{tabular}
\label{tab:real_data_results}
\end{table}

\section{Discussion}

Naturally infected effects characterize the effect of vaccines on post-infection endpoints. We have provided a comprehensive overview of practical estimation of Naturally Infected effects, spanning estimation of bounds, the two most common forms of assumptions for point identification, and a common form of sensitivity analysis. As with many principal effect estimands, bounds are rarely expected to be informative in practice and therefore assumptions required for point identification must be closely scrutinized. We find that for sufficiently well monitored trials, the exclusion restriction often is plausible (see Supplement F for further discussion). However, while Naturally Infected effects estimated assuming the exclusion restriction are larger than similarly estimated population-level effects, hypothesis testing-based inference is rarely different between the two approaches. Thus, the partial principal ignorability assumption will likely be needed in practical applications to estimate Naturally Infected effects. This finding has implications for other areas of biomedicine, e.g., in ``responder'' analysis, where treatment effects are characterized in the principal strata of treatment ``responders'' as indicated by having a biomarker above a certain threshold when treated \citep{nordland2024estimation}.

We also give conditions under which the same observed-data parameter that identifies the principal stratum estimand can be interpreted as an effect among individuals exposed to a sufficiently infectious dose. This estimand aligns with interventionist causal inference \citep{robins2010alternative,richardson2013single} {\color{changed2} and may allay philosophical concerns regarding interpretability of fundamentally cross-world parameters such as $\E\{ Y(1) \mid S(0) = 1 \}$}. As exposure monitoring becomes more feasible, the assumptions required for this interpretation can be tested empirically in practice.

\subsection*{Acknowledgments}
We thank the volunteers who participated in the PROVIDE trial and the PROVIDE study team including Beth Kirkpatrick, Rashidul Haque, and William A Petri, Jr.

\subsection*{Supplementary Materials}
Code is also available on GitHub (https://github.com/allicodi/vaxstrat and https://github.com/allicodi/vaxstrat\_analysis).

\subsection*{Data Availability}
The data underlying this article are available in ClinEpiDB at https://clinepidb.org/ce/app/ workspace/analyses/DS\_fe95d702c3/new.

\appendix
\section*{Supplemental Material}

\renewcommand{\thethm}{S\arabic{thm}}
\setcounter{thm}{0}

\renewcommand{\thefigure}{S\arabic{figure}}
\setcounter{figure}{0}

\renewcommand{\thetable}{S\arabic{table}}
\setcounter{table}{0}

\renewcommand{\thelemma}{S\arabic{lemma}}
\setcounter{lemma}{0}

\renewcommand{\theequation}{S\arabic{equation}}
\setcounter{equation}{0}

\setcounter{assump}{0}

\section{Additional detail on no inference and consistency assumptions}
The no interference assumption states that the counterfactual outcomes for each individual in the study are independent of the vaccine assignment of other individuals. \begin{assump}
    \emph{No interference}. For any two vaccine assignment vectors $\bm{z} = (z_1, \dots, z_n)$ and $\bm{z}' = (z_1', \dots, z_n')$, then we have that if $z_i = z_i'$ then $S_i(\bm{z}) = S_i(\bm{z}')$. Similarly, for any two infection status vectors $\bm{s} = (s_1, \dots, s_n)$ and $\bm{s}' = (s_1', \dots, s_n')$ if $z_i = z_i'$ and $s_i = s_i'$ then $Y_i(\bm{z}, \bm{s}) = Y_i(\bm{z}', \bm{s}')$.
\end{assump} 
While this assumption of no interference is often violated in infectious disease settings \citep{halloran1995causal}, we make the assumption given the motivating example applies to estimating vaccine effects in Phase 3 studies, where participants represent a relatively small fraction of the at-risk population and the vaccine studied in the trial is not available to individuals outside of the study. In these settings, enrolled individuals are unlikely to come in contact. Extensions of the methods to account for interference are possible in future work. With this assumption, counterfactual infection status can be expressed as $S_i(z)$ and the counterfactual post-infection outcome as $Y_i(z,s)$. 
\begin{assump}
    \emph{Causal consistency}. We have that if $Z_i = z$ then $S_i = S_i(z)$ and in addition if $S_i = s$, then we have that $Y_i = Y_i(z, s)$.
\end{assump} 
The assumption of causal consistency stipulates that if we observe an individual to receive vaccine formulation $z$, then the observed infection outcome $S_i$ equals the counterfactual outcome $S_i(z)$. Moreover, we also have that the observed post-infection outcome $Y_i$ equals the counterfactual $Y_i(z, S_i)$. With this assumption, we can express the counterfactual post-infection outcome as $Y_i(z)$.

\renewcommand{\theassump}{S\arabic{assump}}
\setcounter{assump}{0}

\section{Relationship to existing principal strata literature}

Causal effects in principal strata have been widely used in applied statistics to study problems involving noncompliance \citep{angrist1996identification, frumento2012evaluating, mealli2013using}, 
truncation by death \citep{ding2011identifiability, wang2017identification}, mediation \citep{gallop2009mediation, forastiere2018principal, kim2019bayesian}, and the evaluation of surrogate endpoints \citep{frangakis2002principal, gilbert2008evaluating, jiang2016principal}.

Depending on the estimand of interest, identification of principal strata-based estimands is often facilitated through a combination of assumptions: (i) monotonicity, an example of which is given in Assumption 3; (ii) an exclusion restriction that limits the causal pathways whereby $Z$ can affect $Y$ {\color{changed}(as in Assumption 6 in the main text)}; and (iii) principal ignorability, which states that conditional on a set of auxilliary variables, there is independence between potential outcomes and principal strata membership\citep{jo2009use, feller2017principal, jiang2022multiply}. Others have used strong parametric modeling assumptions to facilitate identification \citep{imai2009statistical,zhang2009likelihood}, though these approaches are often sensitive to small changes in modeling assumptions \citep{ho2022weak}. Barring these assumptions, it is often only feasible to draw inference pertaining to bounds on effects in principal strata \citep{imai2008sharp,zhang2008evaluating}. 

{\color{changed} Building on this past work, we develop identification and estimation theory for a scientifically motivated class of principal stratum estimands targeting post-infection outcomes among the Naturally Infected. Specifically, we develop assumption-free identification of bounds on Naturally Infected effects, as well as approaches for point identification under an exclusion restriction and partial principal ignorability. We further discuss the plausibility of these assumptions specifically in the vaccine and infectious disease context, highlighting trial design features and scientific considerations that may guide the choice of assumptions in practice. Collectively, these results provide a unified characterization of identification and estimation across multiple sets of assumption for a common infectious disease setting.}

Our work on bounds is related, but distinct from previous work identifying bounds for effects in the Doomed strata \citep{hudgens2006causal}. We also draw connections between Naturally Infected effects and the chop lump test that has been proposed for testing vaccine effects on infection-necessary post-infection outcomes \citep{follmann2009chop}.

Our results pertaining to identification under an exclusion restriction is closely related to the well-known local average treatment effect under one-sided non-compliance \citep{angrist1996identification} and recent work on efficient ``treatment responder'' analysis \citep{nordland2024estimation}. {\color{changed} However, this appears to be the first discussion of how these approaches can be used to study effects on post-infection endpoints in the context of infectious diseases.}

The findings regarding identification and estimation of effects under a form of principal ignorability relate closely to recent results on efficient and robust estimation of principal strata effects \citep{jiang2022multiply}. However, in contrast to these results, our principal stratum of interest is partially identifiable under standard assumptions, leading to identification under a weaker form of principal ignorability than is typically considered in the literature. {\color{changed} In addition, we derive observed-data restrictions and corresponding semiparametric efficiency results for combinations of assumptions considered in this paper, yielding semiparametric efficient estimators tailored to the Naturally Infected estimands. To complement our results on effects in the Naturally Infected, Supplement \ref{sec:doomed estimand} provides detailed identification and estimation procedures for effects in the Doomed principal stratum \citep{hudgens2006causal, halloran2012causal} under a form of principal ignorability, which to our knowledge have not been previously explicitly developed.}

\section{Inverse probability weighting and plug-in estimators of Naturally Infected effects}

In this section, we describe singly robust estimators of the effects of interest. These estimators are generally compatible with estimation of relevant nuisance parameters using only parametric working models, with inference obtained utilizing appropriate nonparametric bootstrap methods. We provide explicit expressions for the singly robust estimators of $\psi_0$; estimators of other estimands follow straightforwardly from their identifying functionals.

To generate plug-in estimators, we can replace nuisance parameters appearing in their identifying functionals with estimates based on parametric working models. Thus, for example, a plug-in estimator of $\psi_0$ can be computed as \begin{align*}
    \psi_{0,n} = \frac{1}{n} \sum_{i=1}^n  \frac{\bar{\rho}_{0,n}(X_i)}{\bar{\rho}_{0,n}} \mu_{10,n}(X_i) \ .
\end{align*}
Similar estimators for other identifying functionals can easily be constructed.

To generate inverse probability weighted (IPW) estimators, we must first express identifying functions in a suitable IPW form. For example, $\psi_0$ can be expressed as \begin{align*}
    \psi_0 = \E \left[ 
        \frac{S}{\bar{\rho}_0} \frac{(1 - Z)}{\pi_0(X)} Y 
    \right] \ . 
\end{align*}
An IPW estimator can then be constructed as \begin{align*}
    \psi_{0,n} = \frac{1}{n} \sum_{i=1}^n  \frac{S_i}{\bar{\rho}_{0,n}} \frac{(1 - Z_i)}{\pi_{0,n}(X_i)} Y_i \ ,
\end{align*}
where $\pi_{0,n}$ can either be estimated based on a parametric working model or can use the known randomization probability and $\bar{\rho}_{0,n}$ can either be based on a marginalized parametric working model or can be set to the sample proportion of infected placebo recipients.

A similar strategy can be used to generate IPW estimates of the other estimands described. These can be based off the following IPW representations of parameters, including those defined in the Doomed population (see Supplement \ref{sec:doomed estimand}): \begin{align*}
    \psi_{1,\text{ER}} &= 
    \frac{1}{\bar{\rho}_0} \left( \E\left\{ \frac{Z}{\pi_1(X)} Y \right\} - \E\left[\frac{(1-Z)(1-S)}{\pi_0(X)} Y \right] \right)  \\
    \psi_{1,\text{PI}} &= \E \left(
        \frac{S}{\bar{\rho}_0} 
        \left[
            \frac{Z}{\pi_1(X)} + \frac{(1 - Z)}{\pi_0(X)} \left\{1 - \frac{\rho_1(X)}{\rho_0(X)} \right\}
        \right] Y 
    \right) \ , \\
    \eta_0 &= \E \left\{ 
        \frac{(1 - Z)}{\pi_0(X)} \frac{S}{\rho_0(X)} \frac{\rho_1(X)}{\bar{\rho}_1} Y
    \right\} \ , \\
    \eta_1 &= \E \left\{ 
        \frac{Z}{\pi_1(X)} \frac{S}{\bar{\rho}_1} Y
    \right\} \ . 
\end{align*}

\section{Additional results on bounds}

\subsection{Estimation of bounds with tied outcomes}

To estimate bounds for a post-infection outcome with ties, we can compute $\bar{\rho}_{z,n}$,  $\bar{\mu}_{11,n}$, and $q_n$ as above. Let $n_{10} = \sum_{i=1}^n I(Z_i = 1, S_i = 0)$ denote the number of uninfected vaccine recipients and define $n^* = \lceil q_n \times n_{10} \rceil$. To obtain an estimate of the lower bound, we order post-infection outcomes in the uninfected vaccine recipients from smallest to largest. Let $Y^*_{[i]}$ denote the $i$-th smallest value observed in this group, $i = 1, \dots, n_{10}$. We can then compute the estimate $\bar{\mu}_{10,\ell,n} = \frac{1}{n^*} \sum_{i=1}^n Y^*_{[i]}$, which is the average of the $n^*$ smallest values of the post-infection outcome in the vaccine uninfected group. This estimate can then be used to compute the final estimate $\ell_n$ of $\ell$. An estimate of the upper bound can be computed by averaging the $n^*$ largest outcomes in the uninfected vaccine recipients to generate an estimate $\bar{\mu}_{10,u,n}$ that can similarly be used to compute an estimate $u_n$ of $u$.

\subsection{Covariate-adjusted bounds}

We propose the following covariate-adjusted bounds. Let 
\begin{align*}
\ell(x) &= \mu_{11}(x)\frac{\rho_{1}(x)}{\rho_{0}(x)} 
+ \mu_{10,\ell}(x)\left(1 - \frac{\rho_{1}(x)}{\rho_{0}(x)}\right) \ , \ \mbox{and} \\
u(x) &= \mu_{11}(x)\frac{\rho_{1}(x)}{\rho_{0}(x)} 
+ \mu_{10,u}(x)\left(1 - \frac{\rho_{1}(x)}{\rho_{0}(x)}\right) \ ,
\end{align*}
{\color{changed}
where $\mu_{10,\ell}(x) = \E\{ Y \mid Z = 1, S = 0, X = x, Y < Y_{\ell(x)} \}$, with $Y_{\ell(x)}$ equal to the $q(x) = \{\rho_0(x) - \rho_1(x)\} / \{1 - \rho_1(x)\}$-th quantile of the distribution of $Y \mid Z = 1, S = 0, X = x$ and $\mu_{10,u}(x)$ similarly defined using the $\{1 - q(x)\}$-th quantile.

Following the proof of Theorem 3, we have that $(\ell(x), u(x))$ is a valid bound for $\E\{Y(1) \mid S(0) = 1, X = x\}$ for any given $x$. Thus, \begin{align*}
\bar{\ell} = \sum_x l(x) \frac{\rho_0(x)}{\bar{\rho}_0} P(X=x) \ \ \text{and} \ \ \bar{u} = \sum_x u(x) \frac{\rho_0(x)}{\bar{\rho}_0} P(X=x)
\end{align*}
are bounds for the marginal quantity $\E\{Y(1) \mid S(0) = 1\}$.

We now derive a condition under which the covariate-adjusted lower bound is sharper than the unadjusted lower bound under a simple randomization scheme with $Z \perp X$. The covariate-adjusted lower bound can be written
\[
\bar\ell = \frac{1}{\bar\rho_0} \sum_x \left[
\mu_{11}(x)\rho_1(x) + \mu_{10,\ell}(x)\{\rho_0(x)-\rho_1(x)\} \right]P(X=x).
\]

The condition $\bar\ell > \ell$ is equivalent to
\[
\sum_x \left[ \mu_{11}(x)\rho_1(x) + \mu_{10,\ell}(x)\{\rho_0(x)-\rho_1(x)\} \right]P(X=x) > \bar\mu_{11}\bar\rho_1 + \bar\mu_{10,\ell}(\bar\rho_0-\bar\rho_1) \ . 
\]
By the law of total expectation and the assumption that $ Z\perp X$,
\[
\bar\mu_{11} = \sum_x \mu_{11}(x) \frac{\rho_1(x)}{\bar\rho_1} P(X=x) \ . 
\]
Thus, the terms involving $\mu_{11}$ cancel, leaving
\[
\sum_x \mu_{10,\ell}(x)\{\rho_0(x)-\rho_1(x)\}P(X=x) > \bar\mu_{10,\ell}(\bar\rho_0-\bar\rho_1) \ . \]

Let $\rho_{0-1}(X) = \rho_0(X)-\rho_1(X)$. Then we have $\bar{\ell} > \ell$ if $
\E\{\mu_{10,\ell}(X) \rho_{0-1}(X)\} > \bar\mu_{10,\ell}\E\{\rho_{0-1}(X)\}.$ This is equivalent to
\[
\text{Cov}\{\mu_{10,\ell}(X),\rho_{0-1}(X)\}
 >  - (\overline{\mu_{10,\ell}} - \bar\mu_{10,\ell}) \bar{\rho}_{0-1} \ ,
\]
where $\overline{\mu_{10,\ell}} = \E\{ \mu_{10,\ell}(X) \}$ and $\bar{\rho}_{0-1} = \E\{ \rho_{0-1}(X) \}.$

Similar calculations show that for the upper bound, $\bar{u} < u$ if
\[
\text{Cov}\{\mu_{10,u}(X),\rho_{0-1}(X)\}
<
-
(\overline{\mu_{10,u}}-\bar\mu_{10,u})\bar\rho_{0-1},
\]
where $\overline{\mu_{10,u}}=\E\{\mu_{10,u}(X)\}$.

The left hand side of the inequalities involve covariance between trimmed conditional means of $Y$ and the conditional probability of being in the Protected stratum. The right hand side involves the difference between the averaged trimmed mean with trimming at a conditional quantile, compared to the trimmed mean with trimming at the marginal quantile. The relationship between average conditional trimmed means based on varying quantiles to a trimmed marginal mean at a fixed quantile does not appear straightforward to study. We leave a general study of these inequalities to future work and here comment on special cases to provide some basis for intuition in selecting covariates. 

First, we note that in the ``null'' case, $Y \perp X \mid Z = 1, S = 0$ and $S \perp X \mid Z = 1$, covariate-adjusted bounds cannot be tighter than unadjusted bounds as both sides of the inequality evaluate to 0. Thus, we may conclude that covariates $X$ should be selected that are associated with either infection, post-infection outcomes, or both. 

In fact, we can show in the special case for a location-family outcome model, it is sufficient for a covariate $X$ only to be associated with post-infection outcomes and not of risk of infection to generate tighter bounds. Suppose we can assume that: \begin{itemize} 
\item[(i)]$\rho_0(x)$ and $\rho_1(x)$ are constant in $x$ with $\rho_1(x) = r_1 < r_0 = \rho_0(x)$; 
\item[(ii)] $\mu_{10}(x)$ is not constant in $x$; 
\item[(iii)] $Y$ given $Z = 1, S = 0, X$ follows the model: $Y = \mu_{10}(X) + \epsilon$, with $\epsilon \perp X \mid Z = 1, S = 0$, $\E( \epsilon \mid Z = 1, S = 0 ) = 0$. 
\end{itemize}

In this case, both the lower and upper bounds of the covariate-adjusted bounds will be tighter than the unadjusted bounds. To prove this,  we adopt the notational convention that $\E_{10}\{ f(X) \}$ denotes the conditional expectation of $f(X)$ given $Z = 1, S = 0$.

By (i), $\text{Cov}\{\mu_{10,\ell}(X), \rho_{0-1}(X)\} = 0$ and $\text{Cov}\{\mu_{10,u}(X), \rho_{0-1}(X)\} = 0$. Thus, it remains to show that $\bar{\mu}_{10,\ell} < \overline{\mu_{10,\ell}}$ and $\bar{\mu}_{10,u} > \overline{\mu_{10,u}}$.

Also by (i), the quantile used to trim the mean $q(x)$ is constant in $x$. Letting $p$ denote this constant value and letting $\epsilon_p$ denote the $p$-th quantile of the distribution of $\epsilon$, by (iii) we have $\mu_{10,\ell}(x) = \mu_{10}(x) + \E_{10}( \epsilon \mid \epsilon < \epsilon_p)$. Thus, $\overline{\mu_{10,\ell}} = \mu_{10,\ell}(X) = \E_{10}\{ \mu_{10}(X) \} + \E_{10}( \epsilon \mid \epsilon < \epsilon_p )$. Let $m = \E_{10}\{ \mu_{10}(X) \}$ and define the random variable $\tilde{Y} = m + \epsilon$, and note that $\overline{\mu_{10,\ell}} = \E_{10}\{ \tilde{Y} \mid \tilde{Y} < \tilde{Y}_p \}$, where $\tilde{Y}_p$ is the $p$-th quantile of $\tilde{Y}$. 

Under the location shift model (iii), the marginal distribution of $Y \mid Z = 1, S = 0$ can be characterized by noting $Y = m + \epsilon + \{ \mu_{10}(X) - m \}$. Thus, the random variable $Y$ is a mean-preserving spread of of $\tilde{Y}$ \citep{Rothschild1978IncreasingRisk} and therefore $\tilde{Y}$ is strictly stochastically less than $Y$, $\tilde{Y} \prec Y$. A consequence of this stochastic ordering is that the lower tail mean will be smaller for the more dispersed random variable, i.e., that $\bar{\mu}_{10,\ell} < \overline{\mu_{10,\ell}}$. Identical arguments show that $\bar{\mu}_{10,u} > \overline{\mu_{10,u}}$, thus proving that covariate adjusted bounds will be tighter in this case.

Evaluating the inequalities when $X$ relates to both infection risk and post-infection outcomes or when $X$ only relates to infection risk is more challenging, as this induces the need to precisely understand the tail behavior of the distribution of the post-infection outcome. We leave this study to future work.
}

\section{Sensitivity analysis for partial principal ignorability}

To assess sensitivity to the partial principal ignorability assumption, we propose an identification based on the following assumption.
\begin{assump}
    For all $x$ and for $\epsilon \in \mathbb{R}^+$, \begin{equation*}
        \frac{\E\{ Y(1) \mid S(1) = 0, S(0) = 0, X = x\}}{\E\{Y(1) \mid S(1) = 0, S(0) = 1, X = x \}} = \epsilon \ . \label{eq:sens_assmp}
\end{equation*} 
\label{assmp:partial_pi_sens_assump}
\end{assump}
Similar sensitivity analyses for other principal effects are described in \citet{ding2017principal}. 

\begin{thm}
Under Assumptions 1-5 and Assumption S1, $\E\{ Y(1) \mid S(0) = 1 \} = \psi_{1, \text{PI}, \epsilon}$, where \begin{equation}
    \psi_{1, \text{PI}, \epsilon} = \E \left(
    \frac{\rho_0(X)}{\bar{\rho}_0} \left[
     \mu_{11}(X) \frac{\rho_1(X)}{\rho_0(X)} + \mu_{10}(X) \frac{ \{ 1 - \rho_1(X) \} }{ \rho_0(X) - \rho_1(X) + \epsilon \{ 1 - \rho_0(X) \}} \left\{ 1 - \frac{\rho_1(X)}{\rho_0(X)} \right\}
    \right]
    \right)
    \label{eqn:id_psi1_epsilon}
\end{equation}
\label{thm:sens_pi_id}
\end{thm}
Assumption \ref{assmp:partial_pi_sens_assump} relates the expected post-infection outcome in the Immune stratum relative to the Protected stratum. We note that partial principal ignorability (Assumption 7) implies that $\epsilon = 1$ and in that case (\ref{eqn:id_psi1_epsilon}) reduces to the previous identification result given in Theorem 5. A sensitivity analysis can be implemented by varying the value of the constant $\epsilon$ and demonstrating how estimates of $\psi_{1,\epsilon}$ vary with $\epsilon$. We suggest that a relevant sensitivity analysis is to report values of $\epsilon$ at which the point estimate (or confidence interval limit) of the effect is equal to the estimated lower and upper bounds. Beyond this restriction of the possible values of $\epsilon$, we expect that scientific context can often inform a narrower range of values.

To aid in construction of efficient estimators of $\psi_{1, \text{PI}, \epsilon}$ we have the following theorem establishing its efficient influence function in our model. We define the following quantities, which are useful for concisely expressing the result: \begin{align*}
    \psi_{1, \text{PI}, \epsilon} 
    &= \E \left(
    \frac{\rho_0(X)}{\bar{\rho}_0} \left[
     \mu_{11}(X) \frac{\rho_1(X)}{\rho_0(X)} + \mu_{10}(X) \frac{ \{ 1 - \rho_1(X) \} }{ \rho_0(X) - \rho_1(X) + \epsilon \{ 1 - \rho_0(X) \}} \left\{ 1 - \frac{\rho_1(X)}{\rho_0(X)} \right\}
    \right]
    \right) \\
    &= \E \left(
    \underbrace{\frac{\rho_1(X)}{\bar{\rho}_0}\mu_{11}(X)}_{\psi_{11,\text{PI}, \epsilon \mid X}(X)} + \underbrace{\frac{\rho_0(X) -  \rho_1(X) }{\bar{\rho}_0} \frac{ \{ 1 - \rho_1(X) \} }{ (1-\epsilon) \rho_0(X) - \rho_1(X) + \epsilon  \}}  \mu_{10}(X)}_{\psi_{10,\text{PI}, \epsilon \mid X} (X)} 
    \right) \ .
\end{align*}
We similarly define $\psi_{11,\text{PI}, \epsilon} = \E\{ \psi_{11,\text{PI},\epsilon \mid X} (X) \}$ and $\psi_{10, \text{PI}, \epsilon} = \E\{ \psi_{10, \text{PI}, \epsilon \mid X} (X) \}$

\begin{thm}
    The EIF of $\psi_{1, \text{PI}, \epsilon}$ in the observed-data model implied by Assumptions 1-5 is $\Phi_{1, \epsilon} - \psi_{1, \text{PI}, \epsilon}$, where ,
\begin{equation}
\begin{aligned}
    \Phi_{1, \epsilon}(O_i) &= \frac{Z_i}{\pi_1(X_i)} \frac{S_i}{\bar{\rho}_0} \{ Y_i - \mu_{11}(X_i) \} + \frac{Z_i}{\pi_1(X_i)} \frac{\mu_{11}(X_i)}{\bar{\rho}_0} \{ S_i - \rho_{1}(X_i) \} \\ 
    &\hspace{0.15cm} - \frac{\psi_{11,\text{PI},\epsilon}}{\bar{\rho}_0} \frac{(1 - Z_i)}{\bar{\pi}_0} \{ S_i - \bar{\rho}_0 \} + \psi_{11,\text{PI},\epsilon \mid X}(X_i) 
    \\
    &\hspace{0.15cm} + \frac{Z_i}{\pi_1(X_i)} \frac{(1 - S_i)}{\bar{\rho}_0} \frac{\{ \rho_0(X_i) - \rho_1(X_i) \} }{\{ (1 - \epsilon) \rho_0(X_i) - \rho_1(X_i) + \epsilon\} } \{ Y_i - \mu_{10}(X_i) \} \\ 
    &\hspace{0.15cm} + \frac{(1 - Z_i)}{\pi_0(X_i)} \frac{\{1 - \rho_1(X_i)\}}{\{ (1-\epsilon) \rho_0(X_i) - \rho_1(X_i) + \epsilon \} } \frac{\mu_{10}(X_i)}{\bar{\rho}_0} \{ S_i - \rho_0(X_i) \} \\ 
    &\hspace{0.15cm} - \frac{Z_i}{\pi_1(X_i)} \frac{\{ 1- \rho_1(X_i)\}}{\{ (1-\epsilon) \rho_0(X_i) - \rho_1(X_i) + \epsilon\} } \frac{\mu_{10}(X_i)}{\bar{\rho}_0} \{ S_i - \rho_1(X_i) \}   \\
    &\hspace{0.15cm} - \frac{\psi_{10,\text{PI},\epsilon}}{\bar{\rho}_0} \frac{(1 - Z_i)}{\bar{\pi}_0} \{ S_i - \bar{\rho}_0 \} \\
    &\hspace{0.15cm} - \frac{Z_i}{\pi_1(X_i)} \frac{\{ \rho_0(X_i) - \rho_1(X_i) \}}{\bar{\rho}_0} \frac{\mu_{10}(X_i)}{\{ (1-\epsilon) \rho_0(X_i) - \rho_1(X_i) + \epsilon\}} \{ S_i - \rho_1(X_i) \}  \\ 
    &\hspace{0.15cm} - (1-\epsilon) \frac{(1 - Z_i)}{\pi_0(X_i)} \frac{\{ \rho_0(X_i) - \rho_1(X_i)\}}{\bar{\rho}_0} \frac{\{1 - \rho_1(X_i)\}}{\{(1-\epsilon) \rho_0(X_i) - \rho_1(X_i) + \epsilon\}^2} \mu_{10}(X_i) \{ S_i - \rho_0(X_i) \}  \\ 
    &\hspace{0.15cm} + \frac{Z_i}{\pi_1(X)} \frac{\{ \rho_0(X_i) - \rho_1(X_i) \}}{\bar{\rho}_0} \frac{\{ 1 - \rho_1(X_i)\}}{\{(1-\epsilon) \rho_0(X_i) - \rho_1(X_i) + \epsilon\}^2} \mu_{10}(X_i) \{ S_i - \rho_1(X_i) \}  \\ 
    &\hspace{0.15cm} + \psi_{10,\text{PI},\epsilon \mid X}(X_i) \ . 
\end{aligned}
\label{eqn:eif_psi1eps}
\end{equation}
\label{thm:eif_psi1eps}
\end{thm}

This can be shown using the same techniques outlined for other parameters above.

\section{Design considerations for identifying assumptions}

Both the assumption of exclusion restriction and partial principal ignorability are fundamentally cross-world in nature, with both assumptions involving a condition on counterfactuals defined under vaccination and no vaccination. Thus, these conditions must be scrutinized in each context to determine their plausibility. 

A key consideration for the validity of the exclusion restriction assumption is whether and to what extent the random variable $S$ truly measures infection status. If $S$ is a highly sensitive measure of infection, then the exclusion restriction is likely reasonable for many post-infection outcomes -- there is generally no way for a vaccine to impact outcomes that directly result from infection in the absence of an infection. However, randomized trials commonly employ \emph{passive surveillance} for infections, whereby participants are encouraged to seek care if they experience symptoms related to infection. At these visits, infection is confirmed using an appropriate diagnostic. Barring symptoms, however, participants may only be seen at several routinely scheduled study visits. Such a design may lead to asymptomatic or mildly symptomatic infections being missed during the course of follow-up. The possibility for missed infections may call into question the validity of the exclusion restriction, unless it can be argued that either asymptomatic infections are so mild as to have no impact on the post-infection outcome of interest or that the vaccine has no effect on asymptomatic infections. If $S$ is not a sensitive measure of infection and asymptomatic/mildly symptomatic infections are likely to impact the outcome of interest, then it may be preferable to base inference on bounds or the weak principal ignorability estimand and include relevant sensitivity analyses to assess robustness of results to these assumptions.

A notable exception to the above discussion regarding plausibility of the exclusion restriction is post-infection outcomes $Y$ that are also potential side effects of the vaccine, such as conditions that may be candidates for selection as adverse events of special interest (AESI). Such events are often negative side effects that are related to the biological mechanism of the vaccine. Because the mechanism of vaccines is often to simulate a mild infection, clinical AESI events often occur after natural infection as well and therefore may be interesting to study as post-infection endpoints. In this case, the exclusion restriction would be unlikely to hold, as we would expect a negative vaccine effect in the Immune, reflecting the occurrence of vaccine-related AESIs. We would argue that naturally infected effects are unlikely to be the target casual effect of interest for these outcomes because they exclude important vaccine effects in the Immune; population-level effects may be more clinically relevant here.

\section{Semiparametric estimator under exclusion restriction and partial principal ignorability}

\begin{thm}
    If both the exclusion restriction and partial principal ignorability hold then $\mu_{10}(x) = \mu_{00}(x)$ for each $x$ and thus $Y$ is mean independent of $Z$ given $S = 0, X$. Moreover, for each distribution in a semiparametric model that respects this conditional independence, we have $\psi_{1,\text{ER}} = \psi_{1,\text{PI}}$. 
    \label{thm:equiv_er_and_pi_id}
\end{thm}
We use $\psi_{1,\cdot}$ to denote the common value of $\psi_{1,\text{ER}}$ and $\psi_{1, \text{PI}}$ in this model. Under this set of assumptions, it is possible to use either $\psi_{1,\text{ER},n}^+$ or $\psi_{1,\text{PI},n}^+$ to estimate effects of interest; however, both will be inefficient. The key insight is that in this model the $X$-conditional mean of $Y(1)$ in the Protected strata is identified by $\mu_{\cdot 0}(X) = \E(Y \mid S = 0, X)$. Thus, a plug-in estimator can be constructed as  $
\psi_{1,\cdot,n} = n^{-1} \sum_{i=1}^n \left[ \mu_{11,n}(X_i) \rho_{1,n}(X_i) +  \mu_{\cdot 0,n}(X_i) \left\{ \rho_{0,n}(X_i) - \right. \right.$ $\left.\left. \rho_{1,n}(X_i) \right\} \right]/ \bar{\rho}_{0,n},$ where $\mu_{\cdot 0, n}$ can either be estimated via direct regression of $Y$ on $X$ in the subset of data with $S = 0$ or by marginalizing estimates $\mu_{10,n}$ and $\mu_{00,n}$. Efficient one-step estimation is facilitated via the following gradient. Let $\bar{\rho}_{\cdot} = P(S = 1)$.

\begin{thm}
    The efficient gradient for regular estimators of $\psi_{1,\cdot}$ in a semiparametric model for the observed data that assumes positivity and respects both the exclusion restriction and weak principal ignorability is:\begin{align*}
        \Phi_{1, \cdot}(O_i) &= \frac{Z_i}{\pi_1(X_i)} \frac{S_i}{\bar{\rho}_0} \{Y_i - \mu_{11}(X_i) \} + \frac{(1 - S_i)}{1 - \bar{\rho}_{\cdot}} \frac{\rho_0(X_i) - \rho_1(X_i)}{\bar{\rho}_0} \{ Y_i - \mu_{\cdot 0}(X_i) \} 
\\
&\hspace{3em} + \frac{Z_i}{\pi_1(X_i)} \frac{\mu_{11}(X_i) - \mu_{\cdot 0}(X)}{\bar{\rho}_0} \{ S_i - \rho_1(X_i) \} + \frac{(1 - Z_i)}{\pi_0(X_i)} \frac{\mu_{\cdot0}(X_i) - \psi_{1,\cdot}}{\bar{\rho}_0} \{ S_i - \rho_0(X_i) \} \\
&\hspace{3em} - \frac{\psi_{1,\cdot}}{\bar{\rho}_0} \{ \rho_0(X_i) - \bar{\rho}_0 \} + \tilde{\psi}_1(X_i) - \psi_{1, \cdot}
    \end{align*}
    \label{thm:grad_for_semiparam}
\end{thm}

One-step estimators can be constructed based on this gradient. Under regularity conditions given in the Proof section below, $n^{1/2} (\psi_{1,\cdot, n}^+ - \psi_{1,\cdot})$ converges in distribution to a mean-zero Gaussian random variable with variance $\E\{ \Phi_{1,\cdot}(O)^2 \}$. Robustness conditions for $\psi_{1,\cdot}$ are essentially the same as for $\psi_{1,\text{PI}}$ (see Section K.11).

\section{Identification, estimation, and interpretation of effects in the Doomed} \label{sec:doomed estimand}

\subsection{Identification}
The effect of a vaccine on post-infection outcome in the Doomed strata is a contrast in $z$ of $\E\{ Y(z) \mid S(0) = 1, S(1) = 1 \}$.

\begin{thm}
    Under Assumptions 1-5, $\E\{Y(1) \mid S(0) = 1, S(1) = 1 \} = \eta_1$, where \[
    \eta_1 = \bar{\mu}_{11} = \E \left\{ \frac{\rho_1(X)}{\bar{\rho}_1} \mu_{11}(X) \right\} \ .
    \]
    \label{thm:id_doomed_vax}
\end{thm}

To identify the counterfactual mean in the Doomed stratum under placebo, we make two further assumptions.

\begin{assump}
    \emph{Partial principal ignorability}: $S(1) \perp Y(0) \mid S(0) = 1, X$
    \label{assmp:cross_world_for_doomed}
\end{assump}

\begin{thm}
    Under Assumptions 1-5 (from the main body) and Assumption \ref{assmp:cross_world_for_doomed},
    $\E\{Y(0) \mid S(0) = 1, S(1) = 1 \} = \eta_0$, where \[
    \eta_0 = \E \left\{ \frac{\rho_1(X)}{\bar{\rho}_1} \mu_{01}(X) \right\} \ .
    \]
    \label{thm:id_doomed_placebo}
\end{thm}

\subsection{Efficiency theory}

We define $\tilde{\eta}_0(x) = \rho_1(x) \mu_{11}(x) / \bar{\rho}_1$ and $\tilde{\eta}_1(x) = \rho_1(x) \mu_{01}(x) / \bar{\rho}_1$.

\begin{thm}
    The EIF of $\eta_1$ in the observed-data model implied by Assumptions 1-5 and Assumption S2 is \begin{align*}
    \Theta_1(O_i) &= \frac{Z_i}{\pi_1(X_i)} \frac{S_i}{\bar{\rho}_1} \{ Y_i - \mu_{11}(X_i) \} \\
    &\hspace{2em} + \frac{\{\mu_{11}(X_i) - \eta_1\}}{\bar{\rho}_1} \frac{Z_i}{\pi_1(X_i)} \{ S_i - \rho_1(X_i) \} \\
    &\hspace{2em} - \frac{\eta_1}{\bar{\rho}_1} \{ \rho_1(X_i) - \bar{\rho}_1 \} + \tilde{\eta}_1(X_i) - \eta_1 \ .
    \end{align*}
    \label{thm:eif_doomed_vax}
\end{thm}

\begin{thm}
    The EIF of $\eta_0$ in the observed-data model implied by Assumptions 1-5 and Assumption S2 is \begin{align*}
    \Theta_0(O_i) &= \frac{(1 - Z_i)}{\pi_0(X_i)} \frac{S_i}{\bar{\rho}_1} \frac{\rho_1(X_i)}{\rho_0(X_i)} \{ Y_i - \mu_{01}(X_i) \} \\
    &\hspace{2em} + \frac{\{\mu_{01}(X_i) - \eta_0\}}{\bar{\rho}_1} \frac{Z_i}{\pi_1(X_i)} \{ S_i - \rho_1(X_i) \} \\
    &\hspace{2em} - \frac{\eta_0}{\bar{\rho}_1} \{ \rho_1(X_i) - \bar{\rho}_1 \} + \tilde{\eta}_0(X_i) - \eta_0 \ .
    \end{align*}
\end{thm}

As in the main body, these gradients can be used to formulate efficient one-step estimators of the effects of interest.

\subsection{Exposure-conditional interpretation}

We can also formulate an equivalent exposure-conditional interpretation of the Doomed-only estimand as follows. As with the formulation of exposure for the Naturally Infected estimand, we assume there is a binary coarsening $e^*: \tilde{S} \times \mathcal{X} \rightarrow \{0, 1\}$ and define the random variable $E^* = e^*(\tilde{S}, \mathcal{X})$. 
As previously, we make several assumptions regarding this exposure variable.

\begin{assump}
    \emph{Exposure is sufficient for infection irrespective of vaccine}: $P(S = 1 \mid E^* = 1, X = x) = 1$ for all $x$.
    \label{assmp:doom_exp_suff}
\end{assump}

\begin{assump}
    \emph{Exposure is necessary for infection under vaccine}: $P(S = 1 \mid E^* = 0, Z = 1, X = x) = 0$ for all $x$.
    \label{assmp:doom_exp_nec_in_vax}
\end{assump}

\begin{assump}
    \emph{Vaccine is not a cause of and is independent of exposure}: $E^*(z) = E^*$ for $z = 0, 1$ and $E^* \perp Z \mid X$.
    \label{assmp:doom_vax_not_cause}
\end{assump}

\begin{assump}
    \emph{Conditional ignorability of exposure}: $\E(Y \mid Z = 0, E^* = 1, S = 1, X = x) = \E(Y \mid Z = 0, S = 1, X = x)$ for all $x$.
    \label{assmp:doom_no_unmeasured_confounders}
\end{assump}

Notably, for this formulation, we do not require that $E^*$ is necessary for infection under placebo. Thus, for example, we could imagine that relative to the original exposure variable $E$, the variable $E^*$ may represent a higher dosage of challenge to the infectious agent, such that unvaccinated individuals may experience a clinical infection following exposure to either $E$ or $E^*$, but exposure to the higher dose $E^*$ is not necessary for infection under placebo. We consider identifying the parameter $\E\{ Y(v) \mid E^* = 1 \}$ for $v = 0, 1$, which can then be used to construct causal contrasts of interest. 

\begin{thm}
    Under Assumptions 1-5 and S3-S6, $\E\{ Y(1) \mid E^* = 1 \} = \eta_1$ and $\E\{ Y(0) \mid E^* = 1 \} = \eta_0$.
    \label{thm:exp_cond_doomed_id}
\end{thm}

\section{Simulations}
\subsection{Results for ``Asymptotic properties of estimators'' simulation}

\subsubsection{Data generating process details}\label{sec:dgp_generic}

For each simulation, we generated a dataset of size $n \in \{500,4000\}$. Baseline covariates $X=(X_1,X_2,X_3)$ were generated independently with $X_j \sim \text{Bernoulli}(0.5)$ for $j=1,2,3$. To generate infection potential outcomes, we set $P\{S(1)=1,S(0)=1\mid X\}=\text{expit}(-1+0.5X_1-X_1X_2-0.5X_3)$, $P\{S(1)=0,S(0)=0\mid X\}=\text{expit}(-1+0.5X_1-X_1X_3-0.5X_3)$, and the conditional probability of $S(1)=0, S(0)=1$ given by one minus these two conditional probabilities. Infection potential outcomes were generated deterministically based on stratum membership. 

Binary outcome potential outcomes were generated from stratum- and treatment-specific logistic regression models. For individuals in the Doomed stratum, $P\{Y(0)=1 \mid S(0) = 1, S(1) = 1, X \} = \text{expit}(-1+0.5X_1-X_1X_2+0.5X_3)$ and $P\{Y(1)=1 \mid S(0) = 1, S(1) = 1, X\}=\text{expit}[\text{logit}\{P\{Y(0)=1\mid S(0) = 1, S(1) = 1, X\}\}+ 0.1]$, yielding a small effect of vaccine in the Doomed stratum. For Immune individuals, $P\{Y(0)=1\mid S(0) = 0, S(1) = 0, X\}=\text{expit}(-0.5+0.5X_1-X_1X_3+0.5X_2)$ and $P\{Y(1)=1\mid S(0) = 0, S(1) = 0, X\} = \epsilon_{\text{I}} P\{Y(0)=1\mid S(0) = 0, S(1) = 0, X\}$. Thus, the parameter $\epsilon_{\text{I}}$ was used to control the extent to which the exclusion restriction was violated. For Protected individuals, we set $P\{Y(0)=1\mid S(0) = 1, S(1) = 0, X\} = P\{Y(0)=1\mid S(0) = 1, S(1) = 1, X\}$ and $P\{Y(1) = 1 \mid S(0) = 1, S(1) = 0, X\} = \epsilon_{\text{P}} P\{Y(1)=1\mid S(0) = 0, S(1) = 0, X\}$. Thus, $\epsilon_{\text{P}}$ was used to control the extent to which partial principal ignorability was violated. Vaccine assignment $Z$ was generated according to $P(Z=1\mid X = x)=1/2$ for all $x$. 

In the scenario where both assumptions held, we set $\epsilon_{\text{P}}=1$ and $\epsilon_{\text{I}}=1$. In other three scenarios, we set either $\epsilon_{\text{P}}=0.5$ (to violate partial principal ignorability) and/or $\epsilon_{\text{I}}=0.5$ (to violate the exclusion restriction).

Observed infection and outcome were then set as $S=S(Z)$ and $Y=Y(Z)$. 

The true value of counterfactual means in the Naturally Infected in each of the four settings are shown in Table \ref{tab:sim1_truth}. These values were calculated based on a single independent Monte Carlo sample of size 10,000,000 generated from the same data-generating process, leveraging the full set of potential outcomes. 

\begin{table}[!h]
\centering
\caption{\label{tab:sim1_truth}True effects for ``Asymptotic properties of estimators'' simulation}
\fontsize{9}{11}\selectfont
\begin{tabular}{lccrr}
\toprule
 &  &  & \multicolumn{2}{c}{Effect} \\
\cmidrule(lr){4-5}
 & $\E\{Y(1)\mid S(0)=1\}$ 
 & $\E\{Y(0)\mid S(0)=1\}$ 
 & Additive 
 & Multiplicative \\
\midrule
\textbf{PI and ER satisfied} & 0.405 & 0.333 & 0.072 & 1.216\\
\textbf{PI satisfied, ER violated} & 0.257 & 0.333 & -0.076 & 0.772\\
\textbf{PI violated, ER satisfied} & 0.257 & 0.333 & -0.076 & 0.772\\
\textbf{PI and ER violated} & 0.183 & 0.333 & -0.150 & 0.550\\
\bottomrule
\end{tabular}
\label{tab:true_values}
\end{table}

For each scenario and sample size, one thousand simulated data sets were analyzed using bounds and point estimates. Nuisance parameters were estimated using saturated logistic regression models ensuring consistent nuisance parameter estimation. Performance was summarized in terms of bias (scaled by $n^{1/2}$), variance and mean squared error (scaled by $n$), and coverage of nominal 95\% Wald confidence intervals based on the estimated influence function. We report these results for both additive and multiplicative effects.

\subsubsection{Results}

The estimators of the bounds performed well in terms of bias and confidence interval coverage for the theoretical value of the bound across all settings (Table \ref{tab:sim1_bounds_unadjusted_combined}). The bounds also captured true effect a high proportion of the time in small samples and 100\% of the time in large samples, irrespective of whether partial principal ignorability and/or exclusion restrictions held. However, bounds were wide and, as expected, median width was not impacted by sample size. Covariate adjustment narrowed the bounds marginally, though adjusted bounds were still wide (Tables \ref{tab:sim1_bounds_n500_rearranged} and \ref{tab:sim1_bounds_n4000_rearranged})

In settings where both partial principal ignorability and the exclusion restriction were satisfied, all point estimators were approximately unbiased and achieved approximately nominal coverage (Table \ref{tab:tab:bias_var_cov}). The estimators that assume partial principal ignorability tended to have smaller variance and therefore smaller MSE, with the smallest variance achieved by the semiparametric estimator that leveraged both assumptions. As expected, when either principal ignorability or the exclusion restriction was violated, the estimators that relied on the violated assumption exhibited high bias and poor coverage. When both assumptions were violated, all estimators failed to deliver proper inference. Overall, this set of simulations confirmed our theorems establishing asymptotic validity of estimators under our stated assumptions.

\begin{table}[t]
\centering
\fontsize{9}{11}\selectfont
\caption{Performance of bounds across settings and sample sizes. Bias and coverage for $\ell$ and $u$ refer to how well point estimates approximate and confidence intervals respectively cover the true theoretical value of the bound. Coverage for the effect refers to the proportion of simulations where the true effect was in the interval $\ell_n, u_n$. The median width and interquartile range (IQR) for this width is also shown.}
\begin{tabular}[t]{llcccccc}
\toprule
\multicolumn{2}{c}{} & \multicolumn{2}{c}{$n^{1/2} \times$ Bias} & \multicolumn{3}{c}{Coverage} & \multicolumn{1}{c}{} \\
\cmidrule(l{3pt}r{3pt}){3-4} \cmidrule(l{3pt}r{3pt}){5-7}
Setting & $n$ & $\ell$ & $u$ & $\ell$ & $u$ & Effect & Med. Width (IQR) \\
\midrule
\multirow{2}{*}{\textbf{PI and ER satisfied}}
  & 500  & 0.019  & 0.091  & 0.939 & 0.941 & 0.993 & 0.31 (0.28, 0.34) \\
  & 4000 & -0.082 & -0.013 & 0.946 & 0.948 & 1.000 & 0.31 (0.30, 0.32) \\
\midrule
\multirow{2}{*}{\textbf{PI satisfied, ER violated}}
  & 500  & 0.059  & 0.084  & 0.944 & 0.939 & 0.928 & 0.22 (0.20, 0.24) \\
  & 4000 & -0.005 & -0.006 & 0.946 & 0.950 & 1.000 & 0.22 (0.21, 0.23) \\
\midrule
\multirow{2}{*}{\textbf{PI violated, ER satisfied}}
  & 500  & 0.343  & 0.080  & 0.927 & 0.943 & 0.996 & 0.28 (0.26, 0.31) \\
  & 4000 & 0.092  & -0.043 & 0.949 & 0.944 & 1.000 & 0.29 (0.28, 0.30) \\
\midrule
\multirow{2}{*}{\textbf{PI and ER violated}}
  & 500  & 0.043  & 0.065  & 0.942 & 0.934 & 0.901 & 0.15 (0.13, 0.17) \\
  & 4000 & -0.005 & 0.004  & 0.946 & 0.945 & 1.000 & 0.15 (0.14, 0.15) \\
\bottomrule
\end{tabular}
\label{tab:sim1_bounds_unadjusted_combined}
\end{table}

\begin{table}[!h]
\centering
\caption{\label{tab:sim1_bounds_n500_rearranged}Bias, Coverage, and Bound Width ($n = 500$)}
\fontsize{9}{11}\selectfont
\begin{tabular}[t]{llcccccc}
\toprule
\multicolumn{2}{c}{} & \multicolumn{2}{c}{$n^{1/2} \times$ Bias} & \multicolumn{3}{c}{Coverage} & \multicolumn{1}{c}{} \\
\cmidrule(l{3pt}r{3pt}){3-4} \cmidrule(l{3pt}r{3pt}){5-7}
Setting & Covariates & $\ell$ & $u$ & $\ell$ & $u$ & Effect & Med. Width (IQR) \\
\midrule
\multirow{8}{*}{\textbf{PI and ER satisfied}}
  & Unadjusted             & 0.019 & 0.091  & 0.939 & 0.941 & 0.993 & 0.31 (0.28, 0.34) \\
  & $X_1$                  & 0.049 & 0.092  & 0.941 & 0.942 & 0.994 & 0.31 (0.28, 0.34) \\
  & $X_2$                  & 0.166 & 0.092  & 0.937 & 0.940 & 0.994 & 0.31 (0.28, 0.34) \\
  & $X_3$                  & 0.161 & -0.074 & 0.945 & 0.946 & 0.993 & 0.30 (0.27, 0.32) \\
  & $X_1$,$X_2$            & 0.329 & 0.078  & 0.932 & 0.947 & 0.993 & 0.30 (0.27, 0.33) \\
  & $X_1$,$X_3$            & 0.290 & -0.011 & 0.932 & 0.950 & 0.988 & 0.27 (0.25, 0.30) \\
  & $X_2$,$X_3$            & 0.241 & -0.224 & 0.936 & 0.939 & 0.991 & 0.28 (0.25, 0.30) \\
  & $X_1$,$X_2$,$X_3$      & 0.407 & -0.149 & 0.907 & 0.948 & 0.983 & 0.25 (0.23, 0.28) \\
\midrule
\multirow{8}{*}{\textbf{PI satisfied, ER violated}}
  & Unadjusted             & 0.059 & 0.084 & 0.944 & 0.939 & 0.928 & 0.22 (0.20, 0.24) \\
  & $X_1$                  & 0.121 & 0.085 & 0.949 & 0.945 & 0.926 & 0.22 (0.20, 0.24) \\
  & $X_2$                  & 0.153 & 0.084 & 0.943 & 0.942 & 0.931 & 0.22 (0.20, 0.24) \\
  & $X_3$                  & 0.272 & 0.081 & 0.941 & 0.941 & 0.929 & 0.21 (0.19, 0.23) \\
  & $X_1$,$X_2$            & 0.299 & 0.088 & 0.921 & 0.943 & 0.928 & 0.21 (0.19, 0.23) \\
  & $X_1$,$X_3$            & 0.367 & 0.058 & 0.925 & 0.947 & 0.928 & 0.20 (0.18, 0.22) \\
  & $X_2$,$X_3$            & 0.260 & 0.077 & 0.930 & 0.944 & 0.931 & 0.20 (0.18, 0.22) \\
  & $X_1$,$X_2$,$X_3$      & 0.454 & -0.004 & 0.894 & 0.944 & 0.923 & 0.19 (0.17, 0.21) \\
\midrule
\multirow{8}{*}{\textbf{PI violated, ER satisfied}}
  & Unadjusted             & 0.343 & 0.080 & 0.927 & 0.943 & 0.996 & 0.28 (0.26, 0.31) \\
  & $X_1$                  & 0.518 & 0.082 & 0.905 & 0.946 & 0.994 & 0.27 (0.25, 0.30) \\
  & $X_2$                  & 0.268 & 0.081 & 0.940 & 0.945 & 0.990 & 0.27 (0.24, 0.29) \\
  & $X_3$                  & 0.210 & 0.069 & 0.941 & 0.946 & 0.993 & 0.26 (0.24, 0.29) \\
  & $X_1$,$X_2$            & 0.487 & 0.087 & 0.909 & 0.947 & 0.984 & 0.26 (0.24, 0.28) \\
  & $X_1$,$X_3$            & 0.487 & 0.087 & 0.909 & 0.947 & 0.984 & 0.26 (0.24, 0.28) \\
  & $X_2$,$X_3$            & 0.487 & 0.087 & 0.909 & 0.947 & 0.984 & 0.26 (0.24, 0.28) \\
  & $X_1$,$X_2$,$X_3$      & 0.487 & 0.087 & 0.909 & 0.947 & 0.984 & 0.26 (0.24, 0.28) \\
\midrule
\multirow{8}{*}{\textbf{PI and ER violated}}
  & Unadjusted             & 0.043 & 0.065 & 0.942 & 0.934 & 0.901 & 0.15 (0.13, 0.17) \\
  & $X_1$                  & 0.045 & 0.067 & 0.945 & 0.935 & 0.903 & 0.15 (0.13, 0.17) \\
  & $X_2$                  & 0.046 & 0.066 & 0.944 & 0.932 & 0.898 & 0.15 (0.13, 0.17) \\
  & $X_3$                  & 0.064 & 0.063 & 0.942 & 0.934 & 0.894 & 0.15 (0.13, 0.16) \\
  & $X_1$,$X_2$            & 0.079 & 0.071 & 0.943 & 0.934 & 0.895 & 0.15 (0.13, 0.16) \\
  & $X_1$,$X_3$            & 0.116 & 0.055 & 0.941 & 0.937 & 0.887 & 0.14 (0.13, 0.16) \\
  & $X_2$,$X_3$            & 0.127 & 0.063 & 0.939 & 0.937 & 0.885 & 0.14 (0.13, 0.16) \\
  & $X_1$,$X_2$,$X_3$      & 0.233 & 0.011 & 0.931 & 0.937 & 0.867 & 0.14 (0.12, 0.15) \\
\bottomrule
\end{tabular}
\end{table}

\begin{table}[!h]
\centering
\caption{\label{tab:sim1_bounds_n4000_rearranged}Bias, Coverage, and Bound Width ($n = 4000$)}
\fontsize{9}{11}\selectfont
\begin{tabular}[t]{llcccccc}
\toprule
\multicolumn{2}{c}{} & \multicolumn{2}{c}{$n^{1/2} \times$ Bias} & \multicolumn{3}{c}{Coverage} & \multicolumn{1}{c}{} \\
\cmidrule(l{3pt}r{3pt}){3-4} \cmidrule(l{3pt}r{3pt}){5-7}
Setting & Covariates & $\ell$ & $u$ & $\ell$ & $u$ & Effect & Med. Width (IQR) \\
\midrule
\multirow{8}{*}{\textbf{PI and ER satisfied}}
  & Unadjusted             & -0.082 & -0.013 & 0.946 & 0.948 & 1.000 & 0.31 (0.30, 0.32) \\
  & $X_1$                  & -0.084 & -0.014 & 0.944 & 0.946 & 1.000 & 0.31 (0.30, 0.32) \\
  & $X_2$                  & -0.068 & -0.012 & 0.949 & 0.947 & 1.000 & 0.31 (0.30, 0.32) \\
  & $X_3$                  & -0.072 & -0.029 & 0.947 & 0.945 & 1.000 & 0.31 (0.30, 0.32) \\
  & $X_1$,$X_2$            & 0.107  & -0.015 & 0.943 & 0.944 & 1.000 & 0.31 (0.30, 0.32) \\
  & $X_1$,$X_3$            & -0.011 & -0.043 & 0.953 & 0.939 & 1.000 & 0.29 (0.28, 0.30) \\
  & $X_2$,$X_3$            & 0.002  & -0.152 & 0.950 & 0.944 & 1.000 & 0.30 (0.29, 0.31) \\
  & $X_1$,$X_2$,$X_3$      & 0.089  & -0.064 & 0.949 & 0.943 & 1.000 & 0.27 (0.27, 0.28) \\
\midrule
\multirow{8}{*}{\textbf{PI satisfied, ER violated}}
  & Unadjusted             & -0.005 & -0.006 & 0.946 & 0.950 & 1.000 & 0.22 (0.21, 0.23) \\
  & $X_1$                  & -0.006 & -0.007 & 0.947 & 0.950 & 1.000 & 0.22 (0.21, 0.23) \\
  & $X_2$                  & 0.003  & -0.006 & 0.946 & 0.951 & 1.000 & 0.22 (0.21, 0.23) \\
  & $X_3$                  & 0.126  & -0.009 & 0.951 & 0.949 & 1.000 & 0.22 (0.21, 0.22) \\
  & $X_1$,$X_2$            & 0.033  & -0.009 & 0.951 & 0.950 & 1.000 & 0.22 (0.21, 0.23) \\
  & $X_1$,$X_3$            & 0.197  & -0.009 & 0.945 & 0.948 & 1.000 & 0.21 (0.21, 0.22) \\
  & $X_2$,$X_3$            & 0.002  & -0.010 & 0.949 & 0.950 & 1.000 & 0.21 (0.20, 0.22) \\
  & $X_1$,$X_2$,$X_3$      & 0.102  & -0.014 & 0.950 & 0.951 & 1.000 & 0.21 (0.20, 0.22) \\
\midrule
\multirow{8}{*}{\textbf{PI violated, ER satisfied}}
  & Unadjusted             & 0.092  & -0.043 & 0.949 & 0.944 & 1.000 & 0.29 (0.28, 0.30) \\
  & $X_1$                  & 0.250  & -0.045 & 0.936 & 0.945 & 1.000 & 0.29 (0.28, 0.30) \\
  & $X_2$                  & -0.025 & -0.043 & 0.959 & 0.942 & 1.000 & 0.28 (0.27, 0.29) \\
  & $X_3$                  & -0.066 & -0.044 & 0.952 & 0.947 & 1.000 & 0.27 (0.26, 0.28) \\
  & $X_1$,$X_2$            & 0.058  & -0.046 & 0.954 & 0.948 & 1.000 & 0.27 (0.27, 0.28) \\
  & $X_1$,$X_3$            & 0.058  & -0.046 & 0.954 & 0.948 & 1.000 & 0.27 (0.27, 0.28) \\
  & $X_2$,$X_3$            & 0.058  & -0.046 & 0.954 & 0.948 & 1.000 & 0.27 (0.27, 0.28) \\
  & $X_1$,$X_2$,$X_3$      & 0.058  & -0.046 & 0.954 & 0.948 & 1.000 & 0.27 (0.27, 0.28) \\
\midrule
\multirow{8}{*}{\textbf{PI and ER violated}}
  & Unadjusted             & -0.005 & 0.004  & 0.946 & 0.945 & 1.000 & 0.15 (0.14, 0.15) \\
  & $X_1$                  & -0.006 & 0.003  & 0.948 & 0.945 & 1.000 & 0.15 (0.14, 0.15) \\
  & $X_2$                  & -0.006 & 0.003  & 0.944 & 0.946 & 1.000 & 0.15 (0.14, 0.15) \\
  & $X_3$                  & -0.008 & 0.002  & 0.951 & 0.946 & 1.000 & 0.15 (0.14, 0.15) \\
  & $X_1$,$X_2$            & -0.009 & 0.001  & 0.943 & 0.946 & 1.000 & 0.15 (0.14, 0.15) \\
  & $X_1$,$X_3$            & -0.008 & 0.001  & 0.949 & 0.945 & 1.000 & 0.15 (0.14, 0.15) \\
  & $X_2$,$X_3$            & -0.008 & 0.000  & 0.948 & 0.947 & 1.000 & 0.15 (0.14, 0.15) \\
  & $X_1$,$X_2$,$X_3$      & 0.021  & -0.001 & 0.945 & 0.946 & 1.000 & 0.15 (0.14, 0.15) \\
\bottomrule
\end{tabular}
\end{table}

\begin{table}[!ht]
\centering
\caption{Performance of one-step estimators. Var. = variance; MSE = mean squared error; Cov. = coverage of nominal 95\% confidence interval. $^1$ scaled by $n^{1/2}$; $^2$ scaled by $n$; $^3$ computed on the log scale}
\begin{tabular}[t]{lcccccclcc}
\toprule
\multicolumn{2}{c}{ } & \multicolumn{4}{c}{Additive Scale} & \multicolumn{4}{c}{Multiplicative Scale} \\
\cmidrule(l{3pt}r{3pt}){3-6} \cmidrule(l{3pt}r{3pt}){7-10}
Method & $n$ & Bias$^1$ & Var.$^2$ & MSE$^2$ & Cov. & Bias$^{1,3}$ & Var.$^{2,3}$ & MSE$^{2,3}$  & Cov. \\
\midrule
\addlinespace[0.3em]
\multicolumn{10}{l}{\textbf{PI and ER satisfied}}\\
\multirow{2}{*}{$\psi_{1,\text{PI},n}^+ - \psi_{0,n}^+$} & 500 & 0.056 & 1.127 & 1.129 & 0.944 & 0.178 & 8.613 & 8.635 & 0.952\\
 & 4000 & -0.034 & 1.089 & 1.089 & 0.951 & -0.085 & 8.272 & 8.271 & 0.952\\
 \midrule
\multirow{2}{*}{$\psi_{1,\text{ER},n}^+ - \psi_{0,n}^+$} & 500 & 0.077 & 1.737 & 1.741 & 0.933 & 0.148 & 12.523 & 12.532 & 0.940\\
 & 4000 & -0.012 & 1.628 & 1.626 & 0.946 & -0.058 & 11.546 & 11.538 & 0.947\\
 \midrule
\multirow{2}{*}{$\psi_{1,\cdot,n}^+ - \psi_{0,n}^+$} & 500 & 0.046 & 1.029 & 1.030 & 0.953 & 0.163 & 7.978 & 7.997 & 0.955\\
 & 4000 & -0.040 & 1.024 & 1.024 & 0.949 & -0.097 & 7.878 & 7.879 & 0.949\\
 \midrule
\addlinespace[0.3em]
\multicolumn{10}{l}{\textbf{PI satisfied, ER violated}}\\
\multirow{2}{*}{$\psi_{1,\text{PI},n}^+ - \psi_{0,n}^+$} & 500 & 0.054 & 0.984 & 0.986 & 0.947 & 0.144 & 11.925 & 11.934 & 0.952\\
 & 4000 & -0.033 & 0.944 & 0.944 & 0.950 & -0.107 & 10.826 & 10.827 & 0.950\\
 \midrule
\multirow{2}{*}{$\psi_{1,\text{ER},n}^+ - \psi_{0,n}^+$} & 500 & -1.577 & 1.525 & 4.012 & 0.747 & -7.804 & 39.242 & 100.109 & 0.850\\
 & 4000 & -4.664 & 1.410 & 23.165 & 0.031 & -21.435 & 30.290 & 489.724 & 0.019\\
 \midrule
\multirow{2}{*}{$\psi_{1,\cdot,n}^+ - \psi_{0,n}^+$} & 500 & 0.741 & 0.945 & 1.493 & 0.897 & 2.692 & 9.988 & 17.224 & 0.877\\
 & 4000 & 2.025 & 0.932 & 5.032 & 0.476 & 7.436 & 9.445 & 64.732 & 0.339\\
 \midrule
\addlinespace[0.3em]
\multicolumn{10}{l}{\textbf{PI violated, ER satisfied}}\\
\multirow{2}{*}{$\psi_{1,\text{PI},n}^+ - \psi_{0,n}^+$} & 500 & 1.044 & 1.047 & 2.136 & 0.823 & 3.704 & 10.753 & 24.461 & 0.791\\
 & 4000 & 2.744 & 0.987 & 8.518 & 0.222 & 9.867 & 9.763 & 107.106 & 0.127\\
 \midrule
\multirow{2}{*}{$\psi_{1,\text{ER},n}^+ - \psi_{0,n}^+$} & 500 & 0.069 & 1.624 & 1.627 & 0.943 & -0.012 & 22.791 & 22.768 & 0.953\\
 & 4000 & -0.041 & 1.492 & 1.493 & 0.945 & -0.203 & 19.267 & 19.289 & 0.951\\
 \midrule
\multirow{2}{*}{$\psi_{1,\cdot,n}^+ - \psi_{0,n}^+$} & 500 & 1.489 & 0.987 & 3.202 & 0.698 & 5.151 & 9.416 & 35.937 & 0.620\\
 & 4000 & 4.089 & 0.957 & 17.679 & 0.018 & 14.188 & 8.865 & 210.153 & 0.007\\
 \midrule
\addlinespace[0.3em]
\multicolumn{10}{l}{\textbf{PI and ER violated}}\\
\multirow{2}{*}{$\psi_{1,\text{PI},n}^+ - \psi_{0,n}^+$} & 500 & 0.545 & 0.979 & 1.275 & 0.915 & 2.631 & 15.707 & 22.615 & 0.877\\
 & 4000 & 1.381 & 0.918 & 2.824 & 0.719 & 7.157 & 13.382 & 64.593 & 0.523\\
 \midrule
\multirow{2}{*}{$\psi_{1,\text{ER},n}^+ - \psi_{0,n}^+$} & 500 & -1.593 & 1.510 & 4.047 & 0.726 & -13.057 & 128.394 & 298.738 & 0.987\\
 & 4000 & -4.649 & 1.357 & 22.965 & 0.025 & -32.704 & 71.531 & 1141.019 & 0.007\\
 \midrule
\multirow{2}{*}{$\psi_{1,\cdot,n}^+ - \psi_{0,n}^+$} & 500 & 1.463 & 0.959 & 3.100 & 0.685 & 6.759 & 12.061 & 57.733 & 0.470\\
 & 4000 & 4.112 & 0.925 & 17.834 & 0.015 & 19.237 & 10.855 & 380.893 & 0.000\\
 \midrule
\bottomrule
\end{tabular}
\label{tab:tab:bias_var_cov}
\end{table}

\subsection{Additional details and results for ``Comparing power of estimands in realistic setting'' simulation}

\subsubsection{Data generation details}
$X_1$ was generated as a Bernoulli$(0.5)$ variable, $X_2$ was generated from a normal distribution with mean $-0.97$ and standard deviation $0.90$, and $X_3$ was generated from a negative binomial distribution with mean $5.26$ and dispersion chosen to match observed variability, truncated to have minimum value one. Conditional on $X$, principal stratum membership was generated using a multinomial logistic model. Probabilities parameterizing this model were defined by letting $g_{\text{D}}(X)=-1.2+0.81X_1+0.18X_2+0.06X_3+\delta_{\text{P}}$ and $g_{\text{I}}(X)=1.5-0.30X_1+0.10X_2-0.08X_3+\delta_{\text{I}}+\delta_{\text{P}}$. Principal stratum probabilities were defined by softmax transformation of these linear predictors: $P\{S(1)=1,S(0)=1\mid X\}=\exp\{g_{\text{D}}(X)\}/\{1+\exp\{g_{\text{D}}(X)\}+\exp\{g_{\text{I}}(X)\}\}$, $P\{S(1)=0,S(0)=0\mid X\}=\exp\{g_{\text{I}}(X)\}/\{1+\exp\{g_{\text{D}}(X)\}+\exp\{g_{\text{I}}(X)\}\}$, and $P\{S(1)=0,S(0)=1\mid X\}=1/\{1+\exp\{g_{\text{D}}(X)\}+\exp\{g_{\text{I}}(X)\}\}$. Potential infection outcomes were generated deterministically based on principal stratum membership. 

The parameters $\delta_{\text{I}}$ and $\delta_{\text{P}}$ were used to shift the relative probabilities of the Immune and Protected strata, respectively. We refer to these parameters as \emph{strata composition governing}. 

Binary outcome potential outcomes were generated from stratum- and treatment-specific logistic regression models calibrated to antibiotic use patterns observed in PROVIDE. In the Doomed stratum, antibiotic use risk was specified as $P\{Y(1)=1\mid S(0)=1,S(1)=1,X\}=\text{expit}(-0.70+0.78X_1-1.44X_2+0.49X_3)$, which was estimated from the PROVIDE data set by fitting a logistic regression to the infected vaccinated participants. The probability of the potential outcome in the Doomed under placebo was defined as
$P\{Y(0)=1\mid S(0)=1,S(1)=1,X\}=\text{expit}(\text{logit}\{P(Y(1)=1\mid S(0)=1,S(1)=1,X)\}-\eta_{\text{D}})$,
so that positive values of $\eta_{\text{D}}$ corresponded to larger protective effects of vaccination on the post-infection outcome among Doomed individuals. For individuals in the Protected stratum, outcome risk under placebo was set equal to that of the Doomed stratum,
$P\{Y(0)=1\mid S(0)=1,S(1)=0,X\}=P\{Y(0)=1\mid S(0)=1,S(1)=1,X\}$, so that the partial principal ignorability Assumption S2 required to identify the effect in the Doomed stratum was satisfied (see Section H.1). The probability for antibiotic treatment under vaccine in the Protected stratum was then set to 
$P\{Y(1)=1\mid S(0)=1,S(1)=0,X\}=\text{expit}(\text{logit}\{P(Y(0)=1\mid S(0)=1,S(1)=0,X)\}+\eta_{\text{P}})$,
so that positive values of $\eta_{\text{P}}$ corresponded to larger protective effects in the Protected stratum. For Immune individuals, we set $P\{Y(1)=1\mid S(0)=0,S(1)=0,X\}=P\{Y(1)=1\mid S(0)=1,S(1)=0,X\}$, thereby imposing partial principal ignorability. Finally, we set $P\{Y(0)=1\mid S(0)=0,S(1)=0,X\}=P\{Y(1)=1\mid S(0)=0,S(1)=0,X\}$, thereby imposing the exclusion restriction. Observed infection and outcome were defined as $S=S(Z)$ and $Y=Y(Z)$.

We refer to $\eta_{\text{D}}$ and $\eta_{\text{P}}$ as \emph{post-infection outcome effect governing}. We varied these parameters over a two-dimensional grid and for each grid point, computed the corresponding marginal risk differences within the Doomed and Protected strata using a single large independent Monte Carlo sample of size $10^6$. Tables \ref{tab:sim_params1} and \ref{tab:sim_params2} summarize the parameter settings used to construct the power contour simulations. Table \ref{tab:sim_params1} reports the values of the stratum-composition parameters $\delta_{\text{I}}$ and $\delta_{\text{P}}$, which shift the linear predictors governing principal stratum membership and thereby control the marginal proportion of Immune individuals and the marginal vaccine efficacy against infection. For each setting, the resulting marginal probability $P\{S(1)=0,S(0)=0\}$ and marginal vaccine efficacy were computed using a large independent Monte Carlo sample of size $10^6$, and are reported to document the realized stratum composition underlying each set of contour plots. Table \ref{tab:sim_params2} reports the mapping between the outcome-effect parameters $\eta_{\text{D}}$ and $\eta_{\text{P}}$, which govern the strength of the vaccine effect on the post-infection outcome in the Doomed and Protected principal strata, respectively, and the corresponding marginal risk differences used to label the contour plot axes. These risk differences were computed by averaging stratum-specific potential outcomes over baseline covariates in the same large Monte Carlo sample, ensuring that contour axes are interpretable on the risk-difference scale. Together, these tables provide a complete description of the principal stratum composition settings and outcome-effect magnitudes underlying the power contour analyses.

\begin{table}[ht]
\centering
\caption{Principal stratum composition settings used in the power contour simulations. 
Parameters $\delta_{\text{I}}$ and $\delta_{\text{P}}$ shift the linear predictors for the Immune stratum and for the non-Protected strata, respectively, inducing changes in the marginal proportion Immune and in marginal vaccine efficacy against infection. 
Marginal quantities were computed from a large independent Monte Carlo sample of size $10^6$.}
\begin{tabular}{cccc}
\hline
$P\{S(1)=0,S(0)=0\}$ & VE & $\delta_{\text{I}}$ & $\delta_{\text{P}}$ \\
\hline
0.40 & 0.66 & $-0.73$ & $-0.12$ \\
0.60 & 0.66 & $0.16$  & $-0.19$ \\
0.80 & 0.66 & $1.11$  & $-0.11$ \\
0.60 & 0.50 & $-0.29$ & $0.55$  \\
0.60 & 0.85 & $0.95$  & $-1.21$ \\
\hline
\end{tabular}
\label{tab:sim_params1}
\end{table}

\begin{table}[ht]
\centering
\caption{Mapping between outcome-effect parameters and marginal risk differences used as contour plot axes. 
For each value of $\eta_{\text{D}}$ and $\eta_{\text{P}}$, marginal risk differences were computed within the Doomed and Protected principal strata, respectively, by averaging over baseline covariates in a large independent Monte Carlo sample of size $10^6$.}
\begin{tabular}{cccc}
\hline
$\eta_{\text{D}}$ & RD$_{\text{D}}$ & $\eta_{\text{P}}$ & RD$_{\text{P}}$ \\
\hline
$0.0$  & $0.00$ & $0.0$  & $0.00$ \\
$-0.5$ & $-0.06$ & $-0.5$ & $-0.08$ \\
$-1.0$ & $-0.12$ & $-1.0$ & $-0.15$ \\
$-1.5$ & $-0.18$ & $-1.5$ & $-0.23$ \\
$-2.0$ & $-0.24$ & $-2.0$ & $-0.30$ \\
$-2.5$ & $-0.29$ & $-2.5$ & $-0.36$ \\
$-3.0$ & $-0.33$ & $-3.0$ & $-0.41$ \\
\hline
\end{tabular}
\label{tab:sim_params2}
\end{table}

{\color{changed}
\subsection{Impact of cross-fitting with machine learned nuisance parameters}

We generated a datasets of size $n \in \{500,4000\}$ using a similar design as in \ref{sec:dgp_generic}, except in this setting we had 10 covariates and complex, non-linear relationships between covariates and outcomes. Baseline covariates $X=(X_1,\dots, X_10)$ were generated independently with $X_j \sim \text{Bernoulli}(0.5)$ for $j=1,\dots,5$ and $X_j \sim \text{Uniform}(0,1)$ for $j = 6, \dots, 10$. 

We set $P\{S(1)=1,S(0)=1\mid X\}=\text{expit}\{-1 + 0.25 X_1 - 0.25 X_1 X_2 X_3 + 0.5*I(X_6 > 0.25) - 0.5I(X_7 > 0.1, X_8 < 0.9)\}$, 
$P\{S(1)=0,S(0)=0\mid X\}=\text{expit}\{ -1.5 + 0.5X_3 - 0.5 X_2 X_3 + 0.5 I(X_7 < 0.75) - 0.5 X_2I(X_5 < 0.5)\}$, and the conditional probability of $S(1)=0, S(0)=1$ given by one minus these two conditional probabilities. Infection potential outcomes were generated deterministically based on stratum membership. 

Binary outcome potential outcomes were generated from stratum- and treatment-specific logistic regression models. 
For individuals in the Doomed stratum, $P\{Y(0)=1 \mid S(0) = 1, S(1) = 1, X \} = \text{expit}(-1 + 0.5 X_1 X_2 - X_2I(X_6 < 0.25) + 0.5I(X_7 > 0.5))$ and $P\{Y(1)=1 \mid S(0) = 1, S(1) = 1, X\}=\text{expit}[\text{logit}\{P\{Y(0)=1\mid S(0) = 1, S(1) = 1, X\}\}+ 0.1]$, yielding a small effect of vaccine in the Doomed stratum. For Immune individuals, $P\{Y(1)=1\mid S(0) = 0, S(1) = 0, X\} = P\{Y(0)=1\mid S(0) = 0, S(1) = 0, X\} = \text{expit}\{ -0.5 + 0.5 X_2 X_3 - 0.5 X_1 I(X_7 > 0.75) + 0.25 X_1 \}$ so the exclusion restriction was true. For Protected individuals, we set $P\{Y(0)=1\mid S(0) = 1, S(1) = 0, X\} = P\{Y(0)=1\mid S(0) = 1, S(1) = 1, X\}$ and $P\{Y(1) = 1 \mid S(0) = 1, S(1) = 0, X\} = P\{Y(1)=1\mid S(0) = 0, S(1) = 0, X\}$. Vaccine assignment $Z$ was generated according to $P(Z=1\mid X )=\mbox{expit}\{-0.25 - 0.25 X_2 X_4 + 0.5 X_2 I(X_6 > 0.5) \}$. 

Observed infection and outcome were then set as $S=S(Z)$ and $Y=Y(Z)$. 

The true value of counterfactual means in the Naturally Infected was computed based on a single independent Monte Carlo sample of size 10,000,000 generated from the same data-generating process, leveraging the full set of potential outcomes. For this simulation $\E\{ Y(1) \mid S(0) = 1 \} = 0.398$ and $\E\{ Y(0) \mid S(0) = 1 \} = .330$, resulting in an additive effect of 0.068 and a multiplicative effect of 1.207.

\begin{table}[h!]
\centering
\caption{\label{tab:cross_fit_pi_er_satisfied}Scaled bias, variance, and MSE with coverage for estimators with and without cross-fitting when partial principal ignorability and exclusion restrictions are satisfied.}
\centering
\begin{tabular}[t]{llcccccclcc}
\toprule
\multicolumn{3}{c}{ } & \multicolumn{4}{c}{Additive Scale} & \multicolumn{4}{c}{Multiplicative Scale} \\
\cmidrule(l{3pt}r{3pt}){4-7} \cmidrule(l{3pt}r{3pt}){8-11}
Method & Cross-fit & $n$ & $\sqrt{n}$ Bias & $n$ Var. & $n$ MSE & Cov. & $\sqrt{n}$ Bias & $n$ Var. & $n$ MSE  & Cov. \\
\midrule
\addlinespace[0.3em]
\multirow{2}{*}{$\psi_{1,\text{PI},n}^+ - \psi_{0,n}^+$} & Yes & 500 & 0.019 & 1.287 & 1.286 & 0.951 & 0.093 & 9.973 & 9.971 & 0.956\\
 & No & 500 & 0.022 & 1.189 & 1.188 & 0.859 & 0.080 & 9.192 & 9.190 & 0.862\\
\midrule
\multirow{2}{*}{$\psi_{1,\text{ER},n}^+ - \psi_{0,n}^+$} & Yes & 500 & 0.008 & 2.000 & 1.999 & 0.947 & -0.040 & 14.904 & 14.891 & 0.948\\
 & No & 500 & -0.132 & 1.593 & 1.609 & 0.837 & -0.381 & 11.922 & 12.056 & 0.839\\
\midrule
\multirow{2}{*}{$\psi_{1,\cdot,n}^+ - \psi_{0,n}^+$} & Yes & 500 & 0.043 & 1.169 & 1.170 & 0.955 & 0.170 & 9.209 & 9.229 & 0.952\\
 & No & 500 & 0.047 & 1.125 & 1.126 & 0.865 & 0.151 & 8.754 & 8.769 & 0.867\\
\midrule
\midrule
\addlinespace[0.3em]
\multirow{2}{*}{$\psi_{1,\text{PI},n}^+ - \psi_{0,n}^+$} & Yes & 4000 & 0.044 & 1.120 & 1.121 & 0.943 & 0.148 & 8.634 & 8.647 & 0.946\\
 & No & 4000 & 0.017 & 1.110 & 1.109 & 0.916 & 0.063 & 8.529 & 8.524 & 0.917\\
\midrule
\multirow{2}{*}{$\psi_{1,\text{ER},n}^+ - \psi_{0,n}^+$} & Yes & 4000 & 0.081 & 1.723 & 1.728 & 0.951 & 0.211 & 12.448 & 12.480 & 0.945\\
 & No & 4000 & -0.018 & 1.692 & 1.690 & 0.910 & -0.056 & 12.248 & 12.239 & 0.906\\
\midrule
\multirow{2}{*}{$\psi_{1,\cdot,n}^+ - \psi_{0,n}^+$} & Yes & 4000 & 0.027 & 1.047 & 1.047 & 0.949 & 0.110 & 8.179 & 8.183 & 0.953\\
 & No & 4000 & 0.009 & 1.039 & 1.038 & 0.919 & 0.046 & 8.084 & 8.078 & 0.922\\
\midrule
\bottomrule
\end{tabular}
\end{table}

The cross-fitted estimators achieved coverage close to the nominal 95\% level across methods, sample sizes, and both additive and multiplicative scales. In contrast, the non-cross-fitted estimators had substantially lower coverage, particularly at n=500. Estimators were similar in terms of bias, variance, and MSE with cross-fitted estimators exhibiting marginally higher bias and variance.

}

\newpage

\section{PROVDE: Additional results}

\subsection{Covariate adjusted results}

Table \ref{tab:tab:realdata_bounds} shows unadjusted and covariate-adjusted bounds for the Naturally Infected effects in the PROVIDE analysis. There were no covariates that led meaningfully narrower bounds for the effect.

\begin{table}[htb]
\caption{Unadjusted and covariate-adjusted bounds for additive and multiplicative Naturally Infected effects of rotavirus vaccine on antibiotic prescribing within one year of vaccination.}
\centering
\fontsize{9}{11}\selectfont
\begin{tabular}[t]{lcc}
\toprule
\multicolumn{1}{c}{} & \multicolumn{1}{c}{Lower bound (95\% CI)} & \multicolumn{1}{c}{Upper bound (95\% CI)} \\
\midrule

\addlinespace[0.3em]
\multicolumn{3}{l}{\textbf{Additive}}\\

\addlinespace[0.3em]
\multicolumn{3}{l}{\textbf{Unadjusted Bound}}\\
\hspace{1em}Unadjusted & -0.425 (-0.562, -0.260) & 0.086 (0.024, 0.159) \\

\addlinespace[0.3em]
\multicolumn{3}{l}{\textbf{Covariate-Adjusted Bound}}\\
\hspace{1em}Gender & -0.420 (-0.540, -0.257) & 0.086 (0.022, 0.162) \\
\hspace{1em}Enrollment HAZ (bin) & -0.431 (-0.561, -0.263) & 0.085 (0.022, 0.162) \\
\hspace{1em}Household size (bin) & -0.429 (-0.554, -0.258) & 0.085 (0.022, 0.161) \\
\hspace{1em}Gender $\times$ HAZ & -0.426 (-0.560, -0.258) & 0.084 (0.020, 0.159) \\
\hspace{1em}Gender $\times$ Household & -0.426 (-0.564, -0.258) & 0.086 (0.023, 0.160) \\
\hspace{1em}HAZ $\times$ Household & -0.424 (-0.557, -0.252) & 0.087 (0.024, 0.163) \\
\hspace{1em}All interactions & -0.424 (-0.563, -0.257) & 0.085 (0.022, 0.161) \\

\addlinespace[0.8em]
\midrule

\addlinespace[0.3em]
\multicolumn{3}{l}{\textbf{Multiplicative}}\\

\addlinespace[0.3em]
\multicolumn{3}{l}{\textbf{Unadjusted Bound}}\\
\hspace{1em}Unadjusted & 0.524 (0.381, 0.703) & 1.096 (1.025, 1.191) \\

\addlinespace[0.3em]
\multicolumn{3}{l}{\textbf{Covariate-Adjusted Bound}}\\
\hspace{1em}Gender & 0.528 (0.405, 0.710) & 1.096 (1.024, 1.198) \\
\hspace{1em}Enrollment HAZ (bin) & 0.520 (0.388, 0.698) & 1.095 (1.023, 1.197) \\
\hspace{1em}Household size (bin) & 0.521 (0.388, 0.702) & 1.095 (1.024, 1.196) \\
\hspace{1em}Gender $\times$ HAZ & 0.524 (0.388, 0.703) & 1.094 (1.021, 1.191) \\
\hspace{1em}Gender $\times$ Household & 0.523 (0.385, 0.707) & 1.096 (1.024, 1.197) \\
\hspace{1em}HAZ $\times$ Household & 0.525 (0.389, 0.712) & 1.098 (1.025, 1.200) \\
\hspace{1em}All interactions & 0.525 (0.384, 0.703) & 1.096 (1.023, 1.195) \\

\bottomrule
\end{tabular}
\label{tab:tab:realdata_bounds}
\end{table}

\subsection{Sensitivity analysis}

In our PROVIDE sensitivity analysis, we let $\epsilon$ range from 0.55 to 2.2, the former of which was the value that led to the point estimate for the effect approximately equaling the estimated upper bound; the latter was the largest value deemed clinically plausible. For each $\epsilon$, we estimated the sensitivity analysis parameter and plotted the implied estimates as a function of $\epsilon$ along with pointwise 95\% confidence intervals. We found that small positive values of $\epsilon$ led to evidence of positive vaccine effect, indicating that if children in the Immune strata were slightly more likely than children in the Protected to receive antibiotics when vaccinated, then we would have evidence of a positive vaccine effect. However, it may be more realistic to assume that $\epsilon < 1$ since there may be shared exposure pathways and susceptibility factors for rotavirus acquisition as with other causes of diarrhea for which antibiotics may be prescribed. Thus, children who are ``Immune'' with respect to infection with rotavirus, may also be at lower risk for acquisition of other diarrhea-causing pathogens and therefore less likely to receive antibiotics than children who are ``Protected'' with respect to rotavirus.

The results for $\epsilon = 1$ in this analysis differ very slightly from those reported in the main text for point identification. To ensure that the sensitivity parameter estimate is well defined, we must ensure that monotonicity holds in our estimates of $\rho_1$ and $\rho_0$ for each $x$. If not, then it is possible for terms in denominators of the sensitivity parameter to evaluate to 0. In the main analysis, we did not enforce any monotonicity requirement in our nuisance estimation strategies. Here, we estimated $\rho_0$ and $\rho_1$ using a main terms logistic regression model that regressed $S$ on $Z$ and $X$. Because the coefficient associated with $Z$ was negative, monotonicity held for this estimator and the sensitivity analysis could proceed. Future research will be devoted to developing arbitrary ML estimators that respect monotonicity.

\begin{figure}
    \centering
    \includegraphics[width=0.8\linewidth]{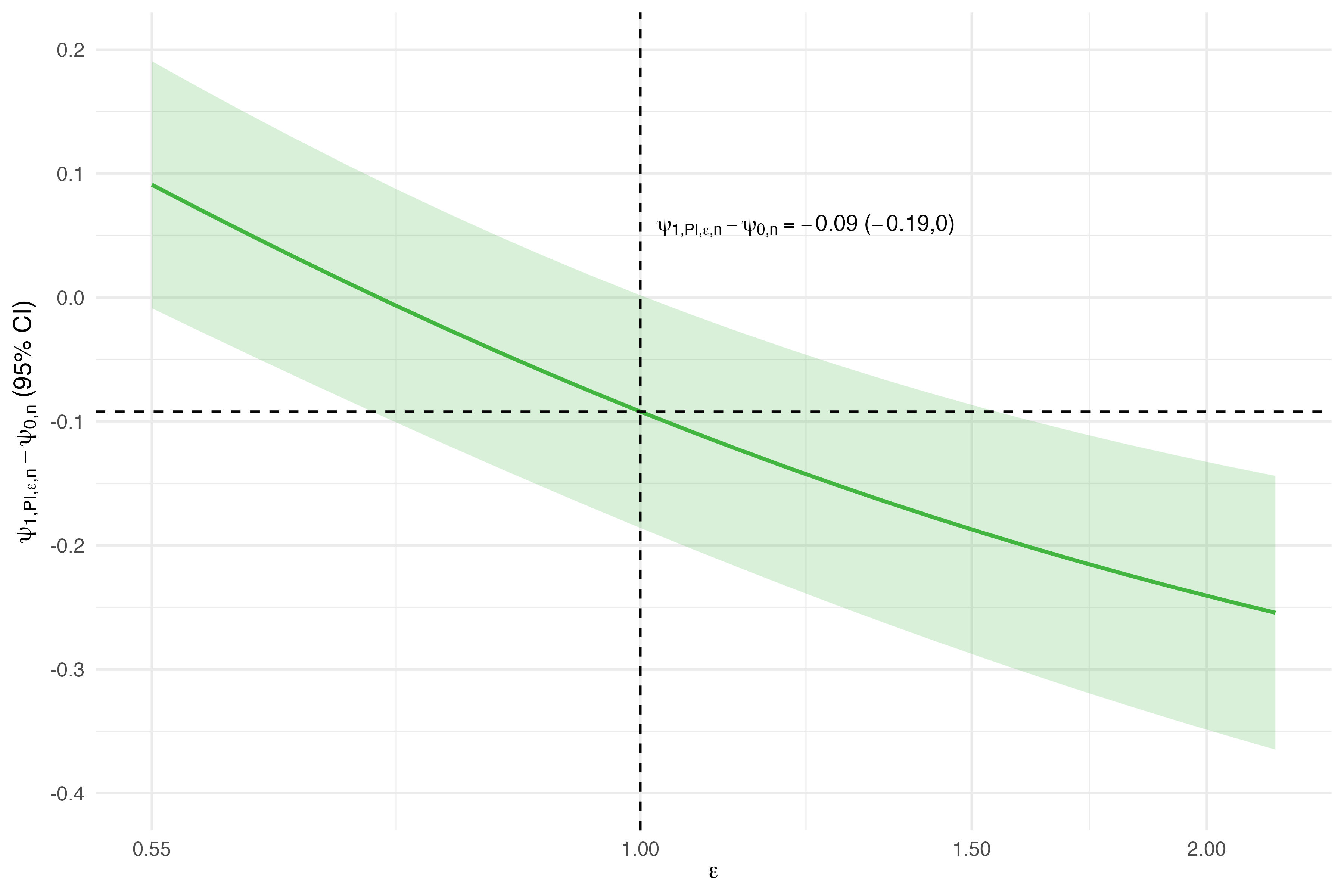}
    \caption{Sensitivity analysis for PROVIDE data. The green line shows the estimated additive effect as a function of the sensitivity parameter $\epsilon$.}
    \label{fig:placeholder}
\end{figure}

\section{Proofs}

\subsection{Proof of Theorem 1}

\begin{proof}
We have that
{\color{changed}
\begin{align*}
\E\{Y(0) \mid S(0) = 1\} &= \sum_x \E\{Y(0) \mid S(0) = 1, X = x \} P\{ X = x \mid S(0) = 1 \} \\
&= \sum_x \E\{Y(0) \mid S(0) = 1, X = x \} \frac{P\{ S(0) = 1 \mid X = x \}}{P\{ S(0) = 1 \}} P(X = x)
\end{align*}
We have that $\E\{Y(0) \mid S(0) = 1, X = x \} = \E( Y \mid Z = 0, S = 1, X = x) = \mu_{01}(x)$ by vaccine randomization and causal consistency. Positivity ensures the identifying functional is well defined. We also have by randomization and consistency that $P\{ S(0) = 1 \mid X = x \} = P(S = 1 \mid Z = 0, X = x) = \rho_0(x)$ and thus that $P\{ S(0) = 1 \} = \bar{\rho}_0$.
}
\end{proof}

\subsection{Proof of Theorem 2}

\begin{proof}
We have that
\begin{align*}
\E\{Y(1) \mid S(0) = 1, S(1) = 1 \} &= \E\{ Y(1) \mid S(1) = 1 \} \\
&= {\color{changed} \sum_x \E\{ Y(1) \mid S(1) = 1, X = x \} P\{ X = x \mid S(1) = 1 \} } \\
&= {\color{changed} \sum_x \E\{ Y \mid Z = 1, S = 1, X = x\} \frac{\rho_1(x)}{\bar{\rho}_1} P(X = x) \ . }
\end{align*}
{\color{changed}The first two equalities follow from monotonicity and the tower rule, respectively. The third follows from randomization, and causal consistency.} Positivity ensures the identifying functional is well defined.

We also have that 
{\color{changed} 
\begin{align*}
P\{ S(1) = 0, S(0) = 0 \} &= P\{ S(0) = 0 \} \\
&= \sum_x P\{ S(0) = 0 \mid X = x \}P(X = x) \\
&= \sum_x P\{ S = 0 \mid Z = 0, X = x\} P(X = x) = 1 - \bar{\rho}_0 \ .
\end{align*}
}
Similarly, {\color{changed}
\begin{align*}
P\{ S(1) = 1, S(0) = 1 \} &= P\{ S(1) = 1 \} \\
&= \sum_x P\{ S(1) = 1 \mid X = x \} P(X = x) \\
&= \sum_x P\{ S = 1 \mid Z = 1, X = x \} P(X = x) = \bar{\rho}_1 \ .
\end{align*}
}
Finally, monotonicity implies that \begin{align*}
    P\{ S(1) = 0, S(0) = 1\} &= 1 - [ P\{ S(1) = 0, S(0) = 0 \} + P\{ S(1) = 1, S(0) = 1 \} ] \\
    &= 1 - \{ 1 - \bar{\rho}_0 + \bar{\rho}_1 \} = \bar{\rho}_0 - \bar{\rho}_1 \ . 
\end{align*}
\end{proof}

For a visual representation, consider Figure S2. From the Figure, we may infer that identification of $\E\{Y(1) \mid S(0) = 1, S(1) = 1\}$ is straightforward as all observed infected vaccinated individuals must be Doomed. We may also infer that the joint distribution of potential infection outcomes is also identified. The fraction of the vaccine arm that is infected $\bar{\rho}_1$ gives the proportion of the population in the Doomed stratum; the fraction of the placebo arm that is uninfected $1 - \bar{\rho}_0$ gives the proportion of the population in the Immune stratum; while one minus the sum of these quantities therefore yields the proportion of the population in the Protected stratum. Thus, the fraction of the Naturally infected who are Doomed is identified by the ratio of infected vaccinated vs. placebo recipients, $\bar{\rho}_1/\bar{\rho}_0$ (dashed line, right side of Figure S2). 

\begin{figure}
\begin{center}
\begin{tikzpicture}
\node[] at (0, 0) (label1) {$Z = 1$};
\node[draw, minimum width=2.8cm, minimum height=1.5cm, align=center, below = 0.05cm of label1] (vaxinf) {$S = 1$ \\[0.5em] Doomed};
\node[draw, minimum width=2.8cm, minimum height=3cm, below=0cm of vaxinf, align = center] (vaxuninf) {\\[0.5em] $S = 0$ \\[0.5em] Immune or \\ Protected};

\node[right=2.5cm of label1] (label2) {$Z = 0$};
\node[draw, minimum width=2.8cm, minimum height=2.3cm, align=center, below = 0.05cm of label2] (plcinf) {$S = 1$ \\[0.5em] Doomed or \\ Protected};
\node[draw, minimum width=2.8cm, minimum height=2.2cm, below=0cm of plcinf, align = center] (plcuninf) {$S = 0$ \\[0.5em] Immune};

\path (plcinf.south west) ++(0,0.8cm) coordinate (entrypoint2);
\draw[dashed] (entrypoint2) -- ++(2.8cm, 0);

\end{tikzpicture}

\caption{
The observed vaccinated (left) and placebo (right) groups can be divided (solid lines) based on observed infection status ($S = 1$, top = infected, $S = 0$, bottom = uninfected). Under monotonicity, these observed strata are mixtures of basic principal strata.
}
\end{center}
\label{fig:bounds}
\end{figure}

\subsection{Proof of Theorem 3}

\begin{proof}
    {\color{changed} Under simple randomization} and consistency, $\E\{ Y(1) \mid S(1) = 0 \} = \E(Y \mid Z = 1, S = 0)$. Thus, the observed strata of vaccinated uninfected participants is a mixture of the Immune and Protected strata with $q = (\bar{\rho}_0 - \bar{\rho}_1) / (1 - \bar{\rho}_1)$ proportion Protected and $(1 - q)$ proportion Immune. Therefore, it must be true that the mean in the Protected is at least as large as $\E(Y \mid Z = 1, S = 0, Y < Y_{\ell})$ and can be no larger than $\E(Y \mid Z = 1, S = 0, Y > Y_u)$.

\end{proof}

\subsection{Proof of Theorem 4}

\begin{proof}
    We have that \begin{align*}
        \E\{ Y(1) \} &= \E\{ Y(1) \mid S(0) = 1, S(1) = 1 \} P\{S(0) = 1, S(1) = 1\} \\
        &\hspace{1em} + \E\{ Y(1) \mid S(0) = 0, S(1) = 0 \} P\{S(0) = 0, S(1) = 0\} \\
        &\hspace{2em} + {\color{changed} \E\{ Y(1) \mid S(0) = 1, S(1) = 0 \} P\{S(0) = 1, S(1) = 0\}} \\
        &= \bar{\mu}_{11} \bar{\rho}_1 + \E\{ Y(1) \mid S(0) = 0, S(1) = 0 \} (1 - \bar{\rho}_0) \\
        &\hspace{2em} + {\color{changed} \E\{ Y(1) \mid S(0) = 1, S(1) = 0 \}} (\bar{\rho}_0 - \bar{\rho}_1) \\
    \end{align*}
The first equality follows from monotonicity and the tower rule; the second equality follows from Theorem 2. We also have that under vaccine randomization, $\E\{ Y(1) \} = {\color{changed} \E\{ \mu_{1\cdot}(X) \} = \E[ \mu_{11}(X) \rho_1(X) + \mu_{10}(X) \{ 1 - \rho_1(X) \} ]}.$ Moreover, under an exclusion restriction \begin{align*}
    \E\{ Y(1) \mid S(0) = 0, S(1) = 0 \} &= \E\{ Y(1) \mid S(0) = 0, S(1) = 0 \} \\
    &= \E\{ Y(0) \mid S(0) = 0, S(1) = 0 \} \\
    &= \E\{ Y(0) \mid S(0) = 0, S(1) = 0 \} = { \color{changed} \E\left\{ \frac{\rho_0(X)}{\bar{\rho}_0} \mu_{00}(X) \right\} } \ .
\end{align*}
That is, if the exclusion restriction holds then outcomes under vaccine in the Immune stratum are no different on average than outcomes under placebo in the Immune stratum. The latter quantity is identified {\color{changed} by a re-weighted average of the outcomes in the placebo uninfecteds, who conditional on $X$, must all belong to the Immune stratum}. Thus, we have argued that \[ {\color{changed}
\bar{\mu}_{1\cdot} =  \E\{ \mu_{11}(X) \rho_1(X) \} + \E[ \mu_{00}(X) \{ 1 - \rho_X \} + \E\{ Y(1) \mid S(0) = 0, S(1) = 1 \} (\bar{\rho}_0 - \bar{\rho}_1) \ . }
\]
Rearranging terms gives the result.
\end{proof}

To understand this result intuitively consider that under randomization, the marginal average post-infection outcome under vaccine is identifiable as $\bar{\mu}_{1\cdot}$. This marginal average decomposes into weighted averages in each of the three basic principal strata. As established in Theorem 2, both the average outcome under vaccine in the Doomed as well as distribution of basic principal strata are identified. The exclusion restriction allows us to identify the average outcome under vaccine in the Immune via the average observed outcome in the placebo uninfecteds. By the exclusion restriction, these observed outcomes, even though observed under placebo, are no different than those we would have observed under vaccine. This then allows us to solve for the mean in the protected as a function of these other identifying parameters.

\subsection{Proof of Theorem 5}

\begin{proof}
We have that \begin{align*}
    \E\{ Y(1) \mid S(1) = 0, S(0) = 1, X = x\} 
    &= \E\{ Y(1) \mid S(1) = 0, X = x\} \\
    &= \E\{ Y(1) \mid Z = 1, S(1) = 0,  X = x\} \\
    &= \E( Y \mid Z = 1, S = 0, X = x)  \ ,
\end{align*}
where the first equality follows from partial principal ignorability. Our positivity assumption ensures that $\E(Y \mid Z = 1, S = 0, X = x)$ is well defined for all $x$ such that $P(S = 1 \mid Z = 0, X = x) > 0$.
\end{proof}

Our assumption of partial principle ignorability is  similar to the assumption of principal ignorability \citep{jo2009use}, which in this case would stipulate that $Y(1), Y(0) \perp S(1), S(0) \mid X$. However, due to the fact that our principal stratum of interest is partially identified, we do not need the full principal ignorability assumption. \citet{feller2017principal} noted a weaker form of principal ignorability that can also often be leveraged to identify principal strata estimands. Their assumption that $Y(1) \perp S(0) \mid X$ is also stronger than needed for identification in this case, since we only require this independence to hold in the $S(1) = 0$ strata.

\subsection{Preliminary results for proofs of efficiency} \label{sec:prelim_eff_results}
{\color{changed}
\textit{Efficiency theory for observed data model implied by Assumptions 1-5}

The efficient gradient, or \textit{efficient influence function} (EIF), is characterized through the pathwise derivative of the parameter along regular parametric submodels $P_\varepsilon$ of the observed-data model with score function $s$. Specifically, if there exists a function $\Phi$ satisfying
\[
\left.\frac{d}{d\varepsilon}\psi(P_\varepsilon)\right|_{\varepsilon=0}
=
\E\{\Phi(O)s(O)\},
\]
for all regular parametric submodels, then $\Phi$ is a gradient of the parameter $\psi$. If $\Phi$ belongs to the closed linear span of the tangent set of the model, then $\Phi$ is the canonical gradient and therefore the EIF. 

In a nonparametric model, the closure of the linear span of scores of all regular parametric submodels is $L_0^2(P)$, the Hilbert space of mean-zero, square-integrable functions equipped with covariance inner product. Therefore, in that model, any gradient satisfying the above derivative equation is necessarily the canonical gradient and hence the EIF. The present model is not fully nonparametric, because positivity and monotonicity impose restrictions on the observed-data law. However, we now argue that these restrictions do not reduce the closed linear span of the relevant tangent set for the parameters considered here. Consequently, any gradient satisfying the above derivative equation belongs to the closed linear span of the tangent set and therefore coincides with the canonical gradient and hence the EIF.

While the assumption of treatment positivity restricts the observed-data model, this assumption does not impose a restriction on the tangent set as we have stated positivity as a non-strict inequality and thus it is an open condition that does not restrict local directions. Moreover, none of our parameters of interest depend on the observed-data distribution through $\pi_z$, so this restriction would not modify the tangent set in a meaningful way in any case.

Monotonicity imposes an inequality constraint and therefore may induce a tangent cone rather than a linear tangent space at any distribution $P$ such that $P\{\rho_1(X) = \rho_0(X)\} > 0$. However, under our positivity condition, $0 < \rho_z(x) < 1$ almost surely, there exist regular submodels perturbing $P(S \mid Z, X)$ whose score directions satisfy the linearized monotonicity constraint at boundary laws. Consequently, every score generated by a regular submodel of the unrestricted model may be expressed as a linear combination of scores generated by regular submodels that remain in the restricted model, implying that the closed linear span of the tangent cone coincides with that of the tangent set for the unrestricted model. Therefore, monotonicity does not reduce the closed linear span relevant for gradient/EIF calculations.

Thus, the gradient identified through pathwise differentiation along arbitrary regular parametric submodels of the unrestricted observed-data model coincides with the canonical gradient in our restricted model and therefore is the EIF.

\textit{Exclusion restriction does not restrict the observed-data model}

We now argue that under the exclusion restriction in Assumption 6, no additional restrictions on the observed-data model are implied and thus derivatives calculated using the method above will be the efficient influence function. To show this, we argue that any observed-data distribution $P$ satisfying randomization, positivity, and monotonicity can be generated by at least one full-data distribution $\mathbb{P}$ for the random variable $\mathbb{O} = (X, Z, S(0), S(1), Y(0), Y(1))$ that additionally satisfies Assumption 6. 

Recall that in Theorem 2, we showed that under monotonicity, 
\begin{equation}
\begin{aligned}
P\{ S(0) = 0, S(1) = 0 \mid X=x \} &= 1 -\rho_0(x) \ ,  \\
P(S(0) = 1, S(1) = 0 \mid X = x) &= \rho_0(x)-\rho_1(x) \ , \\
P(S(0) = 1, S(1) = 1 \mid X = x) &= \rho_1(x) \ . 
\end{aligned}
\label{eqn:ps_id}
\end{equation}
To construct a compatible full-data law,  first, set $\mathbb{P}(X)$ equal to the observed marginal law $P(X)$ and assign the conditional distribution of $\{S(0), S(1)\} \mid X$ according to the probabilities above. Next, define the counterfactual outcome distributions so that the observed conditional means are reproduced. We can select a conditional distribution $\mathbb{P}\{ Y(0) \mid S(0) = 0, S(1) = 0, X = x \}$ by selecting any distribution with conditional mean equal to $\mu_{00}(x)$ (i.e., $\E_{P}\{ Y \mid Z = 0, S = 0, X = x \}$) for each $x$. Similarly, we can select a conditional distribution $\mathbb{P}\{ Y(1) \mid S(0) = 1, S(1) = 1, X = x \}$ with conditional mean $\mu_{11}(x)$.

Next, we note that among individuals with $Z=1$ and $S=0$ (i.e., among the union of the Immune and Protected), the conditional probabilities of the Immune and Protected strata are
\begin{align*}
\mathbb{P}\{S(0) = 0 \mid S(1) = 0, X = x \} &= \frac{1-\rho_0(x)}{1-\rho_1(x)} \\
\qquad
\mathbb{P}\{S(0) = 1 \mid S(1) = 0, X = x \} &=  \frac{\rho_0(x)-\rho_1(x)}{1-\rho_1(x)}.
\end{align*}

We set $\E\{Y(1) \mid S(0) = 0, S(1) = 0, X = x\} = \mu_{00}(x),$ which satisfies the exclusion restriction conditionally and therefore also marginally. 

Then, select any distribution for $\mathbb{P}\{Y(1) \mid S(0) = 1, S(1) = 0, X = x\}$ that has conditional mean \[
\frac{1-\rho_1(x)}{\rho_0(x)-\rho_1(x)} \left[ \mu_{10}(x) - \mu_{00}(x) \left\{ \frac{\rho_0(x)-\rho_1(x)}{1-\rho_1(x)} \right\} \right] \ . 
\]

Then by construction we have that \begin{align*}
\mu_{10}(x)
&=
E_{\mathbb{P}}\{ Y(1)\mid S(0) = 0, S(1) = 0, X = x \} \mathbb{P}\{S(0) = 0 \mid S(1) = 0, X = x\} \\
&\hspace{2em} + \E\{ Y(1) \mid S(0) = 1, S(1) = 0, X = x \} \mathbb{P}\{S(0) = 1 \mid S(1) = 0, X = x\} \ , 
\end{align*}
and thus the observed conditional mean in the stratum with $Z = 1, S = 0, X = x$ is appropriately reproduced.

Next consider individuals with $Z=0$ and $S=1$. Our causal model does not restrict this observed subpopulation other than to say it must be a mixture of Protected and Doomed individuals. Thus, we can simply select any distribution for $\mathbb{P}\{ Y(0) \mid S(0) = 1, S(1) = 0, X = x\}$ and for $\mathbb{P}\{ Y(0) \mid S(0) = 1, S(1) = 1, X = x \}$ that each have conditional mean $\mu_{01}(x)$. This allows us to trivially recover $\E\{ Y \mid Z = 0, S = 1, X = x \}$. 

Finally, the construction above specifies  the conditional marginals of $Y(0)$ and of $Y(1)$ given $S(0), S(1),$ and $X$. To complete the full-data distribution, choose an arbitrary conditional coupling of these marginals.

Thus, we have constructed a full-data distribution that reproduces all relevant observed nuisance parameters, while satisfying the marginal exclusion restriction, thereby demonstrating that Assumption 6 does not constrain the observed data distribution beyond Assumptions 1-5.

\textit{Partial principal ignorability  does not restrict the observed-data model}

We now argue that under Assumptions 1-5 and 7 (partial principal ignorability), no additional restrictions on the observed-data model are implied. To show this, as above, we argue that any observed-data distribution $P$ satisfying randomization, positivity, and monotonicity can be generated by at least one full-data distribution $\mathbb{P}$ for the random variable $\mathbb{O} = (X, Z, S(0), S(1), Y(0), Y(1))$ that additionally satisfies Assumption 7.

To construct a compatible full-data law, as above, set $\mathbb{P}(X)$ equal to the observed marginal law $P(X)$ and assign the conditional distribution of $\{S(0),S(1)\}\mid X$ according to the identifying observed data probabilities given in equation \ref{eqn:ps_id}.

We can then select any conditional distribution $\mathbb{P}\{Y(0) \mid S(0) = 0, S(1)  = 0, X = x \}$ by selecting any distribution with conditional mean equal to $\mu_{00}(x)$ for each $x$. Similarly, we can select any conditional distribution $\mathbb{P}\{Y(1) \mid  S(0) = 1, S(1) = 1 , X = x\}$ with conditional mean $\mu_{11}(x)$.

Partial principal ignorability (Assumption 7) states that for individuals with $S(1)=0$, the conditional distribution of $Y(1)$ should be the same for Immune and Protected individuals. To ensure this assumption is satisfied under $\mathbb{P}$, choose a conditional distribution $\mathbb{P}\{Y(1) \mid  S(0) = 0, S(1) = 0, X = x \}$ and for $\mathbb{P}\{Y(1) \mid S(0) = 1, S(1) = 0, X = x\}$ that each have the same conditional mean equal to $\mu_{10}(x)$ for each $x$. With this choice, partial principal ignorability holds and the conditional mean in the observed stratum $Z = 1, S = 0, X = x$ is reproduced automatically, since this observed group is the union of the Immune and Protected strata.

As above, we can simply select any distribution for $\mathbb{P}\{Y(0)\mid S(0)=1,S(1)=0,X=x\}$ and for $\mathbb{P}\{Y(0)\mid S(0)=1,S(1)=1,X=x\}$ that each have conditional mean $\mu_{01}(x)$ to trivially recover the correct conditional mean in the $Z = 0, S = 1, X = x$ stratum. Then, also as above, we complete the full-data distribution by choosing an arbitrary conditional coupling of these marginals.

Thus, we have constructed a full-data distribution that reproduces all relevant observed nuisance parameters while satisfying partial principal ignorability. Therefore, Assumption 7 alone does not constrain the observed-data distribution beyond Assumptions 1-5.

\textit{Useful results for computing gradients}

The arguments above establish that the closed linear span of the relevant tangent set is the same for deriving the EIFs for $\psi_0$, $\psi_{1,\mathrm{ER}}$, and $\psi_{1,\mathrm{PI}}$. Consequently, these EIFs may be obtained by computing gradients using arbitrary regular parametric submodels of the unrestricted observed-data model.
}
Rather than computing the above pathwise derivative directly for each parameter, we repeatedly use the same three decomposition principles described below. 

\textbf{Conditional mean contributions.} Every nuisance regression produces a residual weighted by the inverse probability of observing that regression stratum. So the pathwise derivative of a conditional mean $\E(Y\mid A=a,X)=m_a(X)$ for some event $A=a$ contributes the residual term $\frac{\mathbb{I}(A=a)}{P(A\mid X)}\{Y-m_a(X)\}$. In the present paper this produces the following terms: 
\[
\frac{\mathbb{I}(Z=z,S=s)}{\pi_z(x)P(S = s \mid Z = z, X = x)}\{Y-\mu_{zs}(x)\}, 
\quad  \text{and} \quad
\frac{\mathbb{I}(Z=z)}{\pi_z(x)}\{S-\rho_z(x)\} \ .
\]

\textbf{Marginal distribution contribution.} If a parameter can be written as an expectation over $X$, $\psi=\E\{h(X)\}$, then perturbations of the marginal law of $X$ contribute $h(X)-\psi$. Thus every functional of $X$ generates a plug-in correction term equal to
\[
\text{conditional functional evaluated at }X - \text{target parameter}.
\]

\textbf{Ratio functionals.} Many parameters in the paper are ratios $\psi=\frac{A}{B}$. If $\Phi_A$ and $\Phi_B$ are influence functions for $A$ and $B$, then the influence function for $\psi$ is $\Phi_\psi=\frac{1}{B}\big(\Phi_A-\psi\Phi_B\big)$.

All together, we derive the EIF for our parameters of interest by following the steps below: 
\begin{enumerate}
\item Express the parameter using only $\mu_{zs}(X)$, $\rho_z(X)$ and expectations over $X$.
\item For each $\mu_{zs}(X)$ include an outcome residual term.
\item For each $\rho_z(X)$ include a selection residual term.
\item Add the marginal $X$ correction $h(X)-\psi$.
\item If the parameter is a ratio, apply the ratio rule.
\end{enumerate}

After simplification the resulting expression is the EIF.

\subsection{Proofs for $\psi_0$}

\subsubsection{Proof of Theorem 6}

\begin{proof}
We want the EIF of 
\[
\psi_0
=E\!\left\{\frac{\rho_0(X)}{\bar\rho_0}\mu_{01}(X)\right\}
=\frac{\E\{\rho_0(X)\mu_{01}(X)\}}{\bar\rho_0} = \frac{\E\{\rho_0(X)\mu_{01}(X)\}}{\E\{\rho_0(X)\}} = \frac{A}{B}.
\]

The numerator depends on $\mu_{01}(X)$, $\rho_0(X)$, and  the distribution of $X$, each with the following contributions: 
\[
\frac{(1-Z)S}{\pi_0(X)}\{Y-\mu_{01}(X)\}, 
\quad 
\frac{(1-Z)}{\pi_0(X)}\mu_{01}(X)\{S-\rho_0(X)\}, 
\quad 
\rho_0(X)\mu_{01}(X)-A.
\]

The denominator contribution is: 
\[
\frac{(1-Z)}{\pi_0(X)}\{S-\rho_0(X)\}
+\rho_0(X)-\bar\rho_0.
\]

Given that we wrote $\psi_0$ as a ratio parameter, the final EIF $\Phi_0$ is given by 
\[
\Phi_0(O)
=\frac{1}{\bar\rho_0}\{\Phi_A-\psi_0\Phi_B\}.
\]

After collecting terms, we arrive at the expression in Theorem 6. 
\[
\Phi_0(O)=
\frac{(1-Z)S}{\pi_0(X)\bar\rho_0}\{Y-\mu_{01}(X)\}
+\frac{(1-Z)}{\pi_0(X)\bar\rho_0}\{\mu_{01}(X)-\psi_0\}\{S-\rho_0(X)\}
-\frac{\psi_0}{\bar\rho_0}\{\rho_0(X)-\bar\rho_0\}
+\tilde\psi_0(X)-\psi_0.
\]
\end{proof}

\subsubsection{Proof of asymptotic linearity and robustness}

For brevity, we adopt the notation $Pf = \E_P\{ f(O) \}$ for a $P$-integrable function $f$. Similarly, we let $P_n$ denote the empirical distribution of $n$ samples from $P$ and thus $P_n f = n^{-1} \sum_{i=1}^n f(O_i)$. We also denote by $|| f ||_P = \left\{\int f(o)^{2} dP(o) \right\}^{1/2}$ the $L_2(P)$-norm of a given integrable function $f$.

We assume the following regularity conditions:\begin{itemize}
    \item $\Phi'_{0,n}$ falls in a $P$-Donsker class with probability tending to 1 and $|| \Phi_{0,n} - \Phi_{0,n} ||_P = o_P(1)$
    \item $ || \mu_{01,n} - \mu_{01} ||_P = o_P(n^{-1/4})$
    \item $ || \rho_{0,n} - \rho_0 ||_P = o_P(n^{-1/4})$
    \item $ || \pi_{0,n} - \pi_0 ||_P = o_P(n^{-1/4})$
    \item $\pi_{0,n}$ and $\bar{\rho}_{0,n}$ are bounded below by constant $\delta > 0$ with probability 1
\end{itemize}

We begin by providing a lemma that establishes the linear expansion for the parameter in our model. We use $P$ to denote the sampling distribution of interest and $P'$ to denote another distribution in our model. We add an apostrophe to nuisance parameters to denote their value under sampling from $P'$. Similarly, we denote by $\Psi_0'$ the EIF evaluated at nuisance parameters under sampling from $P'$. 

\begin{lemma}
For any two distributions $P$ and $P'$ in our model,  \[
\psi_0' - \psi_0 = -P \Phi_0' + R_2(P, P') \ ,
\]
where \begin{align*}
R_2(P, P') &= 
P \left\{ \frac{\rho_0}{\bar{\rho}_0'} \frac{(\pi_0 - \pi_0')}{\pi_0'} (\mu_{01} - \mu_{01}') \right\}   + P\left\{ \frac{(\mu_{01}' - \psi_0')}{\bar{\rho}_0'} \frac{(\pi_0 - \pi_0')}{\pi_0'} (\rho_0 - \rho_0') \right\} \\
&\hspace{2em} + \frac{(\bar{\rho}_0' - \bar{\rho}_0)}{\bar{\rho}_0'} (\psi_0' - \psi_0) \ .
\end{align*}
\label{thm:r2_lemma_psi0}
\end{lemma}
The proof follows from straightforward, albeit cumbersome algebra.

Lemma \ref{thm:r2_lemma_psi0} paves the way for a proof of asymptotic normality and of robustness of the one-step estimator. To this end, we may let $P_n'$ denote any distribution in our model that is compatible with nuisance estimates $\rho_{0,n}, \mu_{01,n}, \pi_{0,n}$, and $\bar{\rho}_{0,n}$ and with the marginal distribution of $X$ implied by $P_n'$ equal to the empirical distribution of $X$. Then letting $\Phi_{0,n}$ denote the EIF with nuisance parameters evaluated at their estimated values, Lemma \ref{thm:r2_lemma_psi0} implies that \[
\psi_{0,n}^+ - \psi_0 = P_n \Phi_{0} + (P_n - P) (\Phi_{0,n}' - \Phi_0) + R_2(P, P_n') \ ,
\]
noting that $P \Phi_0 = 0$. The second term on the right hand side is an empirical process term and is such that if $\Phi'_{0,n}$ falls in a $P$-Donsker class with probability tending to 1 and that $P\{ \Phi_{0,n}' - \Phi_{0,n} \}^2 = o_P(1)$, then $(P_n - P) (\Phi_{0,n}' - \Phi_0) = o_P(n^{-1/2})$ \citep{van1996weak}. Then it remains to show that $R_2(P, P_n') = o_P(n^{-1/2})$. This is often shown via application of boundedness conditions and the Cauchy-Schwarz inequality. For example, considering the first term in $R_2$ in Lemma \ref{thm:r2_lemma_psi0}: \begin{align*}
P \left\{ \frac{\rho_0}{\pi_{0,n}\bar{\rho}_{0,n}} (\pi_0 - \pi_{0,n}) (\mu_{10} - \mu_{10,n}) \right\} &\le P \left\{ \frac{\rho_0}{\pi_{0,n}\bar{\rho}_{0,n}} | \pi_0 - \pi_{0,n}| \ |\mu_{10} - \mu_{10,n}| \right\} \\
&\le \frac{\mbox{sup}_x \rho_0(x)}{\delta^2} P \left\{ | \pi_0 - \pi_{0,n}| \ |\mu_{10} - \mu_{10,n}| \right\} \\
&\le \frac{\mbox{sup}_x \rho_0(x)}{\delta^2} || \pi_0 - \pi_{0,n}||_P || \mu_{10} - \mu_{10,n} ||_P \\
&= o_P(n^{-1/2}) \ .
\end{align*}
Similar arguments can be applied to each of the terms in the remainder to prove asymptotic linearity.

Lemma \ref{thm:r2_lemma_psi0} also implies the double robustness of our estimates indicating that either consistent estimation of $\pi_0$ or consistent estimation of both $\mu_{01}$ and of $\rho_0$ are sufficient to ensure consistency of the one-step estimator of $\psi_0$. The proof of multiple robustness follows directly from Lemma \ref{thm:r2_lemma_psi0} and Cauchy Schwarz, where for this result we only require $L^2(P)$ norms of estimation error for nuisance parameters to be $o_P(1)$. 

For the remainder of the proofs of asymptotic linearity of one-step estimators, we opt to merely state the remainder term understanding that similar calculus along with Cauchy-Schwarz can be used to bound remainder terms.

{\color{changed}
\subsubsection{Extensions to cross-fitted estimators}

It is possible to weaken the assumptions under which asymptotic linearity can be established through the use of cross-fitting. In particular, when nuisance parameters are estimated using observations independent of those used to evaluate the estimating equation, it is no longer necessary to assume that the estimated efficient influence function (EIF) belongs to a $P$-Donsker class with probability tending to one. Instead, it suffices to assume consistency of EIF estimates in $L_2(P)$ together with the same control of second-order remainder terms required above. This potentially enables the use of more flexible machine learning algorithms that may not satisfy Donsker conditions.

We explicitly describe the construction of a cross-fitted estimator for $\psi_0$, though the same approach applies to all remaining parameter estimates:
\begin{itemize}
    \item partition the observations into $K$ folds;
    \item for each fold $k$, estimate all regression-valued nuisance parameters using observations outside fold $k$, hence, the $k$-th \emph{training sample};
    \item evaluate the one-step estimator using observations in the withheld fold (hence, the $k$-th \emph{validation sample});
    \item repeat with each fold serving as the validation sample exactly once;
    \item average the resulting fold-specific one-step estimators.
\end{itemize}

For $\psi_0$, proceed as follows. For $k = 1, \dots, K$, let $\mathcal{V}_k \subset \{1,\dots,n\}$ denote the indices in the $k$-th validation fold and let $\mathcal{T}_k = \{1,\dots,n\} \setminus \mathcal{V}_k$ denote the corresponding training indices. Write $n_{\mathcal R}=|\mathcal R|$ for the cardinality of a set $\mathcal R$.

Using only observations indexed by $\mathcal{T}_k$, estimate $\mu_{01}$, $\pi_0$, and $\rho_0$, yielding nuisance estimators $\mu_{01,n}^k$, $\pi_{0,n}^k$, and $\rho_{0,n}^k$. Define $\bar{\rho}_{0,n}^k = \frac{1}{n_{\mathcal T_k}} \sum_{i\in\mathcal T_k} \rho_{0,n}^k(X_i)$, 
and let $\Phi_{0,n}^k$ denote the estimated EIF obtained by substituting these nuisance estimates into the expression for $\Phi_0$. The fold-specific one-step estimator is
\begin{align*}
\psi_{0,n}^k = \frac{1}{n_{\mathcal V_k}} \sum_{i\in\mathcal V_k} \frac{\rho_{0,n}^k(X_i)} {\bar{\rho}_{0,n}^k} \mu_{01,n}^k(X_i) + \frac{1}{n_{\mathcal V_k}} \sum_{i\in\mathcal V_k} \Phi_{0,n}^k(O_i) \ . 
\end{align*}
The final cross-fitted estimator is
\begin{align*}
\psi_{0,n}^{\mathrm{cf}} = \frac{1}{K} \sum_{k=1}^K \psi_{0,n}^k.
\end{align*}

Asymptotic analysis of cross-fitted estimators follows the same logic as above, except that the linear expansion is performed separately within each validation sample. Let $P_{n,k}^{\mathcal V}$ denote the empirical distribution of $\{O_i:i\in\mathcal V_k\}$. Lemma~\ref{thm:r2_lemma_psi0} implies that
\begin{align*}
\psi_{0,n}^{\mathrm{cf}} - \psi_0 = \frac{1}{K} \sum_{k=1}^K \left\{ P_{n,k}^{\mathcal V}\Phi_0 + (P_{n,k}^{\mathcal V}-P) \left( \Phi_{0,n}^k-\Phi_0 \right) + R_2(P,P_{n,k}^{\mathcal V'})
\right\} \ , 
\end{align*}
where, as above, $P_{n,k}^{\mathcal V'}$ denotes any distribution in our model that is compatible with nuisance parameter estimates obtained from the $k$-th training sample.

Control of the remainder terms proceeds exactly as in the non-cross-fitted analysis by applying the same Cauchy-Schwarz arguments separately to each fold. The key distinction concerns the empirical process term. Since $\Phi_{0,n}^k$ is estimated using observations in $\mathcal T_k$, it is independent of observations in $\mathcal V_k$. Therefore, conditional on the training sample,
\begin{align*}
\E\left[
\left\{
(P_{n,k}^{\mathcal V}-P)
(\Phi_{0,n}^k-\Phi_0)
\right\}^2
\mid
\mathcal{T}_k
\right]
&=
\frac{1}{n_{\mathcal V_k}}
\text{Var}_P\{\Phi_{0,n}^k(O)-\Phi_0(O) \mid \mathcal{T}_k\}
\\
&\le
\frac{1}{n_{\mathcal V_k}}
||\Phi_{0,n}^k-\Phi_0||^2 \ , 
\end{align*}
where the variance in the second line treats $\Phi^k_{0,n}$ as a fixed function of $O$. Since, by assumption $||\Phi_{0,n}^k-\Phi_0 || = o_P(1)$, the conditional variance above is $o_P(n_{\mathcal V_k}^{-1})$. Since $K$ is fixed, $n_{\mathcal V_k}/n \to 1/K$, and therefore the conditional variance is also $o_P(n^{-1})$. Therefore, by conditional Chebyshev's inequality,
\begin{align*}
(P_{n,k}^{\mathcal V}-P)(\Phi_{0,n}^k-\Phi_0) = o_P(n^{-1/2}) \ . 
\end{align*}
Thus, asymptotic linearity of the cross-fitted estimator can be established without requiring the estimated EIF to belong to a $P$-Donsker class with probability tending to one.
}

\subsection{Proofs for $\psi_{1,\text{ER}}$}

\subsubsection{Proof of Theorem 7}
\begin{proof}
We want the EIF of 
\[
\psi_{1,\mathrm{ER}}
=\frac{\bar\mu_{1\cdot}-\bar\mu_{00}(1-\bar\rho_0)}{\bar\rho_0}
= \frac{A}{B}.
\]

The EIF of the numerator $A$ is: 
\[
\Phi_A
=\Phi_{\bar\mu_{1\cdot}}
-(1-\bar\rho_0)\Phi_{\bar\mu_{00}}
+\bar\mu_{00}\Phi_{\bar\rho_0}, 
\]
where 
\begin{align*}
\Phi_{\bar\mu_{1\cdot}} &= \frac{Z}{\pi_1(X)}\{Y-\mu_{1\cdot}(X)\}
+\mu_{1\cdot}(X)-\bar\mu_{1\cdot}, 
\\
\Phi_{\bar\mu_{00}} &= \frac{(1-S)(1-Z)}{(1-\bar\rho_0)\bar\pi_0}\{Y-\bar\mu_{00}\}, 
\\
\Phi_{\bar\rho_0} &= \frac{(1-Z)}{\pi_0(X)}\{S-\rho_0(X)\}
+\rho_0(X)-\bar\rho_0.  
\end{align*}

Applying the ratio rule yields 
\[
\Phi_{1,\mathrm{ER}}
=\frac{1}{\bar\rho_0}(\Phi_A-\psi_{1,\mathrm{ER}}\Phi_{\bar\rho_0}).
\]

After simplification the expression matches Theorem 7. 

\end{proof}

\subsubsection{Proof of asymptotic linearity and robustness}

\begin{lemma}
For any two distributions $P$ and $P'$ in our model,  \[
\bar{\mu}_{1\cdot}' -  \bar{\mu}_{1\cdot} = -P \Phi_{\bar{\mu}_{1\cdot}}' + R_2(P, P') \ ,
\]
where \[
R_2(P, P') = P\left\{ \frac{(\pi_1 - \pi_1')}{\pi_1} (\mu_{1\cdot} - \mu_{1\cdot}') \right\} \ . 
\]
We also have \[
\bar{\mu}_{00}' - \bar{\mu}_{00} = -P \Phi_{\bar{\mu}_{00}}' + R_2(P, P') \ , 
\] 
where \[
R_2(P, P') = \frac{(1 - \bar{\rho}_0)}{(1 - \bar{\rho}_0')}\frac{(\pi_0 - \pi_0')}{\pi_0'} (\bar{\mu}_{00} - \bar{\mu}_{00}') + \frac{(\bar{\rho}_0' - \bar{\rho}_0)}{1 - \bar{\rho}_0'} (\bar{\mu}_{00} - \bar{\mu}_{00}') \ .
\]
We also have \[
\bar{\rho}_0' - \bar{\rho}_0 = -P \Phi_{\bar{\rho}_0} + R_2(P, P') \ ,
\]
where \[
R_2(P, P') = P \left\{ \frac{\pi_0 - \pi_0'}{\pi_0'} (\rho_0 - \rho_0') \right\} \ .
\]
\label{thm:lemma_r2_er}
\end{lemma}

Lemma \ref{thm:lemma_r2_er} implies that, along with appropriate Donsker conditions, the following rate conditions are sufficient to ensure that $\psi_{1,\text{ER},n}^+$ is asymptotically linear:\begin{itemize}
    \item $|| \mu_{1\cdot,n} - \mu_{1\cdot} || = o_P(n^{-1/4})$
    \item $|| \pi_{1,n} - \pi_{1} || = o_P(n^{-1/4})$
    \item $|| \pi_{0,n} - \pi_{0} || = o_P(n^{-1/4})$
    \item $|| \rho_{0,n} - \rho_{0} || = o_P(n^{-1/4})$
\end{itemize}

Similarly, Lemma \ref{thm:lemma_r2_er} implies that the combinations of nuisance estimates shown in Table \ref{tab:consistent_est_of_psi1er} are sufficient to ensure consistent estimation of $\psi_{1,\text{ER}}$. In the context of a randomized trial, where $\pi_1$ and $\pi_0$ are known, consistent estimation is always possible irrespective of inconsistent estimation of $\mu_{1,\cdot}$ and/or $\rho_0$.

\begin{table}[ht]
\centering
\begin{tabular}{cccc}
\toprule
$\pi_1$ & $\pi_0$ & $\rho_0$ & $\mu_{1\cdot}$ \\
\midrule
\checkmark & \checkmark &  &  \\
\checkmark &  & \checkmark &  \\
 & \checkmark &  & \checkmark \\
 &  & \checkmark & \checkmark \\
\bottomrule
\end{tabular}
\caption{Minimal combinations of nuisance parameters sufficient for consistency of the one-step estimator of $\psi_{1,\text{ER}}$}
\label{tab:consistent_est_of_psi1er}
\end{table}

\subsection{Proofs for $\psi_{1,\text{PI}}$}
\subsubsection{Proof of Theorem 8}

\begin{proof}
We want the EIF of 
\[
\psi_{1,\mathrm{PI}}
=\frac{\E\{\rho_1(X)\mu_{11}(X)+(\rho_0(X)-\rho_1(X))\mu_{10}(X)\}}{\bar\rho_0}
=\frac{A}{B}.
\]

We compute influence functions for $A$ and $B$, then combine them using
\[
\Phi_{1,\mathrm{PI}}=\frac{1}{B}(\Phi_A-\psi_{1,\mathrm{PI}}\Phi_B).
\]

Since $\rho_0(X)=\E(S\mid Z=0,X)$, the EIF of $B$ is: 
\[
\Phi_B(O)=\frac{1-Z}{\pi_0(X)}\{S-\rho_0(X)\}+\rho_0(X)-\bar\rho_0.
\]

We write $A=\E\{h(X)\}$ where 
\[
h(X)=\rho_1(X)\mu_{11}(X)+(\rho_0(X)-\rho_1(X))\mu_{10}(X).
\]

The contributions of $\mu_{11}(X)$, $\mu_{10}(X)$, and $\rho_1(X)$ are 
\[
\frac{ZS}{\pi_1(X) \rho_1(X)}\{Y-\mu_{11}(X)\}, 
\quad 
\frac{Z(1-S)}{\pi_1(X) (1-\rho_1(X))}\{Y-\mu_{10}(X)\}, 
\quad 
\frac{Z}{\pi_1(X)}\{S - \rho_1(X)\}. 
\]

Therefore the EIF for $A$ is: 
\begin{align*}
\Phi_A(O) &=
\frac{ZS}{\pi_1(X)}\{Y-\mu_{11}(X)\}
+\frac{Z(1-S)}{\pi_1(X)}\frac{\rho_0(X) - \rho_1(X)}{1 - \rho_1(X)}\{Y-\mu_{10}(X)\}\\
&\hspace{2em} +\frac{Z}{\pi_1(X)}\{\mu_{11}(X)-\mu_{10}(X)\}\{S-\rho_1(X)\}
+h(X)-A.
\end{align*}

By the ratio rule, we have: 
\[
\Phi_{1,\mathrm{PI}}(O)
=\frac{1}{\bar\rho_0}\{\Phi_A(O)-\psi_{1,\mathrm{PI}}\Phi_B(O)\}.
\]

Substituting and simplifying,
\[
\begin{aligned}
\Phi_{1,\mathrm{PI}}(O)=&
\frac{ZS}{\pi_1(X)\bar\rho_0}\{Y-\mu_{11}(X)\}
+\frac{Z(1-S)}{\pi_1(X)\bar\rho_0}\frac{\rho_0(X) - \rho_1(X)}{1 - \rho_1(X)}\{Y-\mu_{10}(X)\}\\
&+\frac{Z}{\pi_1(X)\bar\rho_0}\{\mu_{11}(X)-\mu_{10}(X)\}\{S-\rho_1(X)\}+\frac{\mu_{10}(X) - \psi_{1,\mathrm{PI}}}{\bar\rho_0}\frac{1-Z}{\pi_0(X)}\{S-\rho_0(X)\}\\
&-\frac{\psi_{1,\mathrm{PI}}}{\bar\rho_0}\{\rho_0(X)-\bar\rho_0\}
+\tilde\psi_{1,\mathrm{PI}}(X)-\psi_{1,\mathrm{PI}}.
\end{aligned}
\]

This equals the EIF stated in Theorem 8.
\end{proof}

\subsubsection{Proof of asymptotic linearity and robustness}

\begin{lemma}
For any two distributions $P$ and $P'$ in our model,  \[
\psi_{1,\text{PI}}' - \psi_{1,\text{PI}} = -P \Phi_{1,\text{PI}}' + R_2(P, P') \ ,
\]
where \begin{align*}
R_2(P, P') &= 
P\left\{ \frac{\rho_1}{\bar{\rho}_0'} \frac{(\pi_1 - \pi_1')}{\pi_1'} (\mu_{11} - \mu_{11}') \right\} +
P \left\{ \frac{(\rho_1 - \rho_1')}{\bar{\rho}_0'}(\mu_{11} - \mu_{11}') \right\} \\
&\hspace{2em} + P\left\{ \frac{(1 - \rho_1)}{(1 - \rho_1')} (\rho_0' - \rho_1')\frac{(\pi_1 - \pi_1')}{\pi_1'} ( \mu_{10} - \mu_{10}' ) \right\} + 
P\left\{ \frac{\rho_0' - \rho_1'}{\bar{\rho}_0'} \frac{(\rho_1 - \rho_1')}{(1 - \rho_1')} (\mu_{10} - \mu_{10}') \right\} \\
&\hspace{2em} + P\left\{ \frac{\mu_{11}' - \mu_{10}'}{\bar{\rho}_0'} \frac{(\pi_1 - \pi_1')}{\pi_1'} (\rho_1 - \rho_1') \right\} + 
P\left\{ \frac{(\mu_{10}' - \psi_{1, \text{PI}}')}{\bar{\rho}_0'} \frac{(\pi_0 - \pi_0')}{\bar{\pi}_0'} (\rho_0 - \rho_0') \right\} \\
&\hspace{2em} +
P \left\{ \frac{(\rho_0 - \rho_0')}{\bar{\rho}_0'} (\mu_{10} - \mu_{10}') \right\} + P\left\{ \frac{(\rho_1 - \rho_1')}{\bar{\rho}_0'} (\mu_{10} - \mu_{10}') \right\} \\
&\hspace{2em} + P\left\{\frac{(\mu_{11}' - \mu_{11})}{\bar{\rho}_0'} (\rho_1 - \rho_1') \right\} + \frac{(\bar{\rho}_0' - \bar{\rho}_0)}{\bar{\rho}_0'} (\psi_{1,\text{PI}}' - \psi_{1,\text{PI}}) 
\end{align*}
\label{thm:r2_lemma_psi1_pi}
\end{lemma}

Lemma \ref{thm:r2_lemma_psi1_pi} implies that, along with appropriate Donkser conditions, the following conditions are sufficient to ensure that $\psi_{1,\text{PI},n}^+$ is asymptotically linear:\begin{itemize}
    \item $|| \mu_{11,n} - \mu_{11} || = o_P(n^{-1/4})$
    \item $|| \mu_{10,n} - \mu_{10} || = o_P(n^{-1/4})$
    \item $|| \pi_{1,n} - \pi_{1} || = o_P(n^{-1/4})$
    \item $|| \pi_{0,n} - \pi_{0} || = o_P(n^{-1/4})$
    \item $|| \rho_{0,n} - \rho_{0} || = o_P(n^{-1/4})$
    \item $|| \rho_{1,n} - \rho_{1} || = o_P(n^{-1/4})$
\end{itemize}

Similarly, Lemma \ref{thm:lemma_r2_er} implies that the combinations of nuisance estimates shown in Table \ref{tab:consistent_est_psi1pi} are sufficient to ensure consistent estimation of $\psi_{1,\text{PI}}$. In the context of a randomized trial, where $\pi_1$ and $\pi_0$ are known, consistent estimation is always possible under the minimal combinations shown in Table \ref{tab:consistent_est_psi1pi_rct}.

\begin{table}[ht]
\centering
\begin{tabular}{cccccc}
\toprule
$\pi_1$ & $\pi_0$ & $\rho_1$ & $\rho_0$ & $\mu_{11}$ & $\mu_{10}$ \\
\midrule
\checkmark &  & \checkmark & \checkmark &  &  \\
\checkmark & \checkmark & \checkmark &  &  & \checkmark \\
\checkmark &  &  & \checkmark & \checkmark & \checkmark \\
\checkmark & \checkmark &  &  & \checkmark & \checkmark \\
 &  & \checkmark & \checkmark & \checkmark & \checkmark \\
 & \checkmark & \checkmark &  & \checkmark & \checkmark \\
\bottomrule
\end{tabular}
\caption{Minimal combinations of consistently estimated nuisance parameters that result in consistent estimation of $\psi_{1,\text{PI}}$.}
\label{tab:consistent_est_psi1pi}
\end{table}

\begin{table}[ht]
\centering
\begin{tabular}{cccc}
\toprule
$\rho_1$ & $\rho_0$ & $\mu_{11}$ & $\mu_{10}$ \\
\midrule
\checkmark & \checkmark &  &  \\
\checkmark &  &  & \checkmark \\
 &  & \checkmark & \checkmark \\
\bottomrule
\end{tabular}
\caption{Minimal combinations of consistently estimated nuisance parameters that result in consistent estimation of $\psi_{1,\text{PI}}$ in the context of a randomized trial (where estimators of $\pi_1$ and $\pi_0$ are guaranteed consistent).}
\label{tab:consistent_est_psi1pi_rct}
\end{table}

\subsection{Proofs for exposure-conditional effects}

\subsubsection{Proof of Theorem 9}

\begin{proof}
{
\color{changed}
    We have that \begin{align*}
        \E\{ Y(0) \mid E = 1 \} &= \E[ \E\{ Y(0) \mid  E = 1, X \} \mid E = 1]   \\
        &= \E[ \E\{ Y(0) \mid  Z = 0, E = 1, X \} \mid E = 1] \\
        &= \E\{ \E( Y \mid  Z = 0, E = 1, X ) \mid E = 1 \} \ , 
    \end{align*}
    where the first line follows from the tower rule, the second from randomization and Assumption 9, the third from consistency. We proceed by proving that $\E(Y \mid Z = 0, E = 1, X = x) = \mu_{01}(x)$ and that $P(X = x \mid E = 1) = \rho_0(x) / \bar{\rho}_0$ for each $x$.
    
    Towards the former, we have that \begin{align}
        &\E( Y \mid  Z = 0, E = 1, X = x )  \notag \\
        &\hspace{2em} = \E( Y \mid Z = 0, E = 1, S = 0, X = x ) P( S = 0 \mid Z = 0, E = 1, X = x ) + \notag  \\
        &\hspace{4em} + \E( Y \mid Z = 0, E = 1, S = 1, X = x ) P( S = 1 \mid Z = 0, E = 1, X = x ) \notag \\
        &\hspace{2em} = \E( Y \mid Z = 0, E = 1, S = 0, X = x ) \times 0 + \E( Y \mid Z = 0, E = 1, S = 1, X = x ) \times 1 \notag \\
        &\hspace{2em} = \E( Y \mid Z = 0, E = 1, S = 1, X = x ) \label{eqn:return_here_proof_thm9_2} \ , 
    \end{align}
    where the first equality follows from the tower rule and the second from Assumption 8. The goal is to now show that $\E( Y \mid Z = 0, E = 1, S = 1, X = x ) = \E( Y \mid Z = 0, S = 1, X = x )$. Toward that end, note that
    \begin{equation}
    \begin{aligned}
        \E( Y \mid Z = 0, S = 1, X = x ) &= \E( Y \mid Z = 0, S = 1, X = x, E = 1 ) P(E = 1 \mid S = 1, Z = 0, X = x) \\
        &\hspace{2em} + \E( Y \mid Z = 0, S = 1, X = x, E = 0 ) P(E = 0 \mid S = 1, Z = 0, X = x) \ , \label{eqn:return_here_proof_thm9}
    \end{aligned}
    \end{equation}
    We then write that \begin{align*}
    P(E = 0 \mid S = 1, Z = 0, X = x) &= \frac{P(S = 1 \mid E = 0, Z = 0, X = x) P(E = 0 \mid Z = 0, X = x)}{P(S = 1 \mid Z = 0, X = x)} \\
    &= 0 \ ,
    \end{align*}
    which follows from exposure necessity (Assumption 8). We also have that {\color{changed} \begin{align*}
        &P(E = 1 \mid S = 1, Z = 0, X = x) \\
        &\hspace{2em} = \frac{P(S = 1 \mid E = 1, Z = 0, X = x) P(E = 1 \mid Z = 0, X = x)}{P(S = 1 \mid Z = 0, X = x)} \\
        &\hspace{2em} = \frac{1 \times P(E = 1 \mid Z = 0, X = x)}{P(S = 1 \mid Z = 0, X = x)} \ ,
    \end{align*}
    which follows from exposure sufficiency (Assumption 8).
    }
    But we also have \begin{align*}
    P(S = 1 \mid Z = 0, X = x) &= P(S = 1 \mid Z = 0, X = x, E = 0) P(E = 0 \mid Z = 0, X = x) \\
    &\hspace{2em} + P(S = 1 \mid Z = 0, X = x, E = 1) P(E = 1 \mid Z = 0, X = x) \\
    &= 0 \times P(E = 0 \mid Z = 0, X = x) + 1 \times P(E = 1 \mid Z = 0, X = x) \\
    &= P(E = 1 \mid Z = 0, X = x)
    \end{align*}
    Thus, we have that $P(E = 1 \mid S = 1, Z = 0, X = x) = 1$ and that returning to (\ref{eqn:return_here_proof_thm9}) that $\E( Y \mid Z = 0, E = 1, S = 1, X = x ) = \E( Y \mid Z = 0, S = 1, X = x ) = \mu_{10}(x)$, as desired.

Turning to identification of the weights: \begin{align}
    P(X = x \mid E = 1) &= \frac{P(E = 1 \mid X = x)P(X = x)}{P(E = 1)} \notag \\
     &= \frac{P(E = 1 \mid Z = 0, X = x)}{P(E = 1 \mid Z = 0)} P(X = x) \ ,
     \label{eqn:blah}
    \end{align}
which follows from Assumption 9. The numerator of this expression is \begin{align*}
P(E = 1 \mid Z = 0, X = x) &= P(E = 1 \mid Z = 0, S = 0, X = x)P(S = 0 \mid Z = 0, X = x) \\
&\hspace{2em} + P(E = 1 \mid Z = 0, S = 1, X = x)P(S = 1 \mid Z = 0, X = x) \\
&= 0 \times P(S = 0 \mid Z = 0, X = x) + 1 \times P(S = 1 \mid Z = 0, X = x) \ ,
\end{align*}
where the second line follows from our arguments above. Thus, we have that $P(E = 1 \mid Z = 0, X = x) = P(S = 1 \mid Z = 0, X = x) = \rho_0(x)$ and $P(E = 1 \mid Z = 0) = \E\{ \rho_0(X) \} = \bar{\rho}_0$. Therefore, we have shown that \[
\E\{Y(0) \mid E = 1 \} = \E\left\{ \frac{\rho_0(X)}{\bar{\rho}_0} \mu_{01}(X) \right\} \ .
\]
}
\end{proof}

\subsubsection{Proof of Theorem 10}

\begin{proof}
    As established in the Proof of Theorem 4, we have that $\E\{ Y(1) \} = \bar{\mu}_{1\cdot}$  and $\E\{ Y(0) \} = \bar{\mu}_{0\cdot}$. Furthermore, we have that \begin{align*}
        \E\{Y(1) - Y(0) \} &= \E\{Y(1) - Y(0) \mid E = 1\} P(E = 1) + \E\{Y(1) - Y(0) \mid E = 0\} P(E = 0) \\
        &= \E\{Y(1) - Y(0) \mid E = 1\} P(E = 1) \ , 
    \end{align*}
    which follows from the exposure-conditional exclusion restriction: { \color{changed}\begin{align*}
        \E\{Y(1) - Y(0) \mid E = 0\} &= \E\{Y(1) \mid E = 0\} - \E\{ Y(0) \mid E = 0 \} \\
        &= \E[\E\{Y(1) \mid E = 0, X\} - \E\{ Y(0) \mid E = 0, X \} \mid E = 0] \\
        &= \E[\E\{Y(1) \mid Z = 1, E = 0, X\} - \E\{ Y(0) \mid Z = 0, E = 0, X \} \mid E = 0] \\
        &= \E[\E\{Y \mid Z = 1, E = 0, X\} - \E\{ Y(0) \mid Z = 0, E = 0, X \} \mid E = 0] = 0 \ . 
    \end{align*}
    As established in the Proof of Theorem 9, $\E[\E\{ Y(0) \mid Z = 0, E = 1, X \} \mid E = 1 ] = \psi_0$. We also have that under Assumption 9, \begin{align*}
        P(E = 1) &= \E\{ P(E = 1 \mid X) \} \\
        &= \E\{ P(E = 1 \mid Z = 0, X) \} \\
        &= \E\{ P(S = 1 \mid Z = 0, X) \} = \bar{\rho}_0 \ .
    \end{align*}
    where the third equality is shown in the proof of Theorem 9. Thus, we have shown that $(\bar{\mu}_{1\cdot} - \bar{\mu}_{0\cdot} + \psi_0 ) / \bar{\rho}_0 = \psi_{1,\text{ER}}$.
}
\end{proof}

\subsubsection{Proof of Theorem 11}

\begin{proof}
For simplicity and without loss of generality, assume $X$ is discrete. We have \begin{align*}
\E\{ Y(1) \mid E = 1 \} = \sum_x \E\{ Y(1) \mid E = 1, X = x \} P(X = x \mid E = 1) \ . 
\end{align*}

{\color{changed}
As shown in the proof of Theorem 9, $P(X = x \mid E = 1) = \rho_0(x) / \bar{\rho}_0$.
}
So it suffices to consider identification of $\E[Y(1) \mid E = 1, X = x]$. We note that \begin{align*}
    \E\{ Y(1) \mid E = 1, X = x \}
    &= \E\{ Y(1) \mid Z = 1, E = 1, X = x \}  \\
    &= \E( Y \mid Z = 1, E = 1, X = x ) \\
    &= \sum_{s = 0}^1 \E( Y \mid Z = 1, E = 1, X = x, S = s ) P(S = s \mid Z = 1, E = 1, X = x) \\
    &= \sum_{s = 0}^1 \E( Y \mid Z = 1, X = x, S = s ) P(S = s \mid Z = 1, E = 1, X = x) \ .
\end{align*}
{\color{changed}Here, the equalities follow from randomization and Assumption 9, consistency, the law of total expectation, and Assumption 13, respectively}. Now, we consider identification of $P(S = s \mid Z = 1, E = 1, X = x)$ for $s = 1$. Note that 
\begin{align*}
    P(S = 1 \mid Z = 1, X = x) &= P(S = 1 \mid Z = 1, E = 1, X = x) P(E = 1 \mid Z = 1, X = x) \\
    &\hspace{2em} + P(S = 1 \mid Z = 1, E = 0, X = x) P(E = 0 \mid Z = 1, X = x) \\
    &= P(S = 1 \mid Z = 1, E = 1, X = x) P(E = 1 \mid Z = 1, X = x) \ ,
\end{align*}
which follows from the assumption that $P(S = 1 \mid Z = 1, E = 0, X = x) = 0$. Similarly,\begin{align*}
    P(S = 1 \mid Z = 0, X = x) &= P(S = 1 \mid Z = 0, E = 1, X = x) P(E = 1 \mid Z = 0, X = x) \\
    &\hspace{2em} + P(S = 1 \mid Z = 0, E = 0, X = x) P(E = 0 \mid Z = 0, X = x) \\
    &= P(S = 1 \mid Z = 0, E = 1, X = x) P(E = 1 \mid Z = 0, X = x) \\
    &= P(E = 1 \mid Z = 0, X = x) \\
    &= P(E = 1 \mid Z = 1, X = x)
\end{align*}
Here the equalities follow from law of total probability, the assumption that exposure is necessary for infection, the assumption that exposure is sufficient for infection, and the assumption that $E \perp Z \mid X$. Thus, we have shown that \begin{align*}
P(S = 1 \mid Z = 1, E = 1, X = x) = \frac{P(S = 1 \mid Z = 1, X = x)}{P(S = 1 \mid Z = 0, X = x)} = \frac{\rho_1(x)}{\rho_0(x)}\ .
\end{align*}
Then, trivially it must also be true that \begin{align*}
P(S = 0 \mid Z = 1, E = 1, X = x) &= 1 - P(S = 1 \mid Z = 1, E = 1, X = x) \\
&= 1 - \frac{P(S = 1 \mid Z = 1, X = x)}{P(S = 1 \mid Z = 0, X = x)} = 1 - \frac{\rho_1(x)}{\rho_0(x)} \ .
\end{align*}
Thus, we have shown that $\E\{ Y(1) \mid E = 1 \} = \psi_{1, \text{PI}}$. 

\end{proof}

\subsection{Proof of Theorem \ref{thm:sens_pi_id}}

\begin{proof}
We have that \begin{align*}
    \E\{ Y(1) \mid S(1) = 0, X = x\} &= \E\{ Y(1) \mid S(1) = 0, S(0) = 0, X = x\} P\{S(0) = 0 \mid S(1) = 0, {\color{changed} X = x} \}  \\
    &\hspace{2em} + \E\{ Y(1) \mid S(1) = 0, S(0) = 1, X = x\} P\{ S(0) = 1 \mid S(1) = 0, {\color{changed} X = x} \} \\
    &= \E\{ Y(1) \mid S(1) = 0, S(0) = 0, X = x\} \left\{ \frac{1 - \rho_0(x)}{1 - \rho_1(x)} \right\} \\
    &\hspace{2em} + \E\{ Y(1) \mid S(1) = 0, S(0) = 1, X = x\} \left\{ \frac{\rho_0(x) - \rho_1(x)}{1 - \rho_1(x)} \right\} \\
    &= \epsilon \ \E\{Y(1) \mid S(1) = 0, S(0) = 1, X = x \} \left\{ \frac{1 - \rho_0(x)}{1 - \rho_1(x)} \right\} \\
    &\hspace{2em} + \E\{ Y(1) \mid S(1) = 0, S(0) = 1, X = x\} \left\{ \frac{\rho_0(x) - \rho_1(x)}{1 - \rho_1(x)} \right\} \ .
\end{align*}
where the first equality is the tower rule, the second results from Theorem 2, the third from Assumption \ref{assmp:partial_pi_sens_assump}. Moreover, we also have that \begin{align*}
    \E\{ Y(1) \mid S(1) = 0, X = x \} &= \E\{ Y(1) \mid Z = 1, S(1) = 0, X = x \} \\
    &= \E\{ Y \mid Z = 1, S = 0, X = x \} = \mu_{10}(x)
\end{align*}
Thus, we have that \begin{align*}
    \mu_{10}(x) = \E\{Y(1) \mid S(1) = 0, S(0) = 1, X = x \} \left[ \epsilon \left\{ \frac{1 - \rho_0(x)}{1 - \rho_1(x)} \right\} + \left\{ \frac{\rho_0(x) - \rho_1(x)}{1 - \rho_1(x)} \right\} \right] \ ,
\end{align*}
and thus that \begin{equation}
\E\{Y(1) \mid S(1) = 0, S(0) = 1, X = x \} = \mu_{10}(x) \frac{1}{\epsilon \left\{ \frac{1 - \rho_0(x)}{1 - \rho_1(x)} \right\} + \left\{ \frac{\rho_0(x) - \rho_1(x)}{1 - \rho_1(x)} \right\} } \ . \label{eqn:ref_me_proof_thm_s1}
\end{equation}

{\color{changed}
Recall that under Assumptions 1-5, \begin{align*}
\E\{Y(1) \mid S(0) = 1\} = \frac{1}{\bar{\rho}_0} \left[ \E\{ \mu_{11}(X) \rho_1(X) \}  + \E\{ Y(1) \mid S(0) = 1, S(1) = 0 \} ( \bar{\rho}_0 - \bar{\rho}_1 ) \right] \ , 
\end{align*}
and thus, 
\begin{align*}
\E\{Y(1) \mid S(0) = 1\} = \frac{1}{\bar{\rho}_0}  \E\{ \mu_{11}(X) \rho_1(X) + \E\{ Y(1) \mid S(0) = 1, S(1) = 0, X = x \} \{ \rho_0(X) - \rho_1(X) \} \ .
\end{align*}
Plugging (\ref{eqn:ref_me_proof_thm_s1}) into this equation and re-arranging terms yields the result. 
}
\end{proof}

\subsection{Proofs for $\psi_{1,\cdot}$}
\subsubsection{Proof of Theorem \ref{thm:equiv_er_and_pi_id}}

\begin{proof}
In the proof of Theorem 5, we showed that under partial principal ignorability, $\E\{ Y(1) \mid S(0) = 1, S(1) = 0, X \} = \E( Y \mid Z = 1, S = 0, X)$. However, if we additionally assume an exclusion restriction then 
\begin{align*}
    \E\{ Y(1) \mid S(0) = 1, S(1) = 0, X \} &= \E\{ Y(1) \mid S(0) = 0, S(1) = 0, X \} \\
    &= \E\{ Y(1) \mid S(0) = 0, S(1) = 0, X \} \\
    &= \E\{ Y(0) \mid S(0) = 0, S(1) = 0, X \} \\
    &= \E\{ Y(0) \mid S(0) = 0, S(1) = 0, X \} \\
    &= \E\{ Y(0) \mid S(0) = 0, X \} \\
    &= \E\{ Y(0) \mid Z = 0, S(0) = 0, X \} \\
    &= \E\{ Y \mid Z = 0, S = 0, X \} \ .
\end{align*}
Thus, we have shown that if both principal ignorability and exclusion restriction hold, then $\E( Y \mid Z = 1, S = 0, X) = \E( Y \mid Z = 0, S = 0, X)$. Furthermore, \begin{align*}
\psi_{1,\text{ER}} &= \frac{\bar{\mu}_1 - \bar{\mu}_{00} (1 - \bar{\rho}_0)}{\bar{\rho}_0} \\
&= \frac{\E\left[ \mu_{11}(X) \rho_1(X) + \mu_{10}(X) \{ 1- \rho_1(X) \} \right]}{\bar{\rho}_0} - \frac{\E\{ \mu_{00}(X) \{ 1 - \rho_0(X) \} \}}{\bar{\rho}_0} \\
&\hspace{2em} = \frac{1}{{\bar{\rho}_0}} \E\{ \rho_1(X) \mu_{11}(X) + \{1 - \rho_1(X)\} \mu_{10}(X) - \{1 - \rho_0(X)\} \mu_{00}(X) \} \\ 
&\hspace{2em} = \frac{1}{{\bar{\rho}_0}} \E\{ \rho_1(X) \mu_{11}(X) + \{1 - \rho_1(X)\} \mu_{10}(X) - \{1 - \rho_0(X)\} \mu_{10}(X) \} \\ 
&\hspace{2em} = \frac{1}{{\bar{\rho}_0}} \E\{ \rho_1(X) \mu_{11}(X) + \{\rho_0(X) - \rho_1(X)\} \mu_{10}(X) \} \\
&\hspace{2em} = \psi_{1,\text{PI}} \ .
\end{align*}
\end{proof}

\subsubsection{Proof of Theorem \ref{thm:grad_for_semiparam}}
{\color{changed}
When both the exclusion restriction and partial principal ignorability hold, Theorem \ref{thm:equiv_er_and_pi_id} implies that $Y$ is mean independent of $Z$ given $S=0$ and $X$. We now show that this observed-data restriction is sharp: any observed-data distribution in the observed-data model implied by Assumptions 1-5, and this conditional mean restriction can be generated by at least one full-data distribution satisfying both the exclusion restriction and partial principal ignorability. Thus, the conditional mean restriction is the only additional observed-data restriction implied by assuming both Assumptions 6 and 7. We then characterize the resulting restriction on the tangent space and discuss the implications for efficient influence function calculations.

To show that the observed-data restriction is sharp, consider an arbitrary observed-data distribution $P$ satisfying randomization, positivity, monotonicity, and conditional mean independence, $\mu_{00}(x) = \mu_{10}(x) = \mu_{\cdot0}(x)$ for each $x$.
As in Section \ref{sec:prelim_eff_results}, set $\mathbb{P}(X)$ equal to the observed marginal law $P(X)$ and assign the conditional distribution of $\{S(0), S(1)\} \mid X$ according to equation (\ref{eqn:ps_id}).

For the Immune stratum, set select a distribution $\mathbb{P}\{ Y(0) \mid S(0) = 0, S(1) = 0, X = x \}$ and for $\mathbb{P}\{ Y(1) \mid S(0) = 0, S(1) = 0, X = x \}$ that each have conditional mean equal to $\mu_{00}(x)$ for each $x$. This ensures that $\mathbb{P}$ satisfies the exclusion restriction (Assumption 6).

Next consider the Protected stratum.  Partial principal ignorability (Assumption 7) requires $\E_{\mathbb{P}}\{Y(1) \mid S(0) = 1, S(1) = 0, X = x\} = \E_{\mathbb{P}}\{Y(1) \mid S(0) = 0, S(1) = 0, X = x\}$. Thus, for $\mathbb{P}\{ Y(1) \mid S(0) = 1, S(1) = 0, X = x \}$, we may select any distribution with conditional mean $\mu_{00}(x)$. This choice implies that the observed mean in the observed stratum $Z = 1, S = 0, X = x$, $\mu_{10}(x)$ exactly equals $\mu_{00}(x)$ for each $x$ and thus the assumed observed-data restriction holds for this choice of $\mathbb{P}$

As in Section \ref{sec:prelim_eff_results}, we can select a distribution $\mathbb{P}\{ Y(1) \mid S(0) = 1, S(1) = 1, X = x \}$ that has mean $\mu_{11}(x)$, select distributions $\mathbb{P}\{ Y(0) \mid S(0) = 1, S(1) = 0, X = x \}$ and $\mathbb{P}\{ Y(0) \mid S(0) = 1, S(1) = 1, X = x \}$ that each have conditional mean $\mu_{01}(x)$ for each $x$, and choose an arbitrary conditional coupling of $Y(0)$ and $Y(1)$ given $S(0),S(1),X$. 

In Theorem \ref{thm:equiv_er_and_pi_id}, we established that the set of observed-data distributions compatible with Assumptions 1--7 is contained within the observed-data model that satisfies randomization, positivity, monotonicity, and $\mu_{10}(x)=\mu_{00}(x)$. The argument above establishes the converse: any observed-data distribution satisfying randomization, positivity, monotonicity, and this conditional mean restriction can be generated by at least one full-data distribution satisfying Assumptions 1--7. Therefore, the observed-data model induced by Assumptions 1--7 coincides exactly with the observed-data model satisfying conditional mean independence. Thus, this conditional mean restriction is the only additional observed-data restriction implied by jointly assuming Assumptions 6 and 7.

We can now turn to characterizing the closed linear span of the tangent set for this model. Due to the factorization of the observed-data law, this space admits a decomposition into orthogonal components corresponding to perturbations of each factor in the observed-data likelihood. The conditional mean restriction $\mu_{10}(x) = \mu_{00}(x)$ restricts only the component corresponding to perturbations of the conditional distribution of $Y \mid Z, S = 0, X$. In particular, under this model the corresponding closure of linear span of scores becomes $\mathcal{T}_{Y\mid S=0,X}$, the Hilbert space of functions of $(Y,S,X)$ having conditional mean zero given $(S=0,X)$. An arbitrary score $s\in L^2_0(P)$ can be projected into this space using the projection operator
\[
\Pi\{s \mid \mathcal{T}_{Y \mid S = 0, X}\}(Y,S,X) = \E_{\mathbb{P}}\{s(O) \mid Y, S, X\} - \E_{\mathbb{P}}\{s(O) \mid S = 0, X \} \ . 
\]
Thus, to obtain the efficient gradient for $\psi_{1,\cdot}$, we begin with the gradient derived in the unrestricted model and project the component corresponding to perturbations of $\mu_{10}$ into $\mathcal{T}_{Y\mid S=0,X}$ using this projection operator.}

Let \[
s_1(O) = \frac{Z}{\pi_1(X)} \frac{(1 - S)}{\{1 - \rho_1(X)\}} \frac{\{\rho_0(X) - \rho_1(X)\}}{\{1 - \rho_1(X)\}} \{ Y - \mu_{10}(X)\} \ , 
\]
and compute \begin{align*}
    &\E_{\mathbb{P}}\{ s_1(O) \mid Y, S, X \} - \E_{\mathbb{P}}\{ s_1(O) \mid S = 0, X \}  = \frac{(1 - S)}{1 - \bar{\rho}_{\cdot}} \frac{\{\rho_0(X) - \rho_1(X)\}}{\bar{\rho}_0} \{ Y - \mu_{\cdot0}(X) \} \ .
\end{align*}
Projections of all other pieces of the gradient for $\psi_{1,\text{PI}}$ are easily confirmed to equal zero. The proof is completed by replacing $\mu_{10}$ with $\mu_{\cdot0}$ wherever it appears in the gradient, as these quantities are equivalent in the semiparametric model.

\subsubsection{Proof of asymptotic linearity and robustness}

\begin{lemma}
For any two distributions $P$ and $P'$ in our model,  \[
\psi_{1,\cdot}' - \psi_{1,\cdot} = -P \Phi_{1,\cdot}' + R_2(P, P') \ ,
\]
where \begin{align*}
R_2(P, P') &= 
P\left\{ \frac{\rho_1}{\bar{\rho}_0'} \frac{(\pi_1 - \pi_1')}{\pi_1'} (\mu_{11} - \mu_{11}') \right\} +
P \left\{ \frac{(\rho_1 - \rho_1')}{\bar{\rho}_0'}(\mu_{11} - \mu_{11}') \right\} \\
&\hspace{2em} + \frac{(\bar{\rho}_{\cdot} - \bar{\rho}_{\cdot}')}{\bar{\rho}_0} P\left\{ (\rho_0' - \rho_1')\frac{(\pi_1 - \pi_1')}{\pi_1'} ( \mu_{\cdot0} - \mu_{\cdot0}' ) \right\} + 
P\left\{ \frac{\rho_0' - \rho_1'}{\bar{\rho}_0'} \frac{(\rho_1 - \rho_1')}{(1 - \rho_1')} (\mu_{\cdot0} - \mu_{\cdot0}') \right\} \\
&\hspace{2em} + P\left\{ \frac{\mu_{11}' - \mu_{\cdot0}'}{\bar{\rho}_0'} \frac{(\pi_1 - \pi_1')}{\pi_1'} (\rho_1 - \rho_1') \right\} + 
P\left\{ \frac{(\mu_{\cdot0}' - \psi_{1,\cdot}')}{\bar{\rho}_0'} \frac{(\pi_0 - \pi_0')}{\bar{\pi}_0'} (\rho_0 - \rho_0') \right\} \\
&\hspace{2em} +
P \left\{ \frac{(\rho_0 - \rho_0')}{\bar{\rho}_0'} (\mu_{\cdot0} - \mu_{\cdot0}') \right\} + P\left\{ \frac{(\rho_1 - \rho_1')}{\bar{\rho}_0'} (\mu_{\cdot0} - \mu_{\cdot0}') \right\} \\
&\hspace{2em} + P\left\{\frac{(\mu_{11}' - \mu_{11})}{\bar{\rho}_0'} (\rho_1 - \rho_1') \right\} + \frac{(\bar{\rho}_0' - \bar{\rho}_0)}{\bar{\rho}_0'} (\psi_{1,\cdot}' - \psi_{1,\cdot}) 
\end{align*}
\label{thm:r2_lemma_psi1_both}
\end{lemma}

Lemma \ref{thm:r2_lemma_psi1_both} along with appropriate Donkser conditions implies that the following conditions are sufficient to ensure that $\psi_{1,\cdot,n}^+$ is asymptotically linear:\begin{itemize}
    \item $|| \mu_{11,n} - \mu_{11} || = o_P(n^{-1/4})$
    \item $|| \mu_{\cdot0,n} - \mu_{\cdot0} || = o_P(n^{-1/4})$
    \item $|| \pi_{1,n} - \pi_{1} || = o_P(n^{-1/4})$
    \item $|| \pi_{0,n} - \pi_{0} || = o_P(n^{-1/4})$
    \item $|| \rho_{0,n} - \rho_{0} || = o_P(n^{-1/4})$
    \item $|| \rho_{1,n} - \rho_{1} || = o_P(n^{-1/4})$
\end{itemize}

Similarly, Lemma \ref{thm:r2_lemma_psi1_both} implies that the combinations of nuisance estimates shown in Table \ref{tab:consistent_est_psi1pi} and Table \ref{tab:consistent_est_psi1pi_rct} are sufficient to ensure consistent estimation of $\psi_{1,\cdot}$, where the conditions requiring consistent estimation of $\mu_{10}$ are replaced by consistent estimation of $\mu_{\cdot0}$.

\subsection{Proofs for Doomed Estimand}

\subsubsection{Proof of Theorem \ref{thm:id_doomed_vax}}

\begin{proof}
We have \begin{align*}
    \E\{Y(1) \mid S(0) = 1, S(1) = 1 \} &= \E\{Y(1) \mid S(1) = 1 \} \\
    &= \E\{Y(1) \mid Z = 1, S(1) = 1 \} \\
    &= \E\{Y \mid Z = 1, S = 1 \} \\
    &= E \left\{ \frac{\rho_1(X)}{\bar{\rho}_1} \mu_{11}(X) \right\} \ .
\end{align*}
The first equality follows from monotonicity, the second from vaccine randomization. The third follows from consistency, the last from the tower rule.
\end{proof}

\subsubsection{Proof of Theorem \ref{thm:id_doomed_placebo}}

\begin{proof}
We have \begin{align*}
    \E\{Y(0) \mid S(0) = 1, S(1) = 1 \} 
    &= \E \left\{ \frac{\rho_1(X)}{\bar{\rho}_1} \E\{Y(0) \mid S(0) = 1, S(1) = 1, X \}  \right\} \\
   &= \E \left\{ \frac{\rho_1(X)}{\bar{\rho}_1} \E\{Y(0) \mid S(0) = 1, X \}  \right\} \\
   &= \E \left\{ \frac{\rho_1(X)}{\bar{\rho}_1} \E\{Y(0) \mid Z = 0, S(0) = 1, X \}  \right\} \\
   &= \E \left\{ \frac{\rho_1(X)}{\bar{\rho}_1} \E\{Y \mid Z = 0, S = 1, X \}  \right\} \\
   &= \E \left\{ \frac{\rho_1(X)}{\bar{\rho}_1} \mu_{01}(X)  \right\} \ . 
\end{align*}
The first equality follows from the tower rule, the second from principal ignorability. The third follows from vaccine randomization, the fourth from consistency.
\end{proof}

\subsubsection{Proof of Theorem \ref{thm:exp_cond_doomed_id}}
\label{sec:proof_s9}
\begin{proof}
Assume again that $X$ is discrete. We have \begin{align*}
\E\{ Y(1) \mid E^* = 1 \} &= \sum_x \E\{ Y(1) \mid E^* = 1, X = x \} P(X = x \mid E^* = 1) \  , \ \mbox{and} \\
\E\{ Y(0) \mid E^* = 1 \} &= \sum_x \E\{ Y(0) \mid E^* = 1, X = x \} P(X = x \mid E^* = 1) \  .
\end{align*}

Note that \begin{align}
P(X = x \mid E^* = 1) &= \frac{P(E^* = 1 \mid X = x) P(X = x)}{P(E^* = 1)} \notag \\
&= \frac{P(E^* = 1 \mid Z = 1, X = x) P(X = x)}{P(E^* = 1)}  \label{eqn:ref_me_proof_doomed_exp} 
\end{align}
The second line follows from randomization and Assumption \ref{assmp:doom_vax_not_cause}. 

We now note that \begin{align*}
    P(E^* = 1 \mid Z = 1, X = x) &= P(E^* = 1, S = 0 \mid Z = 1, X = x) + P(E^* = 1, S = 1 \mid Z = 1, X = x) \\
    &= P(E^* = 1 \mid S = 0, Z = 1, X = x) P(S = 0 \mid Z = 1, X = x) \\
    &\hspace{2em} + P(E^* = 1 \mid S = 1, Z = 1, X = x) P(S = 1 \mid Z = 1, X = x) \\ 
    &= P(S = 1 \mid Z = 1, X = x)  \ , 
\end{align*}
where the last line follows because Assumptions \ref{assmp:doom_exp_suff} and \ref{assmp:doom_exp_nec_in_vax} imply that $P(E^* = 1 \mid S = 0, Z = 1, X = x) = 0$ and $P(E^* = 1 \mid S = 1, Z = 1, X = x) = 1$. The former claim can be shown as follows: \begin{align*}
    P(E^* = 1 \mid S = 0, Z = 1, X = x) &= \frac{P(S = 0 \mid E^* = 1, Z = 1, X = x) P(E^* = 1 \mid Z = 1, X = x)}{P(S = 0 \mid Z = 1, X = x)} \\
    &= \frac{0 \times P(E^* = 1 \mid Z = 1, X = x)}{P(S = 0 \mid Z = 1, X = x)} \\ 
    &= 0 \ . 
\end{align*}

The latter can be shown as follows: \begin{align*}
    &P(E^* = 1 \mid S = 1, Z = 1, X = x) \\
    &\hspace{2em} = \frac{P(S = 1 \mid E^* = 1, Z = 1, X = x) P(E^* = 1 \mid Z = 1, X = x)}{P(S = 1 \mid Z = 1, X = x)} \\
    &\hspace{2em} = \frac{1 \times P(E^* = 1 \mid Z = 1, X = x)}{\left\{ P(S = 1 \mid E^* = 1, Z = 1, X = x) P(E^* = 1 \mid Z = 1, X = x) \right. } \\
    &\hspace{6em} \left. + P(S = 1 \mid E^* = 0, Z = 1, X = x) P(E^* = 0 \mid Z = 1, X = x) \right\} \\ 
    &\hspace{2em} = \frac{1 \times P(E^* = 1 \mid Z = 1, X = x)}{\left\{ 1 \times P(E^* = 1 \mid Z = 1, X = x) + 0 \times P(E^* = 0 \mid Z = 1, X = x) \right\}} \\ 
    &\hspace{2em}= \frac{P(E^* = 1 \mid Z = 1, X = x)}{P(S = 1 \mid E^* = 1, Z = 1, X = x) P(E^* = 1 \mid Z = 1, X = x)} \\
    &\hspace{2em}= \frac{P(E^* = 1 \mid Z = 1, X = x)}{P(E^* = 1 \mid Z = 1, X = x)} \\
    &\hspace{2em}= 1 \ . 
\end{align*}

Thus, returning to equation (\ref{eqn:ref_me_proof_doomed_exp}), we have shown that \[
P(X = x \mid E^* = 1) = \frac{P (S = 1 \mid Z = 1, X = x) P(X = x)}{P(E^* = 1 \mid Z = 1)} \ .
\]
Since we above established that $P(E^* = 1 \mid Z = 1, X = x) = P(S = 1 \mid Z = 1, X = x)$ for all $x$, then under Assumption \ref{assmp:doom_vax_not_cause}, we also have that $P(E^* = 1) = \E\{ P(S = 1 \mid Z = 1, X )\} = \bar{\rho}_1$, and thus \[
P(X = x \mid E^* = 1) = \frac{\rho_1(x)}{\bar{\rho}_1} \ .
\]

Thus, it remains to identify $\E[Y(1) \mid E^* = 1, X = x]$ and $\E[Y(0) \mid E = 1, X = x]$.

To identify $\E\{ Y(1) \mid E^* = 1 \}$, we note that  \begin{align*}
    \E\{ Y(1) \mid E^* = 1, X = x \} &= \E \{ Y(1) \mid Z = 1, E^* = 1, X = x \} \\
    &= \E( Y \mid Z = 1, E^* = 1, X = x ) \\
    &= \E( Y \mid Z = 1, E^* = 1, X = x, S = 0 ) P(S = 0 \mid Z = 1, E^* = 1, X = x) \\
    &\hspace{2em} + \E( Y \mid Z = 1, E^* = 1, X = x, S = 1) P(S = 1 \mid Z = 1, E^* = 1, X = x) \\
    &= \E( Y  \mid Z = 1, E^* = 1, X = x, S = 1 ) \ . 
\end{align*}
These equalities follow respectively from randomization and Assumption \ref{assmp:doom_vax_not_cause}, consistency, law of total expectation, Assumptions \ref{assmp:doom_exp_suff} and \ref{assmp:doom_exp_nec_in_vax}. 

Now, we have that \begin{align*}
    \E( Y  \mid Z = 1, X = x, S = 1 ) &= \E( Y  \mid Z = 1, X = x, S = 1, E^* = 1 )P(E^* = 1 \mid Z = 1, X = x, S = 1) \\
    &\hspace{2em} + \E( Y  \mid Z = 1, X = x, S = 1, E^* = 0 )P(E^* = 0 \mid Z = 1, X = x, S = 1) \\
    &= \E( Y  \mid Z = 1, X = x, S = 1, E^* = 1 ) \ , 
\end{align*}
as we established above that $P(E^* = 1 \mid S = 0, Z = 1, X = x) = 0$ and $P(E^* = 1 \mid S = 1, Z = 1, X = x) = 1$. Thus, we have shown that $\E\{ Y(1) \mid E^* = 1, X = x \} = \E( Y  \mid Z = 1, X = x, S = 1 ) = \mu_{11}(x)$.

Now, we consider identification of $\E[Y(0) \mid E^* = 1, X = x]$. We note that \begin{align*}
    \E\{ Y(0) \mid E^* = 1, X = x \}
    &= \E\{ Y(0) \mid Z = 0, E^* = 1, X = x \}  \\
    &= \E( Y \mid Z = 0, E^* = 1, X = x ) \\
    &= \sum_{s = 0}^1 \E( Y \mid Z = 0, E^* = 1, X = x, S = s ) P(S = s \mid Z = 0, E^* = 1, X = x) \\
    &= \E( Y \mid Z = 0, E^* = 1, X = x, S = 1 ) \\
    &= \E( Y \mid Z = 0, X = x, S = 1 )
\end{align*}
Here, the equalities follow respectively from randomization and Assumption \ref{assmp:doom_vax_not_cause}, consistency, the law of total expectation, Assumption \ref{assmp:doom_exp_suff}, and Assumption \ref{assmp:doom_no_unmeasured_confounders}.

\end{proof}

{\color{changed}
\section{Sensitivity analysis for monotonicity assumption}

While the monotonicity assumption is likely to hold for many vaccines that have reached large scale randomized trials, there may be rare exceptions where it does not. If the monotonicity assumption is scientifically uncertain, then one may argue that the naturally infected effect is a less relevant estimand, given that it does not capture the total vaccine effect on the post-infection endpoint. Nevertheless, we provide here some results indicating how a sensitivity analysis could proceed in this case.

The key challenge when monotonicity fails is that there are now four basic principal strata -- the Harmed are individuals in whom $S(0) = 0, S(1) = 1$. In this case, the distribution of principal strata is no longer identified (Figure \ref{fig:bounds_nomonotonicity}). Thus, our sensitivity analyses proceed under the assumption that the proportion of individuals Harmed by the vaccine is known. The analysis then varies this assumption to assess sensitivity of the results. 

\begin{assump}
    \textbf{Known proportion of harmed individuals}. Assume that $P\{ S(0) = 0, S(1) = 1 \mid X = x \} = r$ for all $x$.
    \label{assmp:known_prop_harmed}
\end{assump}

Assumption \ref{assmp:known_prop_harmed} could also be allowed to vary with $x$, i.e.,  $P\{ S(0) = 0, S(1) = 1 \mid X = x \} = r(x)$, if domain knowledge supports this. Here, we opt for the simpler constant probability. Assumption \ref{assmp:known_prop_harmed} allows us to identify principal strata as follows: 
\begin{equation}
    \begin{aligned}
    P\{ S(0) = 0, S(1) = 1 \} &= p \ \ \text{(by assumption)} \\
    P\{ S(0) = 1, S(1) = 1 \} &= \bar{\rho}_1 - p \\
    P\{ S(0) = 1, S(1) = 0 \} &= \bar{\rho}_0 - \bar{\rho}_1 + p \\ 
    P\{ S(0) = 0, S(1) = 0 \} &= 1 - \bar{\rho}_0 - p \ .
\end{aligned}
\label{eqn:ps_id_no_mono}
\end{equation}
Based on these identifications, we also know that $p$ can only take values in \begin{equation}
\text{max}(0, \bar{\rho}_1 - \bar{\rho}_0) \le p \le \text{min}(\bar{\rho}_1, 1 - \bar{\rho}_0) \ . \label{eqn:bounds_on_harmed}
\end{equation}

\begin{figure}
\begin{center}
\begin{tikzpicture}
\node[] at (0, 0) (label1) {$Z = 1$};
\node[draw, minimum width=2.8cm, minimum height=1.5cm, align=center, below = 0.05cm of label1] (vaxinf) {$S = 1$ \\[0.5em] Doomed or \\ Harmed};
\node[draw, minimum width=2.8cm, minimum height=3cm, below=0cm of vaxinf, align = center] (vaxuninf) {\\[0.5em] $S = 0$ \\[0.5em] Immune or \\ Protected};

\node[right=2.5cm of label1] (label2) {$Z = 0$};
\node[draw, minimum width=2.8cm, minimum height=2.3cm, align=center, below = 0.05cm of label2] (plcinf) {$S = 1$ \\[0.5em] Doomed or \\ Protected};
\node[draw, minimum width=2.8cm, minimum height=2.2cm, below=0cm of plcinf, align = center] (plcuninf) {$S = 0$ \\[0.5em] Immune or \\ Harmed};
\end{tikzpicture}
\caption{
The observed vaccinated (left) and placebo (right) groups can be divided (solid lines) based on observed infection status ($S = 1$, top = infected, $S = 0$, bottom = uninfected). Without monotonicity, all observed strata are mixtures of at least two basic principal strata.
}
\end{center}
\label{fig:bounds_nomonotonicity}
\end{figure}

\subsection{Sensitivity to monotonicity under exclusion restriction}

We have that 
\begin{equation}
\begin{aligned}
    &\E\{ Y(1) - Y(0) \} \\
    &\hspace{2em} = \E\{ Y(1) - Y(0) \mid S(0) = 1\} P\{ S(0) = 1 \} + \E\{ Y(1) - Y(0) \mid S(0) = 0 \} P\{ S(0) = 0 \} \ .
\end{aligned}
\label{eqn:start_of_id_non_mono}
\end{equation}

We still have that $\E\{ Y(0) \mid S(0) = 1 \}$ is identified by $\bar{\mu}_{00}$, $\E\{Y(1)\}$ is identification by $\bar{\mu}_{1,\cdot}$ and $\E\{Y(0)\}$ is identified by $\bar{\mu}_{0,\cdot}$. However, identification of other terms in this expression is now more complicated. Under monotonicity and our exclusion restriction (Assumption 6), we previously had that $\E\{ Y(1) - Y(0) \mid S(0) = 0 \} = 0$, as the $S(0) = 0$ subgroup consisted only of Immune individuals. However, lacking monotonicity, there are now Harmed individuals in this strata and thus the vaccine may have an effect here. Thus, we must continue, writing \begin{align*}
    &\E\{ Y(1) - Y(0) \mid S(0) = 0 \} \\
    &\hspace{2em} = \E\{ Y(1) - Y(0) \mid S(0) = 0, S(1) = 0 \} P\{ S(1) = 0 \mid S(0) = 0 \} \\
    &\hspace{4em}+ \E\{ Y(1) - Y(0) \mid S(0) = 0, S(1) = 1 \} P\{ S(1) = 1 \mid S(0) = 0 \} \\
    &\hspace{2em}= \E\{ Y(1) - Y(0) \mid S(0) = 0, S(1) = 1 \} P\{ S(1) = 1 \mid S(0) = 0 \} \ , 
\end{align*}
where the final equality follows from the exclusion restriction. From equation (\ref{eqn:ps_id_no_mono}), we can compute that \[
P\{ S(1) = 1 \mid S(0) = 0 \} = \frac{p}{1 - \bar{\rho}_0} \ .
\]
Thus, it remains to identify $\E\{ Y(1) - Y(0) \mid S(0) = 0, S(1) = 1\}$. This quantity could be point identified under a partial principal ignorability assumption ($Y(1) \perp S(0) \mid S(1) = 1, X$); however, this would stymie our efforts to present an analysis that serves as a sensitivity analysis to only the monotonicity assumption while assuming an exclusion restriction. Thus, we instead proceed as in Assumption \ref{assmp:partial_pi_sens_assump} and assume proportionality between means of $Y(1)$ comparing the Harmed to the Doomed and of $Y(0)$ comparing the Harmed to the Immune. 

\begin{assump}
    For $(\epsilon_1, \epsilon_2) \in \mathbb{R}^+ \times \mathbb{R}^+$, \begin{align*}
        \frac{\E\{Y(1) \mid S(0) = 1, S(1) = 1\}}{\E\{Y(1) \mid S(0) = 0, S(1) = 1\}} &= \epsilon_1 \ \ \text{and} \ \ 
        \frac{\E\{Y(0) \mid S(0) = 0, S(1) = 0\}}{\E\{Y(0) \mid S(0) = 0, S(1) = 1\}} = \epsilon_2 \ . \\
    \end{align*}
    \label{assump:prop_sens_no_mono}
\end{assump}

Under Assumption \ref{assump:prop_sens_no_mono}, we have \begin{align*}
    \E\{ Y(1) \mid S(0) = 0, S(1) = 1 \} &= \frac{\E\{ Y(1) \mid S(1) = 1\}}{\epsilon_1 P\{S(0) = 1 \mid S(1) = 1 \} + P\{ S(0) = 0 \mid S(1) = 1 \} } \  \ \text{and}  \\
    \E\{ Y(0) \mid S(0) = 0, S(1) = 1 \} &= \frac{\E\{ Y(0) \mid S(0) = 0\}}{\epsilon_2 P\{S(1) = 0 \mid S(0) = 0 \} + P\{ S(1) = 1 \mid S(0) = 0 \} } \ , 
\end{align*}
which we can identify as \begin{align*}
    \E\{ Y(1) \mid S(0) = 0, S(1) = 1 \} &= \frac{\bar{\mu}_{11}}{\epsilon_1 \frac{(\bar{\rho}_1 - p)}{\bar{\rho}_1} + \frac{p}{\bar{\rho}_1} } \ , \ \text{and}  \\
    \E\{ Y(0) \mid S(0) = 0, S(1) = 1 \} &= \frac{\bar{\mu}_{00}}{\epsilon_2 \frac{1 - \bar{\rho}_0 - p}{1 - \bar{\rho}_0} + \frac{p}{1 - \bar{\rho}_0} } \ . 
\end{align*}

Thus, we have shown that \begin{align*}
  &\E\{ Y(1) - Y(0) \mid S(0) = 0 \} = \left\{ \frac{\bar{\mu}_{11}}{\epsilon_1 \frac{(\bar{\rho}_1 - p)}{\bar{\rho}_1} + \frac{p}{\bar{\rho}_1} } - \frac{\bar{\mu}_{00}}{\epsilon_2 \frac{1 - \bar{\rho}_0 - p}{1 - \bar{\rho}_0} + \frac{p}{1 - \bar{\rho}_0} } \right\} \frac{p}{1 - \bar{\rho}_0} \ . 
\end{align*}

Returning to equation (\ref{eqn:start_of_id_non_mono}), we can now solve for $\E\{ Y(1) \mid S(0) = 1 \}$, yielding the identification that \[
\E\{ Y(1) \mid S(0) = 1 \} = \frac{1}{\bar{\rho}_0} \left[ \bar{\mu}_{1\cdot} - \bar{\mu}_{0\cdot} - p \left\{ \frac{\bar{\mu}_{11}}{\epsilon_1 \frac{(\bar{\rho}_1 - p)}{\bar{\rho}_1} + \frac{p}{\bar{\rho}_1} } - \frac{\bar{\mu}_{00}}{\epsilon_2 \frac{1 - \bar{\rho}_0 - p}{1 - \bar{\rho}_0} + \frac{p}{1 - \bar{\rho}_0} } \right\} \right] + \bar{\mu}_{01} \ . 
\]

Setting $p = 0$, we find that we indeed recover the identification result from Theorem 4 in the main text.

\subsection{Sensitivity to monotonicity under partial principal ignorability}

We have that \begin{align*}
\E\{ Y(1) \mid S(0) = 1 \} = \sum_x \E\{ Y(1) \mid S(0) = 1, X = x \} P(X = x \mid S(0) = 1) \ ,
\end{align*}
and we have argued previously that under our assumptions \[
P(X = x \mid S(0) = 1) = \frac{\rho_0(x)}{\bar{\rho}_0} P(X = x) \ .
\]
We can then write that 
\begin{equation}
\begin{aligned}
    &\E\{ Y(1) \mid S(0) = 1, X = x \} \\
    &\hspace{2em} = \E\{ Y(1) \mid S(0) = 1, S(1) = 1, X = x \} P\{S(1) = 1 \mid S(0) = 1, X = x\} \\
    &\hspace{4em} + \E\{ Y(1) \mid S(0) = 1, S(1) = 0, X = x \} P\{S(1) = 0 \mid S(0) = 1, X = x\} \ .
\end{aligned}
\label{eqn:prove_id_no_mono_pi}
\end{equation}
The quantity $\E\{ Y(1) \mid S(0) = 1, S(1) = 0, X = x \}$ is identified by $\mu_{10}(x)$ under Assumption 7 and the principal strata distribution is identified under Assumption \ref{assmp:known_prop_harmed}: \begin{align*}
P\{S(1) = 1 \mid S(0) = 1, X = x\} &= \frac{\rho_1(x) - p}{\rho_0(x)} \ , \\
P\{S(1) = 0 \mid S(0) = 1, X = x\} &= \frac{\rho_0(x) - \rho_1(x) + p}{\rho_0(x)} \ .
\end{align*}

As above, we do not wish to make an additional partial principal ignorability assumption here, so we make a further assumption to identify the mean of $Y(1)$ in the Doomed strata.

\begin{assump}
For all $x$ and for $\epsilon_1 \in \mathbb{R}^+$, \begin{align*}
\frac{\E\{Y(1) \mid S(0) = 1, S(1) = 1, X = x\}}{\E\{Y(1) \mid S(0) = 0, S(1) = 1, X = x\}} &= \epsilon_1 \ .
\end{align*}
\label{assump:prop_sens_no_mono_pi}
\end{assump}

Note that this is similar to the first component of Assumption \ref{assump:prop_sens_no_mono}; however, here we make the constancy assumption conditional on $X = x$. Under Assumption \ref{assump:prop_sens_no_mono_pi}, we can then show that \[
\E\{ Y(1) \mid S(0) = 1, S(1) = 1, X = x \} = \frac{\mu_{11}(x)}{\frac{\rho_1(x) - p}{\rho_1(x)} + \frac{p}{\epsilon_1\rho_1(x)}} \ .
\]

Returning to equation (\ref{eqn:prove_id_no_mono_pi}), we have thus argued that \[
\E\{ Y(1) \mid S(0) = 1, X = x \} = \frac{\mu_{11}(x)}{\frac{\rho_1(x) - p}{\rho_1(x)} + \frac{p}{\epsilon_1\rho_1(x)}} \left\{ \frac{\rho_1(x) - p}{\rho_0(x)}\right\} + \mu_{10}(x) \left\{\frac{\rho_0(x) - \rho_1(x) + p}{\rho_0(x)} \right\} \ .
\]
Thus, we have shown that \[
\E\{ Y(1) \mid S(0) = 1 \} = \E\left(
\frac{\rho_0(X)}{\bar{\rho}_0} \left[ 
\frac{\mu_{11}(X)}{\frac{\rho_1(X) - p}{\rho_1(X)} + \frac{p}{\epsilon_1\rho_1(X)}} \left\{ \frac{\rho_1(X) - p}{\rho_0(X)}\right\} + \mu_{10}(X) \left\{\frac{\rho_0(X) - \rho_1(X) + p}{\rho_0(X)} \right\} 
\right]
\right)
\]

It is straightforward to confirm that if $p = 0$, then this identification result reduces to the identification under monotonicity.

\subsection{Sensitivity to monotonicity under exclusion restriction and partial principal ignorability}

As previously described, when both Assumption 6 and 7 hold, the conditional mean in the Protected is identified by $\mu_{\cdot0}$ and thus identical logic as in the previous subsection shows that under Assumption \ref{assump:prop_sens_no_mono_pi}, \[
\E\{ Y(1) \mid S(0) = 1 \} = \E\left(
\frac{\rho_0(X)}{\bar{\rho}_0} \left[ 
\frac{\mu_{11}(X)}{\frac{\rho_1(X) - p}{\rho_1(X)} + \frac{p}{\epsilon_1\rho_1(X)}} \left\{ \frac{\rho_1(X) - p}{\rho_0(X)}\right\} + \mu_{\cdot0}(X) \left\{\frac{\rho_0(X) - \rho_1(X) + p}{\rho_0(X)} \right\} 
\right]
\right) \ . 
\]

We conducted a sensitivity analysis in the PROVIDE data under both exclusion restriction and partial principal ignorability assumptions. Using Assumption S7, we identify the minimum proportion harmed $p$ to be 0 and the maximum 0.047. {\color{changed2} These were selected as values of bounds such that a covariate-conditional version of (\ref{eqn:bounds_on_harmed}) held for all observed $x$ based on the estimated infection risks. That is, we found that $\rho_{1,n}(X_i) - \rho_{0,n}(X_i) < 0$ for $i = 1, \dots, n$ and that $\mbox{min}_{i = 1,\dots, n} [\mbox{min} \{ \rho_{1,n}(X_i), 1 - \rho_{0,n}(X_i)\} ]  = 0.047$.}

As with the sensitivity analysis to the partial principal ignorability assumption in Supplement J.2, we let $\epsilon$ range from 0.55 to 2.2, {\color{changed2} the values of $\epsilon$ that yielded estimates of Naturally Infected effects corresponding to the upper and lower bound on the effects}. We allowed both $p$ and $\epsilon$ to vary within these ranges, and estimated the additive effect for each combination of values. 

Recall that $\epsilon$ describes the counterfactual risk ratio for antibiotic prescribing under vaccine comparing the Doomed stratum to the Harmed stratum. Thus, values of $\epsilon > 1$ indicate that Doomed members are more likely to be prescribed antibiotics than Harmed individuals and values $\epsilon < 1$ indicate that Harmed members are more likely to be prescribed antibiotics than Doomed. 

Figure~\ref{fig:monotonicity_sensitivity} displays the estimated additive effect across the joint range of the two sensitivity parameters. For values of $\epsilon$ up to around 1.6, we found that increasing the proportion harmed tended to increase the size of the estimated effect. This follows intuition given the scientific context -- infected vaccinated children are more likely to receive antibiotics. Relative to the analysis that assumes no Harmed children, some children who were infected under vaccination and previously assumed to belong to the naturally infected are reclassified as Harmed and thus are removed from the effect in the naturally infected thereby increasing the size of the effect. It is only in situations where we assume that antibiotic prescribing is dramatically more common in Doomed versus Harmed individuals do we see the opposite impact -- that more Harmed individuals leads to smaller effect sizes in the naturally infected. This is again intuitive as jointly $p$ and $\epsilon$ are controlling the number of children who are reclassified as Harmed and how likely antibiotic prescribing is for those children who remain in the Doomed (and thus Naturally infected) strata.

\begin{figure}
    \centering
    \includegraphics[width=0.8\linewidth]{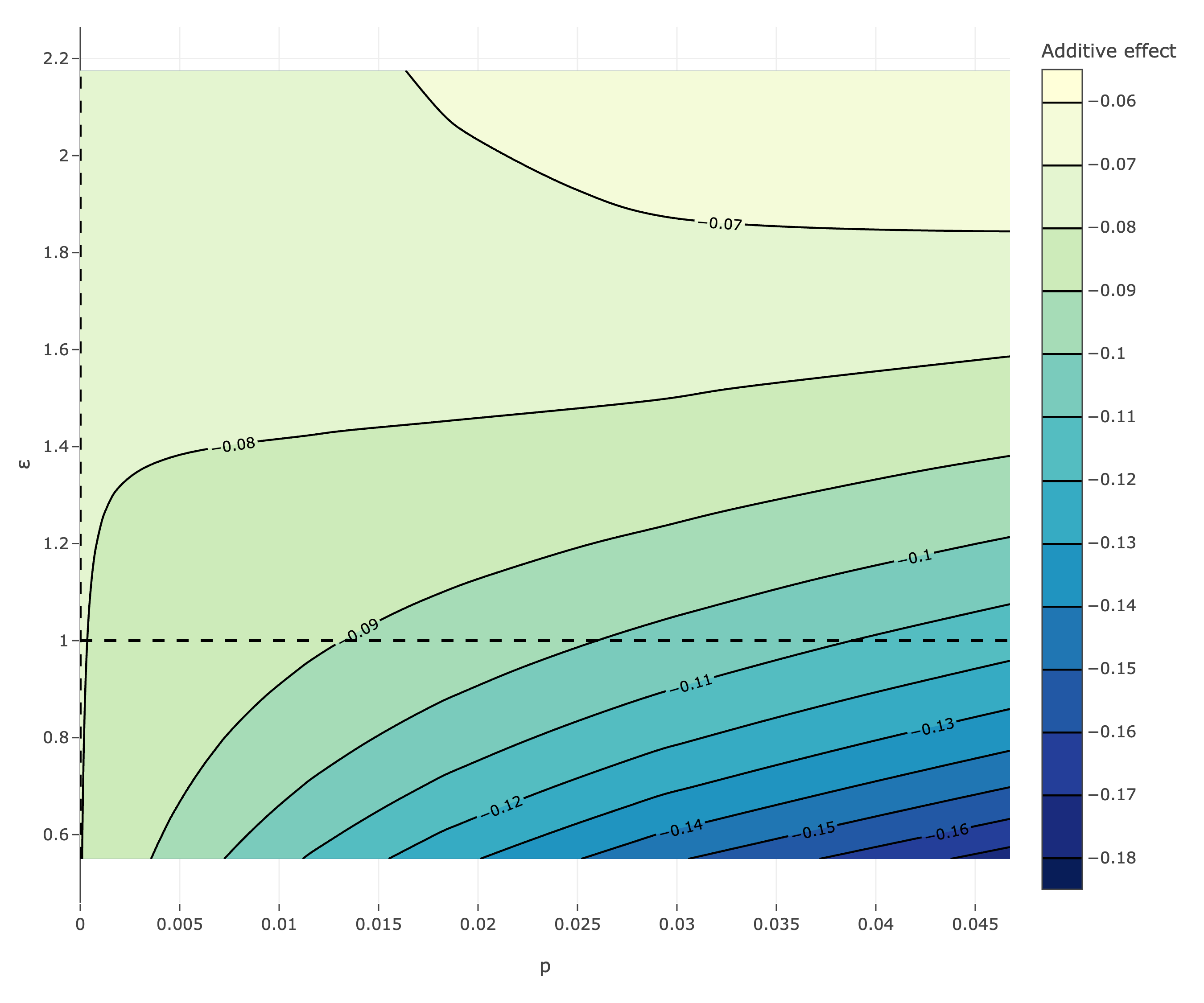}
    \caption{Monotonicity sensitivity analysis for PROVIDE data. $p$ denotes the proportion assumed to be in the Harmed stratum; $\epsilon$ is the sensitivity parameter used to control violations of partial principal ignorabiliy.}
    \label{fig:monotonicity_sensitivity}
\end{figure}

}

\bibliographystyle{apalike}
\bibliography{sample_short}

\end{document}